\documentclass[12pt,draftclsnofoot,onecolumn]{IEEEtran}
\usepackage[utf8]{inputenc}
\usepackage{graphicx, flushend}
\usepackage{amsmath}
\usepackage{cases}
\usepackage{bbm}
\usepackage{amssymb}
\usepackage{array}
\usepackage{array}
\usepackage{lipsum,color}
 \usepackage{epstopdf}

\DeclareGraphicsExtensions{.pdf,.png,.jpg,.eps,.ps}
\newcommand{\argmin}{\arg\!\min}
\newcommand{\e}{\text{e}}
\usepackage{algorithm}
\usepackage{algorithmicx}
\usepackage{algpseudocode}
\usepackage[square, numbers, comma, sort&compress]{natbib}
\usepackage{xcolor}
\usepackage{hyperref}

\newcommand{\minus}{\scalebox{0.75}[0.75]{$-$}}
\newcommand{\plus}{\scalebox{0.75}[0.75]{$+$}}

\newcommand{\tinyDelta}{\scalebox{0.75}[0.55]{$\Delta$}}

\usepackage{url}
\usepackage{amsthm,amssymb}

\usepackage{algorithm,algpseudocode}
\makeatletter
\newcommand{\algmargin}{\the\ALG@thistlm}
\makeatother
\newlength{\whilewidth}
\algdef{SE}[parWHILE]{parWhile}{EndparWhile}[1]
  {\parbox[t]{\dimexpr\linewidth-\algmargin}{%
     \hangindent\whilewidth\strut\algorithmicwhile\ #1\ \algorithmicdo\strut}}{\algorithmicend\ \algorithmicwhile}%
\algnewcommand{\parState}[1]{\State%
  \parbox[t]{\dimexpr\linewidth-\algmargin}{\strut #1\strut}}

\newtheorem{lemma}{Lemma}
\newtheorem{theorem}{Theorem}
\newtheorem{remark}{Remark}

\newtheorem{proposition}{Proposition}

\newcommand{\bseq}{\begin{subequations}}
\newcommand{\eseq}{\end{subequations}}
\newcommand{\baln}{\begin{align}}
\newcommand{\ealn}{\end{align}}
\newcommand{\balnd}{\begin{aligned}}
\newcommand{\ealnd}{\end{aligned}}
\newcommand{\beq}{\begin{equation}}
\newcommand{\eeq}{\end{equation}}
\newcommand{\beqn}{\begin{eqnarray}}
\newcommand{\eeqn}{\end{eqnarray}}
\newcommand{\beqno}{\begin{eqnarray*}}
\newcommand{\eeqno}{\end{eqnarray*}}
\newcommand{\bma}{\begin{displaymath}}
\newcommand{\ema}{\end{displaymath}}
\newcommand{\bnu}{\begin{enumerate}}
\newcommand{\enu}{\end{enumerate}}
\newcommand{\bce}{\begin{center}}
\newcommand{\ece}{\end{center}}
\newcommand{\btb}{\begin{tabular}}
\newcommand{\etb}{\end{tabular}}
\newcommand{\ba}{\begin{array}}
\newcommand{\ea}{\end{array}}
\newcommand{\mb}{\mathbf}
\renewcommand{\baselinestretch}{1.38}

\usepackage{textcomp}
\usepackage[font=small,skip=4pt]{caption}
\begin{document}
\title{Joint Data Compression and Computation Offloading  in Hierarchical Fog-Cloud Systems}

\author{\IEEEauthorblockN{Ti Ti Nguyen, {\em Student Member, IEEE}, Vu Nguyen Ha, {\em Member, IEEE},\\
 Long Bao Le, {\em Senior Member, IEEE}, and Robert Schober, {\em Fellow, IEEE}}
	
	\thanks{Ti Ti Nguyen and Long B. Le are with INRS-EMT, University of Qu\'{e}bec,  Montr\'{e}al, Qu\'{e}bec, Canada (emails: \{titi.nguyen,long.le\}@emt.inrs.ca).}
	\thanks{Vu N. Ha is with \'{E}cole Polytechnique de Montr\'{e}al,  Montr\'{e}al, Qu\'{e}bec, Canada (email: vu.ha-nguyen@polymtl.ca).}
	\thanks{Robert Schober  is with Friedrich-Alexander-University Erlangen-Nuremberg, Germany  (email: robert.schober@fau.de).}
}
 
\markboth{IEEE TRANSACTION ON WIRELESS COMMUNICATIONS, SUBMITTED}{NGUYEN ET AL ... }

\maketitle

\begin{abstract}

Data compression (DC) has the potential to significantly improve the computation offloading performance in hierarchical fog-cloud systems. However, it remains unknown how to optimally determine the compression ratio jointly with the computation offloading decisions and the resource allocation. This joint optimization problem is studied in the current paper where we aim to minimize the maximum weighted energy and service delay cost (WEDC) of all users.
First, we consider a scenario where DC is performed only at the mobile users. We prove that the optimal offloading decisions have a threshold structure. Moreover, a novel three-step approach employing convexification techniques is developed to optimize the compression ratios and the resource allocation. Then, we address the more general design where DC is performed at both the mobile users and the fog server. We propose three efficient algorithms to overcome the strong coupling between the offloading decisions and the resource allocation.
We show that the proposed optimal algorithm for DC at only the mobile users can reduce the WEDC by up to 65\% compared to computation offloading strategies that do not leverage DC  or use sub-optimal optimization approaches. Besides, the proposed algorithms for additional DC at the fog server can further reduce the WEDC.

\end{abstract}

\IEEEpeerreviewmaketitle
\section{Introduction}
Currently, mobile edge/cloud computing (MEC/MCC) technologies are 
considered as promising solutions for enhancing the mobile usability and prolonging the mobile battery life by offloading computation heavy applications to a remote fog/cloud server \cite{ren2017latency,ti2017,nguyen2017joint}. In an MCC system, enormous computing resources are available in the core network, but the limited backhaul capacity can induce significant delay for the underlying applications. In contrast, an MEC system, with computing resources deployed at the network edge in close proximity to mobile devices, can enable computation offloading and meet demanding application requirements \cite{lyu2017}.

Hierarchical fog-cloud computing systems which leverage the advantages of both MCC and MEC can further
 enhance the system performance  \cite{jalali2016fog, chiang2016fog, deng2016optimal,shah2018hierarchical,du2017computation}
where fog servers deployed at the network edge can operate collaboratively with the more powerful cloud servers to 
execute computation-intensive user applications. 
Specifically, when the users' applications require high computing power or low latency, their computation tasks can be offloaded and processed at the fog and/or remote cloud servers. 
However, the upsurge of mobile data and the constrained radio spectrum may result in significant delays in transferring
offloaded data between the mobile users and the fog/cloud servers, which ultimately degrades the quality of service (QoS) \cite{liu2018cooperative}. 
To overcome this challenge, advanced DC techniques can be leveraged to reduce the amount of incurred data (i.e., the input data of a user's application)
\cite{deepu2017hybrid, alsheikh2016rate}.
However, employment of DC entails additional computations 
for the execution of the corresponding DC and decompression algorithms \cite{zhang2017mobile}. 
Therefore, an efficient joint design of DC, offloading decisions, and resource allocation is needed
to take full advantage of DC while meeting QoS requirements and other system constraints. 

\subsection{Related Works}

Computation offloading design for MCC/MCE systems has been studied extensively in the literature, see recent surveys \cite{yuyi2017, mouradian2017comprehensive} and the references therein. Most existing works consider two main
performance metrics for their designs, namely energy-efficiency \cite{YouHCK16, wang2018joint, zhao2017energy, wang2017multi} and delay-efficiency
\cite{ning2018cooperative,  gu2018joint, bi2018computation, ren2019collaborative}.  
Focusing on energy-efficiency, the authors of \cite{YouHCK16} develop partial offloading frameworks for  multiuser MEC systems employing
time division multiple access and frequency-division multiple access. 
In \cite{wang2018joint}, wireless power transfer is integrated into the computation offloading design. Moreover, different binary 
offloading frameworks are developed in \cite{zhao2017energy, wang2017multi} where various 
 branch-and-bound and heuristic algorithms are proposed to tackle the resulting mixed integer optimization
problems.

Considering computation offloading from the delay-efficiency point of view, an iterative heuristic algorithm 
to optimize the binary offloading decisions for minimization of the overall computation and transmission delay in a hierarchical fog-cloud system is proposed in \cite{ning2018cooperative}. 
The authors in \cite{gu2018joint} formulate the computation offloading and
 resource allocation problem as a student-project-allocation game with the objective to maximize the ratio between the average offloaded data rate and the offloading cost at the users. In \cite{bi2018computation}, the authors study a binary computation offloading problem 
for maximization of the weighted sum computation rate. Then, they propose a coordinate descent based algorithm in which  the offloading decision and time-sharing variables are iteratively updated until convergence. Considering the partial computation offloading problem, the authors in  \cite{ren2019collaborative} propose a framework to minimize weighted-sum latency of all mobile users via the collaboration between cloud computing and fog computing assuming the TDMA based resource sharing strategy \cite{ren2019collaborative}.

Some recently proposed schemes for computation offloading consider both energy and delay efficiency aspects  \cite{deng2016optimal, liu2018multiobjective, du2017computation}.
 In particular, 
the work in \cite{deng2016optimal} proposes a radio and computing resource allocation framework where 
the computational loads of fog and cloud servers are determined and the trade-off between power consumption and service delay is investigated. 
Additionally, the authors of \cite{liu2018multiobjective} jointly optimize the transmit power and offloading probability for minimization of the average weighted energy, delay, and payment cost.  
In \cite{du2017computation}, the authors study fair computation offloading design 
 minimizing the maximum WEDC of all users in a hierarchical fog-cloud system. 
In this work, a two-stage algorithm is proposed where the offloading decisions are determined in the first stage using a semidefinite relaxation and probability rounding based method while the radio and computing resource allocation is determined in the second stage. However, references \cite{YouHCK16, wang2018joint, zhao2017energy, wang2017multi, ning2018cooperative,  gu2018joint, bi2018computation,deng2016optimal, liu2018multiobjective, du2017computation} have not exploited DC for computation offloading.

There are few existing works that explore DC for computation offloading.
Specifically, the authors of \cite{liu2018cooperative} propose an analytical framework
to evaluate the outage performance of a hierarchical fog-cloud system. 
Moreover, the work in \cite{zhang2017mobile} considers DC for computation offloading design
for systems with a single server but assumes a fixed compression ratio (i.e., this parameter is not optimized).
In general, the compression ratio should be jointly optimized with the computation offloading decisions and resource allocation
to achieve optimal system performance.
However, the computational load incurred by compression/decompression is a non-linear function of the
compression ratio, which makes this joint optimization problem very challenging.


\subsection{Contributions and Organization of the Paper}

To the best of our knowledge, the joint design of DC, computation offloading, and resource 
allocation for hierarchical fog-cloud systems has not been considered  in the existing literature.
The main contributions of this paper can be summarized as follows:
\begin{itemize}
	\item We propose a non-linear computation model which can be fitted to accurately capture the computational load incurred by DC and decompression. In 
	particular, the compression and decompression computational load as well as the quality of data recovery are modeled as functions of the compression ratio. 
	
	\item For DC at only the mobile users, we  formulate the fair joint design of the compression ratio, computation offloading, and resource allocation as a mixed-integer non-linear programming (MINLP) optimization problem. This problem formulation takes into account practical constraints on the maximum transmit power, wireless access bandwidth, backhaul capacity, and computing resources. 
	We propose an optimal algorithm, referred to as Joint DC, Computation offloading, and Resource 
	Allocation (JCORA) algorithm, which solves this challenging 
	problem optimally. To develop this algorithm, we first prove that users incurring higher WEDC when executing their application locally 
	should have higher priority for offloading. Based on this result, the bisection search method is employed to optimally classify users 
	into two user sets, namely the set of offloading users, and the set of remaining users, and JCORA globally optimizes the
	decision optimization variables for both user sets.
	
	\item We then study a more general design where DC is performed at both the mobile users and the fog server (with different compression
	ratios) before the compressed data are transmitted over the wireless link and the backhaul link connecting the fog server and cloud server, respectively. This enhanced design can lead to a significant performance gain when both wireless access and backhaul networks are congested.	
	Three different solution approaches are proposed to solve this more general problem. In the first approach, we extend the design principle of the JCORA algorithm by employing the piece-wise linear approximation (PLA) method to tackle the coupling of optimization variables. In the remaining approaches, we utilize the Lagrangian method and solve the dual optimization problem. Specifically, in the second approach, referred to as One-dimensional $\lambda$-Search based Two-Stage (OSTS) algorithm, a one-dimensional search is employed to determine the optimal value of the Lagrangian multiplier, while in the third approach, referred to as Iterative $\lambda$-Update based Two-Stage (IUTS) algorithm, a low-complexity iterative sub-gradient projection technique is adopted to tackle the problem.
	
	\item Extensive numerical results are presented to evaluate the performance gains of the proposed designs in 
	comparison with conventional strategies that do not employ DC. Moreover, our results confirm the excellent performance achievable by 
	joint optimization of DC,  computation offloading decisions, and resource allocation in a hierarchical fog-cloud system.
 \end{itemize}

The remainder of this paper is organized as follows. \textbf{Section~\ref{st2}} presents the system model, the computation and transmission energy models, and the problem formulation. \textbf{Section~\ref{st3}} develops the proposed optimal algorithm for the case when DC is performed only at the mobile users. \textbf{Section~\ref{st4a}} provides the enhanced problem with DC also at the fog server and three 
methods for solving it. \textbf{Section~\ref{st5}} evaluates the performance of the proposed algorithms. Finally, \textbf{Section~\ref{st6}} concludes this work.

\section{System Model and Problem Formulation} \label{st2}
\subsection{System Model}

\begin{figure}[t]
	\centering
	\begin{minipage}[t]{0.9\linewidth}
		\includegraphics[width=1\textwidth]{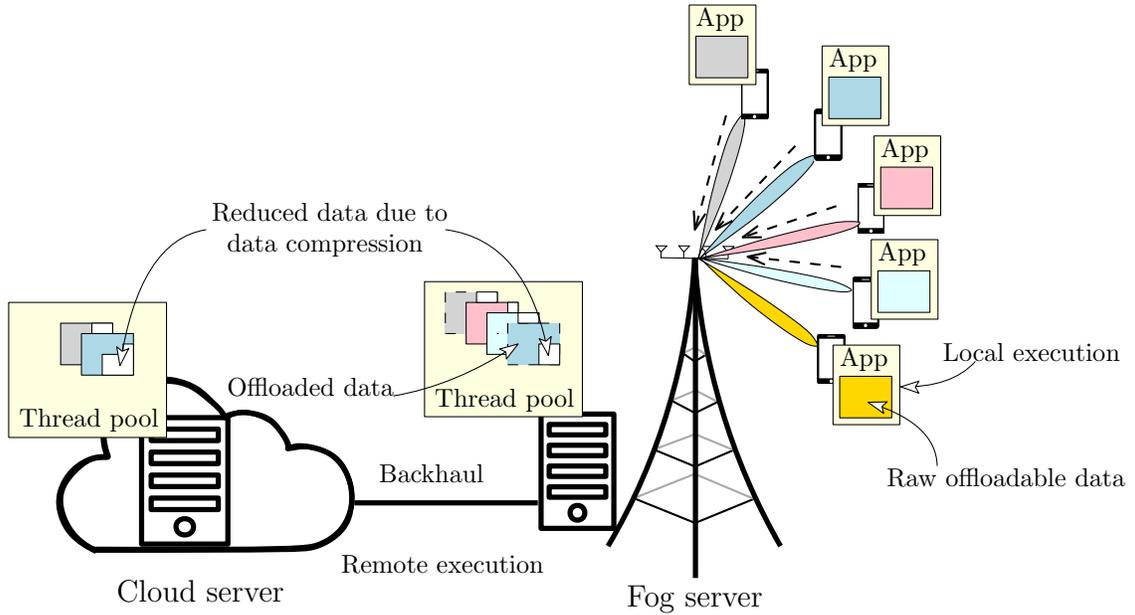}
		\caption{Data compression and computation offloading in hierarchical fog-cloud systems.}
		\label{P2-figsys}
		\vspace{-0.5 cm}
	\end{minipage}%
\end{figure}
We consider a hierarchical fog-cloud system consisting of $K$ mobile users, one cloud server, and one fog server co-located with
a base station (BS) equipped with multiple antennas. 
In this system, BS communicates to its users through wireless links while the (wired) backhaul link are deployed between the BS co-located with fog server and the cloud server\footnote{This system can capture a scenario in cellular network where the fog server is deployed at the base station to provide the computing service to the mobile users and the base station is connected to the cloud server through the backhaul transportation network.}.
For convenience, we denote the set of users as $\mathcal{K}$.
We assume that each user $k$ needs to execute an application requiring $c_{k}$ CPU cycles within an interval of $T_k^{\sf{max}}$ seconds, 
in which $c_{k,0}$ CPU cycles must be executed locally at the mobile device and the remaining offloadable $c_{k,1}$ CPU cycles can be processed locally or offloaded and processed at the fog/cloud server for energy saving and delay improvement.
A sequential processing order for the unoffloadable and offloadable computing components is assumed 
in this paper\footnote{Our design can be extended to tackle the parallel processing order.}. 
Let $b_{k}^{\sf{in}}$ be the number of bits representing the corresponding
incurred data (i.e., programming states, input text/image/video)
of the possibly-offloaded $c_{k,1}$ CPU cycles.
To overcome the wireless transmission bottleneck caused by the capacity-limited wireless links between the users 
and the BS, DC is 
employed at the users for reducing the amount of data transferred to the fog server. Fig.~\ref{P2-figsys}
illustrates the considered system. 

In particular, once $c_{k,1}$ CPU cycles are offloaded, user $k$ first compresses the corresponding $b_{k}^{\sf{in}}$ bits down 
to $b_{k}^{\sf{out,u}}$ bits before sending them to the remote fog server. 
The ratio between $b_k^{\sf{in}}$ and $b_k^{\sf{out,u}}$ for user $k$ is called the compression ratio, which is denoted as $\omega_k^{\sf{u}} = b_k^{\sf{in}}/b_k^{\sf{out,u}}$.
Depending on the available fog computing resources, the offloaded computation task can be directly processed at the fog server or be further offloaded to the cloud server. 
The amount of data containing the computation outcome sent back to the users is usually
much smaller than that incurred by offloading the task. Therefore, similar to \cite{YouHCK16, liu2018multiobjective, du2017computation}, we do not consider the downlink transmission of the 
computation results in this paper\footnote{The design in this paper can be extended to consider the downlink transmission of 
feedback data as in \cite{TiLong2019}.}.

\begin{remark}
\label{mrk:offloadable}
Running an application requires executing several unoffloadable sub-tasks that handle user interaction or access local I/O devices and cannot be executed remotely  and other offloable sub-tasks that can be executed locally or remotely based on the employed offloading strategy \cite{Kwak-2015, liu2018multiobjective}.  Practically, the workload corresponding to each sub-task of a specific application has to be pre-determined and remain unchanged 
	according to the pre-programmed source code.
	Hence, the total workload of the offloadable components is typically fixed and cannot be optimized.
	In this work, we assume a binary offloading decision for all offloadable sub-tasks of each user. 
	This corresponds to the practical scenario where all offloadable sub-tasks are strongly related so they cannot be
	executed at different places.
\end{remark}

\subsubsection{Data Compression Model}
DC can be achieved by eliminating only statistical redundancy (i.e., lossless compression) or by also removing unnecessary information (i.e., lossy compression). 
To realize it, compression and decompression algorithms must be executed at the data source and destination, respectively, which
induces additional computational load.  To the best of our knowledge, in the literature, there is no theoretical result regarding a mathematical model for the computational workload required
	for data compression process.
Hence, we employ a practical data-fitting approach to 
capture the compression computational load, decompression computational load, and compression quality as non-linear
 functions of the compression ratio.
In particular, the following model is proposed to capture the compression workload,
\begin{eqnarray}
c_k^{\sf{x,u}} &=& \gamma_{k,0}^{\sf{u}}  \big[ \gamma_{k,1}^{\sf{x,u}}(\omega_k^{\sf{u}})^{\gamma_{k,2}^{\sf{x,u}}} + \gamma_{k,3}^{\sf{x,u}} \big],  \text{ for }  \omega_k^{\sf{u}} \in [ \omega_{k,1}^{\sf{u,min}}, \omega_{k,1}^{\sf{u,max}}], \label{eqcomplex} \\
q_k^{\sf{qu,u}} &=& \gamma_{k,3}^{\sf{qu,u}}- \big[\gamma_{k,1}^{\sf{qu,u}}(\omega_k^{\sf{u}})^{\gamma_{k,2}^{\sf{qu,u}}} \big], \text{ for }  \omega_k^{\sf{u}} \in [ \omega_{k,1}^{\sf{u,min}}, \omega_{k,1}^{\sf{u,max}}], \label{eqqua}
\end{eqnarray} 
where `$\sf{x}$' $=$ `$\sf{co}$' and `$\sf{de}$' stands for compression and decompression, respectively, $[ \omega_{k,1}^{\sf{u,min}}, \omega_{k,1}^{\sf{u,max}}]$ represents the possible range of $\omega_k^{\sf{u}}$ due to the compression algorithm employed at user $k$, $c_k^{\sf{co,u}}$ and $c_k^{\sf{de,u}}$ denote the additional CPU cycles at source and destination needed for compression and decompression, respectively\footnote{Note that when the compression and decompression algorithms are executed at a fixed CPU clock speed, the computational load in CPU cycles is linearly proportional to the execution time.}; $q_k^{\sf{qu,u}}$ represents the perceived QoS (i.e, this parameter, which is only considered for lossy compression, measures the deviation between the true data and the decompressed data); $\gamma_{k,0}^{\sf{u}}$  is the maximum number of CPU cycles; 
$\gamma_{k,i}^{\sf{co/de/qu,u}}, i=1,2,3,$ are constant parameters where $\gamma_{k,1}^{\sf{co/de/qu,u}}, \gamma_{k,3}^{\sf{co/de/qu,u}} \geq 0$. 
It is worth noting that $\gamma_{k,i}^{\sf{co/de/qu,u}}, i=1,2,3,$ in our paper are selected based on the experimental data which is collected by running the compression algorithms GZIP, BZ2, and JPEG in Python 3.0\footnote{For validation, we first turned off all other applications to keep the CPU clock speed almost constant when executing the compression and decompression algorithms by using \textit{`cpupower tool'} in Linux. Then, we ran algorithms GZIP, BZ2, and JPEG in Python 3.0 via a Linux terminal using Ubuntu 18.04.1 LTS on a computer equipped with CPU chipset Intel(R) core(TM) i7-4790, and 12 GB RAM. The experimental data was obtained over 1000 realizations. This allowed us to estimate the normalized  execution time, which is proportional to the normalized computational load.}.

The accuracy of the proposed model is validated by Fig.~\ref{P2-fig1} which illustrates the relation between the normalized compression/decompression execution time and the compression ratio 
using the lossless algorithms GZIP and BZ2 for the benchmark text files \textit{``alice.txt''} and \textit{``asyoulik.txt''} from Canterbury Corpus \cite{Canter}, and the lossy algorithm `JPEG' for images \textit{``clyde-river.jpg''} and \textit{``frog.jpg''} from the Canadian Museum of Nature \cite{Nature}, obtained from simulation and fitting the proposed model. Here, the normalized execution time is the ratio of the actual execution time and the maximum execution time 
over all values of the compression ratio. 
The figure shows that the curves obtained through fitting using the proposed model match
the simulation results well. 

\begin{remark}
\label{mrk:DCmodel}
A detailed accuracy comparison between our compression computational load model and existing models is provided in \textbf{Appendix G}.
\end{remark}
%
%

\begin{figure} [t]
	\centering	
	\includegraphics[width=1\textwidth]{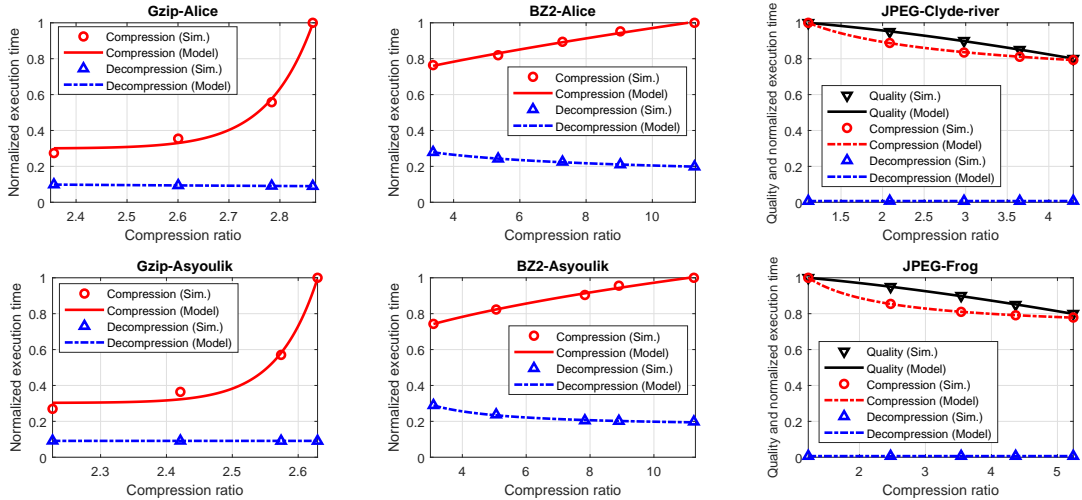}
	\caption{Compression quality and normalized execution time.}
	\label{P2-fig1}
\end{figure}

\subsubsection{Computing and Offloading Model}
We now introduce the binary offloading decision variables $s_{k}^{\sf{u}}$, $s_{k}^{\sf{f}}$, and $s_{k}^{\sf{c}}$ for the computation task of user ${k}$, where $s_{k}^{\sf{u}} = 1$, $s_{k}^{\sf{f}} = 1$, and $s_{k}^{\sf{c}} = 1$ denote the scenarios where the application is executed at the mobile device, the fog server, and the cloud server, respectively; and these variables are zero otherwise. Moreover, we assume that the $c_{k,1}$ CPU cycles can be executed at exactly one location, which implies
$
s_{k}^{\sf{u}} + s_{k}^{\sf{f}} + s_{k}^{\sf{c}} =1.
$
Then, the total computational load of user ${k}$ at the mobile device, denoted as $c_{k}^{\sf{u}}$, and at the fog server, denoted as $c_{k}^{\sf{f}} $, are given as, respectively, 
\begin{eqnarray}
c_{k}^{\sf{u}} = c_{k,0} + s_{k}^{\sf{u}} c_{k,1} + (1-s_{k}^{\sf{u}})c_{k}^{\sf{co,u}}  \text{ and }
c_{k}^{\sf{f}} = s_{k}^{\sf{f}} \big(c_{k,1} + c_{k}^{\sf{de,u}}\big).
\label{eq6}
\end{eqnarray}

As the fog and cloud servers are generally connected to the power grid while the capacity of a mobile battery is limited, we will focus on the energy consumption of the users \cite{du2017computation}. 
The local computation energy consumed by user $k$ and the local computation time can be expressed, respectively, as 
$
\xi_{1,k}^{\sf{u}} = \alpha_{k}{f_{k}^{\sf{u}}}^2 c_{k}^{\sf{u}}, \text{ and }
t_{1,k}^{\sf{u}} = {c_{k}^{\sf{u}}}/{f_{k}^{\sf{u}}}, 
$
where $f_{k}^{\sf{u}}$ is the CPU clock speed of user $k$ and $\alpha_k$ denotes the energy coefficient specified by the CPU model \cite{Zhang13}. Let $f_{k}^{\sf{f}}$ denote the CPU clock speed used at the fog server to process $c_{k,1}$. Then, the computing time at the fog server is given by
$
t_{1,k}^{\sf{f}} = {c_{k}^{\sf{f}}}/{f_{k}^{\sf{f}}}. 
$
We assume that the computation task of each user is executed at the cloud server with a fixed delay of $T^{\sf{c}}$ seconds\footnote{The delay time for the cloud server consists of two components: the execution time and the CPU set-up time. Due to the huge computing resource in the cloud server, the execution time is generally much smaller than the CPU set-up time \cite{CloudSigma}, which is identical for all users.}.

\subsubsection{Communication Model}
In order to send the incurred data during the offloading process, we 
assume that zero-forcing beamforming is applied at the BS and the average uplink rate from user $k$ to the
BS (fog server) is expressed as 
$
r_{k} = \rho_k \log_2 (1 {\plus} p_{k} \beta_{k,0}),
$
where $p_{k}$ is the uplink transmit power per Hz of user $k$, $\rho_k$ denotes the transmission bandwidth, and $\beta_{k,0} = M_0 \beta_{k}/\sigma_{\sf{bs}}$ in which $\beta_k$ represents the large-scale fading coefficient, $\sigma_{\sf bs}$ is the noise power density (watts per Hz), and $M_0$ is the MIMO beamforming gain \cite{Hien13}. It is assumed that the number of antennas is sufficiently large so that $M_0$ is identical for all users. Then, the uplink transmission time and energy of user $k$ can be computed, respectively, as 
$
t_{2,k}^{\sf{u}} = {(1-s_{k}^{\sf{u}}) b_k^{\sf{out,u}}}/r_{k}   \text{ and } 
\xi_{2,k}^{\sf{u}}  =  \rho_k (p_{k} + p_{k,0}) t_{2,k}^{\sf{u}},
$
where $p_{k,0}$ denotes the circuit power consumption per Hz.
For the data transmission between the fog server and the cloud server, a backhaul link with capacity  $D^{\sf{max}}$ bps (bits per second) is assumed. 
Let $d_{k}$ denote the backhaul rate allocated to user $k$, then the transmission time from the fog server to the 
cloud server is\footnote{Consideration of more sophisticated rate/bandwidth sharing models over the shared backhaul link is outside the scope of this paper, which is left for our future work. In fact, the fixed rate or capacity allocation
for different users sharing the backhaul link was also assumed in recent papers \cite{ren2019collaborative,  zhang2016energy}.}
$t_{2,k}^{\sf{f}} ={s_{k}^{\sf{c}} b_k^{\sf{out,u}}}/d_{k}.$

\subsection{Problem Formulation}

Assume the users have to pay for their usage of the radio and computing resources at the fog/cloud servers. Then, the service cost 
of user $k$ can be modeled as
$
\Theta_k = (1-s_{k}^{\sf{u}})(w^{\sf{BW}} \rho_{k} + w^{\sf{C}}  c_{k,1}),
$
where $w^{\sf{BW}}$ is the price per $1$~Hz of bandwidth for wireless data transmission, and $w^{\sf{C}}$ is the price
paid to execute one CPU cycle at the fog/cloud servers. Assuming that a pre-determined contract agreement specifies a maximum service cost 
$\Theta_{k}^{\sf{max}}$ then $\Theta_k \leq  \Theta_{k}^{\sf{max}}$. This constraint can be rewritten equivalently as
$
(1-s_{k}^{\sf{u}}) \rho_k \leq \rho_k^{\sf{max}} = \frac{\Theta_{k}^{\sf{max}} - w^{\sf{C}}  c_{k,1}}{w^{\sf{BW}}}.
$
Beside the constrained service cost, two important metrics for each user are the service latency and the consumed energy. Specifically, the total delay for completing the computation task of user $k$ includes the computation delay of the mobile device, the average transmission delay of the mobile device, the computation delay of the fog server, the average transmission delay of the fog server over the backhaul link, and the computation delay of the cloud server, which is given by
\begin{eqnarray}
T_k &=&  t_{1,k}^{\sf{u}} + t_{2,k}^{\sf{u}} + t_{1,k}^{\sf{f}} + t_{2,k}^{\sf{f}} + s_{k}^{\sf{c}}T^{\sf{c}} \nonumber\\
&=& \frac{c_{k,0} {+} s_{k}^{\sf{u}} c_{k,1} {+} (1{-}s_{k}^{\sf{u}})c_{k}^{\sf{co,u}}} {f_{k}^{\sf{u}}} +   \frac{(1{-}s_{k}^{\sf{u}})b_k^{\sf{in}}}{\omega_k^{\sf{u}}\rho_k \log_2 (1 {\plus} p_{k} \beta_{k,0})} {+} \frac{s_{k}^{\sf{f}} \big(c_{k,1} {+} c_{k}^{\sf{de,u}}\big)}{f_{k}^{\sf{f}}} {+} \frac{s_{k}^{\sf{c}} b_k^{\sf{in}}}{\omega_k^{\sf{u}}d_{k}} {+}s_{k}^{\sf{c}}T^{\sf{c}}.
\end{eqnarray}
Note that  massive MIMO communications with interference cancellation using MIMO beamforming is assumed in this paper cuch that multiple mobile users can transmit their data to the fog server at the same time over the same frequency band.
Moreover, unlike \cite{ren2019collaborative}, we do not consider the TDMA transmission strategy where
the users are scheduled and have to wait for their turns to transmit their data in the uplink. In other work, no time-based
scheduling is required for the multi-user MIMO communication system considered in our paper.

In addition, the overall energy consumed at user $k$ for processing its task comprises the energy for local computation and for data transmission in the offloading case. Hence, the energy consumption of user $k$ is given by
\begin{eqnarray}
\xi_k = \xi_{1,k}^{\sf{u}} + \xi_{2,k}^{\sf{u}} =  \alpha_{k}{f_{k}^{\sf{u}}}^2 (c_{k,0} {+} s_{k}^{\sf{u}} c_{k,1} {+} (1{-}s_{k}^{\sf{u}})c_{k}^{\sf{co,u}}) +   \frac{(p_{k} {+} p_{k,0})(1{-}s_{k}^{\sf{u}})b_k^{\sf{in}}}{\omega_k^{\sf{u}} \log_2 (1 {\plus} p_{k} \beta_{k,0})} .
\end{eqnarray}

Practically, all users want to save energy and enjoy low application execution latency. 
Hence,  we adopt the WEDC as the objective function of each user $k$ as follows:
\begin{eqnarray}
\Xi_k = w_k^{\sf{T}} T_{k} + w_k^{\sf{E}} \xi_{k}, \\[-30 pt] \nonumber
\end{eqnarray} 
where $w_k^{\sf{T}}$ and $w_k^{\sf{E}}$ represent the weights corresponding to the service latency and consumed energy, respectively. 
These weights can be pre-determined by the users to reflect their priorities or interests. 
The proposed design aims to minimize the WEDC function for each user while maintaining fairness among all users. 
Towards this end, we consider the following min-max optimization problem:
\begin{eqnarray}
\begin{aligned}
(\mathcal{P}_1)  \; \; &  \min\limits_{\Omega_1} \max\limits_{k} \; \Xi_k   \nonumber \\[-5 pt]
\text{s.t} \quad
&(\text{C}1): f_{k}^{\sf{u}} \leq F_k^{\sf{max}},  \forall k , 
\hspace{1.25 cm} (\text{C}4): s_{k}^{\sf{u}} + s_{k}^{\sf{f}} +  s_{k}^{\sf{c}} =1,  \forall k ,
\hspace{1.05 cm} (\text{C}7): 0 \leq \rho_k \leq \rho_k^{\sf{max}}, \; \forall k , \\[-2 pt] 
& (\text{C}2): \sum\nolimits_{k}f_{k}^{\sf{f}} \leq F^{\sf{f,max}}, 
\hspace{0.9 cm} (\text{C}5):\omega_k^{\sf{u,min}} \leq \omega_k^{\sf{u}}\leq \omega_k^{\sf{u,max}},  \forall k , 
\hspace{0.35 cm} (\text{C}8): \sum\nolimits_k d_{k} \leq D^{\sf{max}}, \\[-2 pt]
&(\text{C}3):s_{k}^{\sf{u}},s_{k}^{\sf{f}},s_{k}^{\sf{c}}{\in} \{0,1\}, \forall k , 
\hspace{0.3 cm} (\text{C}6): 0 \leq \rho_k p_{k} \leq P_{k}^{\sf{max}},  \forall k , 
\hspace{1.05 cm} (\text{C}9): T_k \leq T_k^{\sf{max}}, \forall k,  \\[4 pt]
\end{aligned}
\label{p1}	
\end{eqnarray}
where $\Omega_1 = \cup_{k\in \mathcal{K}}\Omega_{1,k}$, $\Omega_{1,k} =\{s_k^{\sf{u}}, s_k^{\sf{f}}, s_k^{\sf{c}},  \omega_k^{\sf{u}}, f_k^{\sf{u}}, f_k^{\sf{f}}, p_k,$ $ \rho_k, d_k\}$;
$F_k^{\sf{max}}$ is the maximum CPU clock speed of user $k$, $F^{\sf{f,max}}$ is the maximum CPU clock speed of the fog server,  
$P_{k}^{\sf{max}}$ is the maximum transmit power of user $k$, $[\omega_k^{\sf{u,min}} , \omega_k^{\sf{u,max}}]$ denotes the
feasible range of the compression ratio $\omega_k^{\sf{u}}$ which can guarantee the required QoS of the recovered data. In particular, for lossless DC where the perceived QoS $q_k^{\sf qu,u} = 1$ for all $\omega_{k}^{\sf u}$, this feasible range is determined as $\omega_k^{\sf{u,min}}=\omega_{k,1}^{\sf{u,min}}$ and $\omega_k^{\sf{u,max}}= \omega_{k,1}^{\sf{u,max}}$.
For lossy DC where the perceived QoS  is required to be greater than $q_k^{\sf qu,u,min}$, this range is determined as $\omega_k^{\sf{u,min}}=\omega_{k,1}^{\sf{u,min}}$  and $\omega_k^{\sf{u,max}} = \min\left\{ \omega_{k,1}^{\sf u,max}, \left((\gamma_{k,3}^{\sf qu,u} -q_k^{\sf qu,u,min})/\gamma_{k,1}^{\sf qu,u} \right)^{1/\gamma_{k,2}^{\sf qu,u}} \right\} $.
In this problem, (C1) and (C2) represent the constraints on the computing resources at the users and the fog server, respectively, while the offloading decision constraints are characterized by (C3) and (C4). The constraints on the compression ratio are captured by (C5), while (C6) and (C7)  impose constraints on the maximum user transmit power and the bandwidth, respectively. Finally, (C8) and (C9) are the constraints on the limited backhaul capacity and delay, respectively.

\section{Optimal Algorithm Design for DC at only Mobile Users} \label{st3}
\subsection{Problem Transformation}
To gain insight into its non-smooth min-max objective function, we recast $(\mathcal{P}_1)$ into the following equivalent problem:
\begin{eqnarray}
\begin{aligned}
(\mathcal{P}_2)  \; \;   \underset{\Omega_1 \cup \eta}{\min} \; \eta   \; \; \text{s.t} \quad
(\text{C}0):\; \Xi_{k} \leq \eta, \forall k,  \;  (\text{C}1)-(\text{C}9),\nonumber
\end{aligned}
\label{p3a1}		
\end{eqnarray}
where $\eta$ is an auxiliary variable. $(\mathcal{P}_2)$ is a
MINLP problem which is difficult to solve due to \textit{the complex fractional and bilinear form of the transmission time and energy consumption, the logarithmic transmission rate function, and the mix of binary offloading decision variables and continuous variables}. Conventional approaches usually 
decompose the problem into multiple subproblems which optimize the
 offloading decision, and the computing and radio resource  allocation separately as in \cite{bi2018computation, du2017computation} or 
relax the binary variables  as in \cite{zhao2017energy, wang2017multi}. These approaches can obtain only sub-optimal solutions. 

To solve the problem optimally, we first study how to classify the users into two sets, namely, a \textit{``locally executing user set''} which is the set of users executing
their applications locally, and an \textit{``offloading user set''} which is the set of users offloading their applications for processing at the fog/cloud server.
This classification is important because, in all constraints of $(\mathcal{P}_2)$, the optimization variables corresponding to the locally executing users are independent from
the optimization variables of the other users. 
Hence, the decisions for the locally executing users 
can be optimized by decomposing $(\mathcal{P}_2)$ into  user independent subproblems which can be solved separately. 
The optimal algorithm is developed based on the bisection search approach where in each search iteration, we perform: 1) 
user classification based on the current value of $\eta$ using the results in \textbf{Theorem~\ref{thrm_opt_class}} below; 2) feasibility verification
for sub-problem $(\mathcal{P}_{\mathcal{B}})$ of $(\mathcal{P}_2)$ corresponding to the offloading user set $\mathcal{B}$; and
3) updates of  lower and upper bounds on $\eta$ according to the feasibility verification outcome. The detailed design
is presented in the following. 


\renewcommand{\baselinestretch}{1.2}
\setlength{\textfloatsep}{5 pt}
\begin{algorithm}[t]
	\footnotesize
	\caption{Optimal Joint DC, Offloading, and Resource Allocation (JCORA)}
	\label{alg1}
	\begin{algorithmic}[1]
		\State \textbf{Initialize}: Compute $\eta_k^{\sf{lo}}, \forall k \in \mathcal{K}$ as in (\ref{eq20}), choose $\epsilon$,  assign $\eta^{\sf{min}} = 0$, $\eta^{\sf{max}}= \max\limits_k(\eta_k^{\sf{lo}})$, and set $\text{BOOL}=\textit{False}$.
		\While{$(\eta^{\sf{max}}-\eta^{\sf{min}}>\epsilon)$ \& $(\text{BOOL} = \textit{False})$}
		\State Assign $\eta = (\eta_{\sf{max}}+\eta_{\sf{min}})/2$, and then define sets $\mathcal{A} =  \lbrace k \vert \eta_{k}^{\sf{lo}} \leq \eta \rbrace$ and $\mathcal{B} =\mathcal{K}/\mathcal{A}$.
		\State Check feasibility of $(\mathcal{P}_{\mathcal{B}})$ as in \textbf{Section~\ref{stC}}. 
		\State \textbf{if} {\textit{$(\mathcal{P}_{\mathcal{B}})$ is feasible}} \textbf{then} $\eta^{\sf{max}} = \eta$, \text{BOOL} = \textit{True},
		\textbf{else}  $\eta^{\sf{min}} = \eta$, \text{BOOL} = \textit{False},
		\textbf{end if}	
		\EndWhile 
	\end{algorithmic}
\end{algorithm}
\renewcommand{\baselinestretch}{1.45}

\subsection{User Classification} \label{UeClass}
Let $\mathcal{A}$ be the locally executing user set, and $\mathcal{B}$ be the offloading user set. We further define any pair of sets $(\mathcal{A},\mathcal{B})$ satisfying $\mathcal{B}=\mathcal{K} \backslash \mathcal{A}$ as a user classification. 
By defining \begin{eqnarray} \mathcal{Q}_{k,0}(f_k^{\sf{u}}) = w_k^\text{\tiny E} \alpha_k (f_k^{\sf{u}})^2 c_{k} +  w_k^\text{\tiny T} {c_{k}}/{f_k^{\sf{u}}}, \end{eqnarray} and $\Omega_{\mathcal{B}} = \cup_{k\in \mathcal{B}}\Omega_{1,k}$, then for a given classification $(\mathcal{A},\mathcal{B})$, problem $(\mathcal{P}_2)$ can be tackled by solving two sub-problems $(\mathcal{P}_{\mathcal{A}})$ and $(\mathcal{P}_{\mathcal{B}})$ for the users in sets $\mathcal{A}$ and $\mathcal{B}$, respectively,  as follows:
 
	\begin{minipage}[t]{0.4\linewidth}
		\begin{eqnarray}
		\begin{aligned}
		(\mathcal{P}_{\mathcal{A}})  \; \; &  \underset{\lbrace f_k^{\sf{u}}\rbrace_{k \in \mathcal{A}}, \eta}{\min} \; \eta \nonumber \\[-0 pt]
		\text{s.t} \quad
		&(\text{CA}0):\; \mathcal{Q}_{k,0}(f_k^{\sf{u}}) \leq \eta, \forall k \in \mathcal{A},\\[-0 pt]               
		&(\text{CA}2):\; {c_{k}}/{T_k^{\sf{max}}} \leq f_{k}^{\sf{u}} \leq F_k^{\sf{max}}, \forall k \in \mathcal{A},\\[-0 pt]
		\end{aligned}
		\label{p2a}	
		\end{eqnarray}
	\end{minipage}%
	\hfill%
	\begin{minipage}[t]{0.4\linewidth}
		\begin{eqnarray}
		\begin{aligned}
		(\mathcal{P}_{\mathcal{B}})  \; \; &  \underset{\Omega_{\mathcal{B}}, \eta}{\min} \; \eta \nonumber \\[-0 pt]
		\text{s.t} \quad
		&(\text{C}0):\; \Xi_{k} \leq \eta, \forall k \in \mathcal{B},\\[-0 pt]               
		&(\text{C}1)-(\text{C}9), \forall k \in \mathcal{B},\\[-0 pt]
		\end{aligned}
		\label{p3a}	
		\end{eqnarray}
	\end{minipage} 
	
\vspace{0.5 cm}

Note that the variable set $\Omega_{1,k}$ corresponding to user $k$ in $\mathcal{A}$ becomes $\lbrace f_k^{\sf{u}} \rbrace$ since we have $s_k^{\sf{u}}=1$ and the other variables can be set equal to zero when user $k$ executes its application locally. 
In such a scenario, $\Xi_k$ can be simplified to $\mathcal{Q}_{k,0}(f_k^{\sf{u}})$.
To attain more insight into the user classification, we now study the relationship between  optimization sub-problems $(\mathcal{P}_{\mathcal{A}})$ and $(\mathcal{P}_{\mathcal{B}})$ in the following lemma.
\begin{lemma}
\label{lemma_1} We denote the optimal values of $(\mathcal{P}_2)$, $(\mathcal{P}_{\mathcal{A}})$, and $(\mathcal{P}_{\mathcal{B}})$ as  $\eta^{\star}$, $\eta_{\mathcal{A}}^{\star}$, and $\eta_{\mathcal{B}}^{\star}$, respectively. Then, we have
\begin{enumerate}
\item $\eta^{\star} \leq \max(\eta_{\mathcal{A}}^{\star},\eta_{\mathcal{B}}^{\star})$ for any classification $(\mathcal{A},\mathcal{B})$.
\item The merged optimal solutions of $(\mathcal{P}_{\mathcal{A}})$ and $(\mathcal{P}_{\mathcal{B}})$ are the optimal solution of $(\mathcal{P}_2)$ if 
\beq \label{opt_class_cond}
\eta^{\star} = \max(\eta_{\mathcal{A}}^{\star},\eta_{\mathcal{B}}^{\star}).
\eeq
\item If $\mathcal{B}^{\prime} \subset \mathcal{B}$, then, we have $\eta_{\mathcal{B}^{\prime}}^{\star} \leq \eta_{\mathcal{B}}^{\star}$.
\end{enumerate}
\end{lemma}
\begin{IEEEproof} The proof is given in Appendix~\ref{prf_lemma_1} .
\end{IEEEproof}

Considering \textbf{Lemma~\ref{lemma_1}}, instead of solving $(\mathcal{P}_2)$, we can equivalently solve the two sub-problems $(\mathcal{P}_{\mathcal{A}})$ and $(\mathcal{P}_{\mathcal{B}})$. 
Moreover, a classification $(\mathcal{A},\mathcal{B})$ is optimal if the condition in \eqref{opt_class_cond} holds. The optimal solution of $(\mathcal{P}_{\mathcal{A}})$ can be obtained as described in \textbf{Proposition~\ref{prop_PA}} while solving $(\mathcal{P}_{\mathcal{B}})$ requires a more complex approach which will be discussed in \textbf{Section~\ref{st_fea_B}}.
\begin{proposition}
\label{prop_PA}
The optimal objective value of $(\mathcal{P}_{\mathcal{A}})$ can be expressed as
$
\eta_{\mathcal{A}}^{\star} = \max_{k \in \mathcal{A}} \eta_k^{\sf{lo}},
$
where $\eta_k^{\sf{lo}}$ is defined as 
\beq \label{eq20}
\eta_k^{\sf{lo}} = \begin{cases}
	\mathcal{Q}_{k,0} (f_k^{\sf{u,sta}}),  \text{ if }  f_k^{\sf{u,sta}} \in [f_k^{\sf{u,min}},F_k^{\sf{max}}]\\
	\min\big(\mathcal{Q}_{k,0}(f_k^{\sf{u,min}}), \mathcal{Q}_{k,0}(F_k^{\sf{max}}) \big), \text{ otherwise},
	\end{cases}
\eeq
where $f_k^{\sf{u,min}} = {c_{k}}/{T_k^{\sf{max}}}$ and $f_k^{\sf{u,sta}}= \sqrt[3]{{w_k^{\sf{T}}}/{(2w_k^{\sf{E}}  \alpha_k)}}$.

\end{proposition}
\begin{IEEEproof} 
The proof is given in Appendix~\ref{prf_prop_PA}.
\end{IEEEproof}

Based on the results in \textbf{Lemma~\ref{lemma_1}} and \textbf{Proposition~\ref{prop_PA}}, the optimal user classification 
can be performed as described in the following theorem. 
\begin{theorem} \label{thrm_opt_class}
If $\eta^{\star}$ is the optimum objective value of problem $(\mathcal{P}_2)$, then an optimal classification, $(\mathcal{A}^{\star},\mathcal{B}^{\star})$, can be determined as
$ 
\mathcal{A}^{\star}  =  \lbrace k \vert \eta_k^{\sf{lo}} \leq \eta^{\star} \rbrace, \label{opt_class1} \text{ and }
\mathcal{B}^{\star}  =  \mathcal{K} \backslash \mathcal{A}^{\star}. \label{opt_class2}
$
\end{theorem}
\begin{IEEEproof}
The proof is given in \textbf{Appendix~\ref{prf_thrm_opt_class}}.
\end{IEEEproof}

\begin{figure}[t]
	\centering
	\begin{minipage}[t]{0.79\linewidth}
		\includegraphics[width=0.9\textwidth]{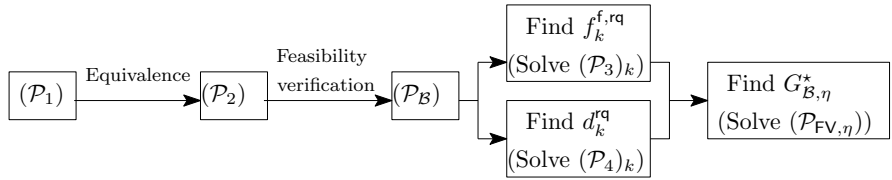}
		\caption{Relationship between the (sub)problems when solving $(\mathcal{P}_1)$  by the JCORA algorithm.}
		\label{P2-figAlg1}
	\end{minipage}%
\end{figure}

\subsection{General Optimal Algorithm Design} \label{stC}
The results in \textbf{Theorem~\ref{thrm_opt_class}} are now employed to develop an optimal algorithm for solving $(\mathcal{P}_2)$ by iteratively 
solving $(\mathcal{P}_{\mathcal{A}})$ and $(\mathcal{P}_{\mathcal{B}})$ and updating $({\mathcal{A}},{\mathcal{B}})$ until the optimal $({\mathcal{A}}^{\star},{\mathcal{B}}^{\star})$ is obtained. 
The general optimal algorithm is presented in \textbf{Algorithm~\ref{alg1}}. In this algorithm, we initially calculate $\eta_k^{\sf{lo}}$ for all users in $\mathcal{K}$ as in \eqref{eq20}.
Then, we employ the bisection search to find the optimum $\eta^{\star}$ where  upper bound $\eta^{\sf{max}}$ and lower bound $\eta^{\sf{min}}$ are iteratively updated until the difference between them becomes sufficiently small, 
$(\mathcal{P}_{\mathcal{B}})$
is feasible, and the sets $\mathcal{A}$ and $\mathcal{B}$ do not change. 
At convergence, the optimal classification solution can be obtained by merging the solutions of $(\mathcal{P}_{\mathcal{A}})$ and $(\mathcal{P}_{\mathcal{B}})$. 
The optimal solution of $(\mathcal{P}_{\mathcal{A}})$ can be determined using \textbf{Proposition~\ref{prop_PA}}
and the verification of the feasibility of $(\mathcal{P}_{\mathcal{B}})$ is addressed in the following. The relationship between the (sub)problems when solving $(\mathcal{P}_1)$ is illustrated in Fig.~\ref{P2-figAlg1}.


\subsection{Feasibility Verification of $(\mathcal{P}_{\mathcal{B}})$ } \label{st_fea_B}
In order to verify the feasibility of $(\mathcal{P}_{\mathcal{B}})$, we consider the following problem
\begin{eqnarray}
(\mathcal{P}_{\sf{FV},\eta})  \; \;   \min\limits_{\Omega_{\mathcal{B}}} \; \sum\limits_{k \in \mathcal{B}} f_k^{\sf{f}} 
\quad \text{s.t.} \; (\text{C}0), (\text{C}1),  (\text{C}3) -  (\text{C}9). \nonumber 
\label{p5}	
\end{eqnarray}
This problem minimizes the total required computing resource of the fog server subject to all constraints of $(\mathcal{P}_{\mathcal{B}})$ except $(\text{C}2)$.  
Let $G_{\mathcal{B},\eta}^{\star}$ be the objective value of problem $(\mathcal{P}_{\sf{FV},\eta})$. 
Then, the feasibility of $(\mathcal{P}_{\mathcal{B}})$ can be verified by comparing $G_{\mathcal{B},\eta}^{\star}$ to the available fog computing resource $F^{\sf{f,max}}$. 
In particular, problem $(\mathcal{P}_{\mathcal{B}})$ is feasible if $G_{\mathcal{B},\eta}^{\star} \leq F^{\sf{f,max}}$. 
Otherwise, $(\mathcal{P}_{\mathcal{B}})$ is infeasible.

We propose to solve $(\mathcal{P}_{\sf{FV},\eta})$ as follows.
First, recall that there are two possible scenarios for executing the tasks of the users in set $\mathcal{B}$ (referred to as modes): \textbf{\textit{Mode 1}} - task execution at the fog server, i.e., $s_k^{\sf{f}}=1$; \textbf{\textit{Mode 2}} -  task execution at the cloud server, i.e., $s_k^{\sf{c}}=1$. 
In addition, the fog computing resources are only required by the users in \textbf{\textit{Mode 1}} and the backhaul resources are only used by the users in \textbf{\textit{Mode 2}}.
Considering these two modes, a three-step solution approach is proposed to verify the feasibility of sub-problem  $(\mathcal{P}_{\mathcal{B}})$ as follows.
In \textbf{Step~1}, the minimum required fog computing resource of every user is determined by assuming that it is in \textbf{\textit{Mode 1}}. This step is fulfilled by solving  sub-problem $(\mathcal{P}_3)_k$ for every user $k$, see \textbf{Section~\ref{st4a1}}.
In \textbf{Step~2}, the minimum required backhaul rate for each user is optimized by assuming that it is in \textbf{\textit{Mode 2}}. This step can be accomplished by solving subproblem $(\mathcal{P}_4)_k$ for every user $k$, see \textbf{Section~\ref{st4a2}}.
In \textbf{Step~3}, using the results obtained in the two previous steps, problem $(\mathcal{P}_{\sf{FV},\eta})$ is equivalently 
transformed to a mode-mapping problem, see \textbf{Section~\ref{st5D}}.

\subsubsection{\textbf{Step 1} - Minimum Fog Computing Resources for User $k \in \mathcal{B}$} \label{st4a1}
If the application of user $k$ is executed at the fog server, the minimum fog computing resource required for this application, denoted as $f_{k}^{\sf{f,rq}}$, can be optimized based on the following sub-problem:
\begin{eqnarray}
(\mathcal{P}_3)_k  \; \;   \min\limits_{\Omega_{2,k}} \; f_k^{\sf{f}} 
\quad \text{s.t.} \quad s_k^{\sf{f}}  =  1,\;  (\text{C}0)_k, \;(\text{C}1)_k, \; (\text{C}5)_k -  (\text{C}7)_k, \; (\text{C}9)_k, \nonumber 
\label{p3k}	
\end{eqnarray} 
where $\Omega_{2,k} = \{ \omega_k^{\sf{u}}, f_k^{\sf{u}}, f_k^{\sf{f}},  p_k, \rho_k  \}$, $(\text{C}0)_k$, $(\text{C}1)_k$, $(\text{C}5)_k-(\text{C}7)_k$, and $(\text{C}9)_k$ denote the respective constraints of user $k$ corresponding to $(\text{C}0)$, $(\text{C}1)$, $(\text{C}5)-(\text{C}7),$ and $(\text{C}9)$.
In sub-problem $(\mathcal{P}_3)_k$, the WEDC function $\Xi_k$ consists of posynomials and
other terms involving $\log(1+p_k \beta_{k,0} )$. We can convert $\Xi_k$ into a convex function via logarithmic transformation
as follows. 
When $s_k^{\sf{f}} = 1$, all variables in set $\Omega_{2,k}$ must be positive to satisfy constraints (C0) and (C9); therefore, we can employ the following variable transformations: ${\tilde{\omega}_k^{\sf{u}}} = \log(\omega_k^{\sf{u}})$, ${\tilde{f}_k^{\sf{u}}} = \log (f_k^{\sf{u}})$, ${\tilde{f}_k^{\sf{f}}} = \log (f_k^{\sf{f}})$, ${\tilde{p}_k} = \log (p_k)$, and ${\tilde{\rho}_k} = \log (\rho_k)$. With these transformations, the objective function and all constraints of $(P_3)_k$ except $(\text{C}0)_k$ and $(\text{C}9)_k$ are converted into a linear form  while the total delay and the WEDC in  $(\text{C}9)_k$ and $(\text{C}0)_k$ can be 
rewritten, respectively, as 
$  T_k  =   \frac{ b_k^{\sf{in}} \e^{- {\tilde{\omega}_k^{\sf{u}}} - {\tilde{\rho}_k}}  }{\log \big(1+ \beta_{k,0} \e^{{\tilde{p}_k}}  \big)}  + \mathcal{Q}_{k,1}, \text{ and }
 \Xi_k  = \frac{ w_k^\text{\tiny E} b_k^{\sf{in}} \big[ \e^{{\tilde{p}_k} - {\tilde{\omega}_k^{\sf{u}}} }+ p_{k,0}\e^{-{\tilde{\omega}_k^{\sf{u}}}} \big]  }{\log \big(1+ \beta_{k,0} \e^{{\tilde{p}_k}}  \big)} + w_k^\text{\tiny E} \alpha_k \mathcal{Q}_{k,2} + w_k^\text{\tiny T}T_k, 
$
where $\mathcal{Q}_{k,1} = \big(c_{k,0} {\plus} \gamma_{k,0}^{\sf{u}} \gamma_{k,3}^{\sf{co}} \big) \e^{-{\tilde{f}_k^{\sf{u}}}}             {\plus}            \gamma_{k,0}^{\sf{u}} \gamma_{k,1}^{\sf{co}} \e^{\big(-{\tilde{f}_k^{\sf{u}}} {\plus} \gamma_{k,2}^{\sf{co}} {\tilde{\omega}_k^{\sf{u}}} \big)}          
{\plus}              \big(c_{k,1} {\plus} \gamma_{k,0}^{\sf{u}} \gamma_{k,3}^{\sf{de}} \big)\e^{-{\tilde{f}_k^{\sf{f}}}}     {\plus}      \gamma_{k,0}^{\sf{u}} \gamma_{k,1}^{\sf{de}} \e^{\big(-{\tilde{f}_k^{\sf{f}}}{+} \gamma_{k,2}^{\sf{de}} {\tilde{\omega}_k^{\sf{u}}} \big)}
$ and $\mathcal{Q}_{k,2} =  \big(c_{k,0} \plus \gamma_{k,0}^{\sf{u}} \gamma_{k,3}^{\sf{co}} \big) \e^{2{\tilde{f}_k^{\sf{u}}}}  
\plus  \gamma_{k,0}^{\sf{u}} \gamma_{k,1}^{\sf{co}} \e^{\big(2{\tilde{f}_k^{\sf{u}}} + \gamma_{k,2}^{\sf{co}} {\tilde{\omega}_k^{\sf{u}}} \big)} $. The convexity of $(P_3)_k$ is formally stated in the following proposition.
\begin{proposition} \label{prop2}
	Sub-problem $(\mathcal{P}_3)_k$ is convex with respect to set $\tilde{\Omega}_{2,k} \cup {\tilde{l}_k}$, where ${\tilde{l}_k} = {\tilde{\omega}_k^{\sf{u}}} + {\tilde{\rho}_k}$  and $\tilde{\Omega}_{2,k} = \{ {\tilde{\omega}_k^{\sf{u}}}, {\tilde{f}_k^{\sf{u}}}, {\tilde{f}_k^{\sf{f}}}, {\tilde{p}_k}, {\tilde{\rho}_k} \}$.	
\end{proposition}
\begin{IEEEproof}
	The proof is given in \textbf{Appendix~\ref{prf_prop2}}.
\end{IEEEproof}

Based on \textbf{Proposition~\ref{prop2}}, we can apply the interior point method to find the optimal solution $\tilde{\Omega}_{2,k}^\star = \{{\tilde{\omega}_k}^{\sf{u}\star}$, ${\tilde{f}_k^{\sf{u}\star}}$, ${\tilde{f}_k^{\sf{f}\star}}$, ${\tilde{p}_k}^\star$, ${\tilde{\rho}_k}^\star \} $ of $(\mathcal{P}_3)_k$ \cite{boyd2004convex}. The original optimal solution ${\Omega}_{2,k}^\star =\{ \omega_k^{\sf{u}\star}, \; f_k^{{\sf{u}} \star}, \;f_k^{{\sf{f}}\star},  \; p_k^\star, \; \rho_k^\star \} $ can then be obtained from $\tilde{\Omega}_{2,k}^\star$. If $(\mathcal{P}_3)_k$ is infeasible, we set $s_k^{\sf{f}}=0$. It is noted that $f_k^{{\sf{f}}\star}$ is also the value of $f_k^{\sf{f,rq}}$.

\renewcommand{\baselinestretch}{1.2}
\setlength{\textfloatsep}{6 pt}
\begin{algorithm}[t]
	\footnotesize
	\caption{Feasibility Verification of $(\mathcal{P}_{\mathcal{B}})$}
	\label{alg2}
	\begin{algorithmic}[1]	
		\parState {Solve $(\mathcal{P}_3)_k$ to find $f_{k}^{\sf{f,rq}}, \forall k \in \mathcal{B}$, as  in \textbf{Section~\ref{st4a1}}.}
		\parState {Solve $(\mathcal{P}_4)_k$ to find $d_k^{\sf{rq}}, \forall k \in \mathcal{B}$, as in \textbf{Section~\ref{st4a2}}.} 
		\If {$\exists k$ such that $ s_k^{\sf{f}} + s_k^{\sf{c}}=0$ } Return  $(\mathcal{P}_{\mathcal{B}})$ is infeasible,
		\Else		\; {Solve $(\mathcal{P}_{\sf{FV},\eta})$ to find $G_{\mathcal{B},\eta}^{\star}$, as in \textbf{Section~\ref{st5D}}.}
		\parState{\textbf{if} $G_{\mathcal{B},\eta}^{\star} < F^{\sf{f,max}}$ \textbf{then} Return  $(\mathcal{P}_{\mathcal{B}})$ is feasible,
		\textbf{else} Return  $(\mathcal{P}_{\mathcal{B}})$ is infeasible \textbf{end if}}
		\EndIf			
	\end{algorithmic}
\end{algorithm}
\renewcommand{\baselinestretch}{1.45}

\subsubsection{\textbf{Step 2} - Minimum Allocated Backhaul Resource  for User $k \in \mathcal{B}$} \label{st4a2}
If the application of user $k$ is executed at the cloud server, the minimum backhaul capacity for transferring its application to the cloud server, denoted as $d_k^{\sf{rq}}$, can be determined by solving the following sub-problem:
\begin{eqnarray}
(\mathcal{P}_4)_k \quad  \min\limits_{\Omega_{2,k}\cup {d_k}\backslash f_k^{\sf{f}}} \quad  d_k \quad \text{ s.t. } \quad s_k^{\sf{c}}  =  1, \; (\text{C}0)_k, \;(\text{C}1)_k,\; (\text{C}5)_k  -  (\text{C}7)_k, \; (\text{C}9)_k. \nonumber
\label{p4k}	
\end{eqnarray}
Similar to  $(\mathcal{P}_{3})_k$, $(\mathcal{P}_{4})_k$ can be converted to a convex problem via logarithmic transformations; thus, we can find the optimal point $d_k^{\sf{rq}}$. If $(\mathcal{P}_4)_k$ is infeasible, we set $s_k^{\sf{c}}=0$.

\subsubsection{\textbf{Step 3} - Feasibility Verification} \label{st5D}
With the obtained values $f_k^{\sf{f,rq}}$ and $d_k^{\sf{rq}}$, problem $(\mathcal{P}_{\sf{FV},\eta})$ can be transformed to
\begin{eqnarray}
(\mathcal{P}_{\sf{FV},\eta})   &  \min\limits_{\Omega_3}  \;  \mathcal{G}_{\mathcal{B},\eta}(\Omega_3) =\sum_{k \in \mathcal{B}} (1-s_k^{\sf{c}}) f_{k}^{\sf{f,rq}}  \;
\text{ s.t. } \;
(\text{C}3,4,8): \sum\nolimits_{k \in \mathcal{B}} s_k^{\sf{c}} d_k^{\sf{rq}} \leq D^{\sf{max}}, \;s_{k}^{\sf{c}} \in \{0,1\}, \nonumber
\label{p51}	
\end{eqnarray}
where $\Omega_{3} = \{ s_k^{\sf{c}}| k \in \mathcal{B} \} $ for a given $\eta$. 
In fact, $(\mathcal{P}_{\sf{FV},\eta})$ is a \textit{``0-1 knapsack''} problem \cite{martello1990knapsack}, which can be solved optimally and effectively using the CVX solver.
If $G_{\mathcal{B},\eta}^{\star} \leq F^{\sf{f,max}}$, combining the set of all solutions of the $(\mathcal{P}_3)_k$'s, $(\mathcal{P}_4)_k$'s, and $(\mathcal{P}_{\sf{FV},\eta})$ yields a feasible solution of $(\mathcal{P}_\mathcal{B})$ for this value of $\eta$. Hence, $(\mathcal{P}_\mathcal{B})$ is feasible in such scenario.  
The feasibility verification of $(\mathcal{P}_{\mathcal{B}})$ is summarized in \textbf{Algorithm~\ref{alg2}}.

\subsection{Optimal JCORA Algorithm to Solve $(\mathcal{P}_{2})$} \label{st4}
Based on the results presented in the previous sections, the solution of $(\mathcal{P}_2)$ can be found by employing \textbf{Algorithm~\ref{alg1}} and the $(\mathcal{P}_\mathcal{B})$ feasibility verification presented in \textbf{Algorithm~\ref{alg2}}. The optimality of the obtained solution 
is formally stated in the following theorem.

\begin{theorem} \label{thrm_opt_P2}
	The integration of \textbf{Algorithm~\ref{alg2}} into \textbf{Algorithm~\ref{alg1}} yields the global optimum of MINLP $(\mathcal{P}_2)$. 
\end{theorem}
\begin{IEEEproof}
	\textbf{Algorithm~\ref{alg2}} verifies the feasibility of $(\mathcal{P}_{\mathcal{B}})$ for any given value 
	of $\eta_\mathcal{B} = \eta$. Therefore, if \textbf{Algorithm~\ref{alg1}} employs \textbf{Algorithm~\ref{alg2}}, $(\mathcal{P}_2)$ is solved optimally. Note that after convergence, the optimal variables are given by the optimal solution of $(\mathcal{P}_3)_k$ if  $s_k^{\sf{f}}=1$ or $(\mathcal{P}_4)_k$ if  $s_k^{\sf{c}}=1$ where the values of the $s_k^{\sf{f}}$'s and $s_k^{\sf{c}}$'s are the outcomes of $(\mathcal{P}_{\sf{FV},\eta})$.
\end{IEEEproof}

\subsection{Complexity Analysis}
\label{sec:ComAna}
We analyze the computational complexity of the JCORA algorithm in terms of the required number of arithmetic operations. In all proposed algorithms, the while-loop for the bisection search of $\eta$ requires $\log_2(\frac{\eta^{\sf{max}}-\eta^{\sf{min}}}{\epsilon})$ iterations. To verify the feasibility of $(\mathcal{P}_{\mathcal{B}})$ for a given $\eta$, the convex problems $(\mathcal{P}_3)_k$ and $(\mathcal{P}_4)_k$ can be solved by using the interior point method with complexity $\mathcal{O}(m_1^{1/2}(m_1+m_2)m_2^2)$, where $m_1$ is the number of equality constraints and $m_2$ represents the number of variables \cite{hoang2016energy}. It can be verified that $(\mathcal{P}_3)_k$ and $(\mathcal{P}_4)_k$  have the same complexity. On the other hand, the knapsack problem $(\mathcal{P}_{\sf{FV},\eta})$ for $|\mathcal{B}|$ users  can be solved by \textbf{Algorithm \ref{alg2}} in pseudo-polynomial time with complexity $\mathcal{O}(\nu_1|\mathcal{B}|)$, where $\nu_1$ is determined by coefficients in $(\mathcal{P}_{\sf{FV},\eta})$ \cite{martello1990knapsack}.  Moreover, $(\mathcal{P}_3)_k$ and $(\mathcal{P}_4)_k$ can be solved independently for all users $k \in \mathcal{B}$; therefore, the complexity of each bisection search step can be expressed as $|\mathcal{B}|\mathcal{O}((\mathcal{P}_3)_k) + |\mathcal{B}|\mathcal{O}((\mathcal{P}_4)_k) + \mathcal{O}(\mathcal{P}_{\sf{FV},\eta}) = \mathcal{O}(\nu_2|\mathcal{B}|)$, where $\nu_2 = \nu_1+2m_1^{1/2}(m_1+m_2)m_2^2$. Consequently, the overall complexity of \textbf{Algorithm~\ref{alg1}} 
is $\mathcal{O}(\log_2(\frac{\eta^{\sf{max}}-\eta^{\sf{min}}}{\epsilon})\nu_2|\mathcal{B}|)$.

\section{DC at Both Mobile Users and Fog Server } \label{st4a}

We now consider the more general case where the fog server also
performs DC before transmitting the compressed data over the backhaul link to
the cloud server. This design option can further enhance the performance for systems with a congested backhaul link.
The backhaul compression ratio is defined as $\omega_k^{\sf{f}} = b_k^{\sf{in}}/b_k^{\sf{out,f}}$ where $b_k^{\sf{out,f}}$ 
stands for the number of bits transmitted over the backhaul link.
Note that if $b_k^{\sf{out,f}}=b_k^{\sf{out,u}}$, then no DC 
 is employed at the fog server, which corresponds to the design in \textbf{Section~\ref{st3}}.
Hence, \textbf{\textit{Mode 2}} in \textbf{Section~\ref{st4a1}} is equivalent to the scenario that the task is executed at the cloud server without DC at the fog server.
However, the fog server can re-compress the data before transmitting it to the cloud server for processing, which is referred to
 as \textbf{\textit{Mode 3}} in the following.
Denote $s_k^{\sf m}$ as the binary variable indicating whether or not DC is performed at the fog server for user $k$ 
($s_k^{\sf m}=1$ for DC, and $s_k^{\sf m}=0$, otherwise).
Then, we have $s_k^{\sf{f}}=1$ if user $k$ is in \textbf{\textit{Mode 1}};  $s_k^{\sf{c}}=1$  if user $k$ is in \textbf{\textit{Mode 2}}; $s_k^{\sf{m}}=1$  if user $k$ is in \textbf{\textit{Mode 3}}. In this general case, constraints (C3) and (C4)  can be
 rewritten as
$\text{$\check{\text{C}}$3:} \; s_k^{\sf{u}}, s_k^{\sf{f}}, s_k^{\sf{c}}, s_k^{\sf{m}} \in \{0,1\}, \forall k \in \mathcal{K}$ and $(\check{\text{C}}4)$: $s_k^{\sf{u}}+ s_k^{\sf{f}}+ s_k^{\sf{c}}+ s_k^{\sf{m}} =1, \forall k \in \mathcal{K}$.

Then, the computational load for compression and the output data corresponding to \textbf{\textit{Mode 3}} can be modeled as
$
c_k^{\sf{co,f}} = \gamma_{k,0}^{\sf{f}}  \big[ \gamma_{k,1}^{\sf{co},\sf{f}}(\omega_k^{\sf{f}})^{\gamma_{k,2}^{\sf{co},\sf{f}}} + \gamma_{k,3}^{\sf{co},\sf{f}} \big] \text{ and }
b_k^{\sf{out,f}} = b_k^{\sf{in}}/\omega_k^{{\sf{f}}}, 
$ respectively,
where $\gamma_{k,0}^{\sf{f}}, \gamma_{k,1}^{\sf{co},\sf{f}}, \gamma_{k,3}^{\sf{co},\sf{f}} \in \mathbb{R}_+$ are positive numbers. 
Here, we have additional constraints for the compression processes at the fog server as
$
(\check{\text{C}}10)$:  $\omega_{k}^{\sf{f}} \in [\omega_{k}^{\sf{f,min}}, \omega_{k}^{\sf{f,max}}], \; \forall k \in \mathcal{K}. 
$

Then, the total computational load for user ${k}$ at the fog server becomes
$
\check{c}_{k}^{\sf{f}} = s_{k}^{\sf{f}} \big(c_{k,1} + c_{k}^{\sf{de,u}}\big) + s_{k}^{\sf{m}}  (c_{k}^{\sf{co},\sf{f}} + c_{k}^{\sf{de,u}}),
$ and the computing time at the fog server is $\check{t}_{1,k}^{\sf{f}} =\check{c}_{k}^{\sf{f}}/f_k^{\sf{f}}$.
Moreover, the transmission time incurred by offloading the data of user $k$ from the fog server to the cloud server can be rewritten as
$
\check{t}_{2,k}^{\sf{f}} = \big(s_{k}^{\sf{f}} b_k^{\sf{out,u}}  + s_{k}^{\sf{m}} b_k^{\sf{out,f}} \big)/d_{k}.
$
Then, the total delay for completing the computation task of user $k$ is given by $\check{T}_k = t_{1,k}^{\sf{u}} + t_{2,k}^{\sf{u}} +\check{t}_{1,k}^{\sf{f}} +\check{t}_{2,k}^{\sf{f}}+(s_k^{\sf{c}}+s_k^{\sf{m}})T^{\sf{c}}$, and the WEDC becomes $\check{\Xi}_k = w_k^{\sf{T}}\check{T}_k + w_k^{\sf{E}}\xi_k$. Then, constraint (C9) is rewritten as $ 
(\check{\text{C}}9)$: $\check{T}_k \leq T_k^{\sf{max}}.
$

\begin{figure}[t]
	\centering
	\begin{minipage}[t]{0.8\linewidth}
		\includegraphics[width=1\textwidth]{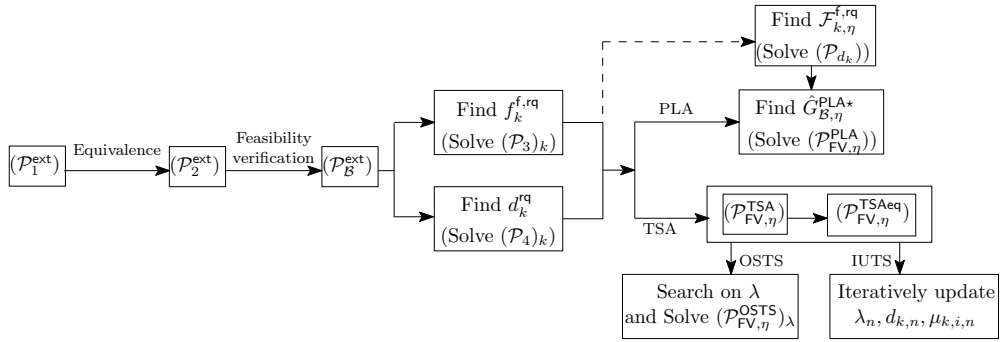}
		\caption{Relationship between the (sub)problems when solving $(\mathcal{P}_1^{\sf{ext}})$.}
		\label{P2-figAlg2}
	\end{minipage}%
\end{figure}

With the additional variables $s_k^{\sf m}$ and $\omega_{k}^{\sf{f}} , \forall k \in \mathcal{B}$, the extended versions of problems ($\mathcal{P}_1$) and ($\mathcal{P}_2$) can be stated, respectively, as 
\begin{eqnarray}
\begin{aligned}
(\mathcal{P}_1^{\sf{ext}}) & \; \;   \min\limits_{\Omega_1 \cup_k \{s_k^{\sf m},\omega_{k}^{\sf{f}}\}} \max\limits_{k} \; \check{\Xi}_k  \; \; \text{s.t} \quad
(\text{C}1),(\text{C}2),(\text{C}5)-(\text{C}8), (\check{\text{C}}3), (\check{\text{C}}4), (\check{\text{C}}9), (\check{\text{C}}10).\\[-0 pt]
(\mathcal{P}_2^{\sf{ext}}) & \; \;   \underset{\Omega_1 \cup_k { \{s_k^{\sf m},\omega_{k}^{\sf{f}}\}}\cup \eta}{\min} \; \eta   \; \; \text{s.t} \quad
(\check{\text{C}}0): \; \check{\Xi}_k \leq \eta, \; (\text{C}1),(\text{C}2),(\text{C}5)-(\text{C}8), (\check{\text{C}}3), (\check{\text{C}}4), (\check{\text{C}}9), (\check{\text{C}}10). \nonumber \\[-0 pt]
\end{aligned}
\end{eqnarray}

The main challenge for solving the extended problem in comparison to the original one comes from the  users in \textbf{\textit{Mode 3}}. These users require both fog computing and backhaul resources. 
To solve the extended problem, we employ the general solution approach presented in \textbf{Section~\ref{st3}} 
but modify the feasibility verification for $(\mathcal{P}_{\mathcal{B}})$.
In particular, \textbf{Algorithm~\ref{alg1}} is used to determine sets $\mathcal{A}$ and $\mathcal{B}$ for a given $\eta$ and 
we update $\eta$ using the bisection search method.
The results in \textbf{Theorem~\ref{thrm_opt_class}} are still applicable for the extended problem.
In the following, we propose several techniques for dealing with \textbf{\textit{Mode 3}} and verify the feasibility 
of user classification for a given $\eta$ in \textit{Step 4} of \textbf{Algorithm~\ref{alg1}}.

For a given $\eta$, $(\mathcal{P}_{\mathcal{B}}^{\sf{ext}})$ is obtained by adding $ (\check{\text{C}}10)$ to $(\mathcal{P}_{\mathcal{B}})$ and replacing $\Xi_k$ and $T_k$  by $\check{\Xi}_k$ and $\check{T}_k$, respectively. 
To verify the  feasibility of $(\mathcal{P}_{\mathcal{B}}^{\sf{ext}})$, a similar
three-step solution approach as for $(\mathcal{P}_{\mathcal{B}})$ is employed.
In \textbf{Steps 1} and \textbf{2}, $f_k^{\sf{f,rq}}$ and $d_k^{\sf{rq}}$ which correspond to the users in \textbf{\textit{Mode 1}} and \textit{\textbf{2}} are optimized by solving $(\mathcal{P}_3)_k$ and $(\mathcal{P}_4)_k$ as in \textbf{Sections~\ref{st4a1}} and \textbf{\ref{st4a2}}, respectively.
In \textbf{Step 3}, we first investigate the network resources required by the users in \textbf{\textit{Mode 3}}, modify problem $(\mathcal{P}_{\sf{FV},\eta})$ to adapt it to the extended problem, and solve that problem to verify the feasibility. 
Three different methods for this extended problem will be proposed as follows.

In the first approach, we represent $f_k^{\sf{f,rq}}$ of user $k$ in \textbf{\textit{Mode 3}} as a function of $d_k$ by employing a piece-wise linear approximation (PLA) method. Based on this approximation, we transform $(\mathcal{P}_{\sf{FV},\eta})$ into a standard mixed-integer linear programming (MILP) problem, $(\mathcal{P}_{\sf{FV},\eta}^{\sf{PLA}})$, which can be solved effectively by using the CVX solver.
In the other two approaches, we directly deal with the modified problem $(\mathcal{P}_{\sf{FV},\eta}^{\sf{TSA}})$ without approximating $f_k^{\sf{f,rq}}$ of user $k$ in \textbf{\textit{Mode 3}}. To cope with this challenging MINLP problem, we first 
reduce the optimized variable set by exploiting some useful relations among the variables. 
Then, two algorithms are proposed to solve the resulting problem for the remaining variables.
One algorithm is based on a one-dimensional search for the Lagrangian multiplier, see \textbf{Section~\ref{P61}}, while the 
other algorithm iteratively updates the Lagrangian multiplier, see \textbf{Section~\ref{P62}}. The relationship between the (sub)problems when solving $(\mathcal{P}_1^{\sf ext})$ is illustrated in Fig.~\ref{P2-figAlg2}.

\subsection{Piece-wise Linear Approximation based Algorithm (PLA)} \label{st5a}
After determining the minimum computing and backhaul resources, $f_{k}^{\sf{f,rq}}$ and $d_k^{\sf{rq}}$, required in \textbf{\textit{Modes 1}} and \textbf{\textit{2}}, respectively,  one can set $d_k \in (0,d_k^{\sf{rq}})$ for the users in \textbf{\textit{Mode 3}}.
We now study the relationship between $f_{k}^{\sf{f}}$ and $d_k$ in \textbf{\textit{Mode 3}} where user $k$ demands both fog computing resources
for re-compression and backhaul capacity resources.
Towards this end, we determine the required fog computing resources for a given  $d_k \in (0,d_k^{\sf{rq}})$ by 
solving the following problem:
\begin{eqnarray}
({\mathcal{P}}_{d_k})  \; \;   \min\limits_{\Omega_{2,k} \cup \{\omega_{k}^{\sf{f}}\}}  \; f_k^{\sf{f}} \quad
\text{s.t.} \quad
 s_k^{\sf{m}} = 1, (\check{\text{C}}0)_k, (\text{C}1)_k, (\text{C}5)_k-(\text{C}7)_k, (\check{\text{C}}9)_k, (\check{\text{C}}10)_k. \nonumber
\label{p4ke}	
\end{eqnarray}
Let $\mathcal{F}_{k,\eta}^{\sf{f,rq}}(d_k)$ be the optimal solution of this problem, which can be obtained
by employing the logarithmic transformations described in \textbf{Section~\ref{st4a1}}. 
However, finding a closed-form expression for $\mathcal{F}_{k,\eta}^{\sf{f,rq}}(d_k)$  is not tractable. 

\renewcommand{\baselinestretch}{1.2}
\setlength{\textfloatsep}{6 pt}
\begin{algorithm}[t]
	\footnotesize
	\caption{PLA-based Feasibility Verification for $(\mathcal{P}_{\mathcal{B}}^{\sf{ext}})$}
	\label{alg_PLA}
	\begin{algorithmic}[1]
		\State \textbf{Initialize}: $L$, $\eta$ 
		\parState {Compute $f_{k}^{\sf{f,rq}}$ and $d_k^{\sf{rq}}$ for all $k \in \mathcal{B}$ as in Step~1~and~2 of \textbf{Algorithm~\ref{alg2}}.}
		\parState {Define $d_{k,l} = (d_k^{\sf{rq}}-\epsilon_{\sf d})l/L, ~ \forall k \in \mathcal{B}, l=0:L$.}		
		\parState {Compute $\mathcal{F}_{k,\eta}^{\sf{f,rq}}(d_{k,l})$. \textbf{If} $\mathcal{F}_{k,\eta}^{\sf{f,rq}}(d_{k,l})$ is unbound \textbf{then} Remove point $d_{k,l}$ \textbf{end if}.}
		\parState {Compute $A_{k,l}$,  $B_{k,l}$, and then solve $(\mathcal{P}_{\sf{FV},\eta}^{\sf{PLA}})$ to get optimal value $\hat{G}_{\mathcal{B},\eta}^{\sf{PLA}\star}$  of $(\mathcal{P}_{\sf{FV},\eta}^{\sf{PLA}})$.}
		\parState { \textbf{if} $\hat{G}_{\mathcal{B},\eta}^{\sf{PLA}\star} \leq F^{\sf{f,max}}$ \textbf{then} Return  $(\mathcal{P}_{\mathcal{B}}^{\sf{ext}})$ is feasible,  
		\textbf{else} Return  $(\mathcal{P}_{\mathcal{B}}^{\sf{ext}})$ is infeasible
		\textbf{end if}}
	\end{algorithmic}
\end{algorithm}
\renewcommand{\baselinestretch}{1.45}
Hence, we propose to employ the \textit{``Piece-wise Linear Approximation''} (PLA) method to divide 
the original domain into multiple small segments such that $\mathcal{F}_{k,\eta}^{\sf{f,rq}}(d_k)$ can be 
approximated by a linear function 
in each segment. 
Suppose that the interval $[\epsilon_{\sf d},d_k^{\sf{rq}}{-}\epsilon_{\sf d}]$ is divided into $L$ segments of equal size, where $\epsilon_{\sf d}$ is a very small number compared to $d_k^{\sf{rq}}$, e.g, $\epsilon_{\sf d} = 1$. 
Specifically, the ${l}^\text{th}$ segment corresponds to  interval $[d_{k,l},d_{k,l+1} ]$, 
where $d_{k,l} {=} (d_k^{\sf{rq}}{-}\epsilon_{\sf d})l/L$ is a point such that $\mathcal{F}_{k,\eta}^{\sf{f,rq}}(d_{k,l})$ and the value of the approximated function at this point are equal. 
Then, we can approximate $\mathcal{F}_{k,\eta}^{\sf{f,rq}}(d_k)$ as
$
\hat{\mathcal{F}}_{k,\eta}^{\sf{f,rq}}\big(V_{k}, U_{k}\big) {=}  \sum_{l=0}^{L-1}  \left( v_{k,l} A_{k,l} {+} u_{k,l} B_{k,l} \right),
$
where $V_{k} {=} \{v_{k,l}, l{=}0{:}L{-}1\}$,  $U_{k} {=} \{u_{k,l}, l{=}0{:}L{-}1\}$, $A_{k,l} {=} {\left(\mathcal{F}_{k,\eta}^{\sf{f,rq}}(d_{k,l+1}) {-} \mathcal{F}_{k,\eta}^{\sf{f,rq}}(d_{k,l})\right)}/{\left(d_{k,l+1}{-}d_{k,l}\right)}$,  $B_{k,l} {=} \mathcal{F}_{k,\eta}^{\sf{f,rq}}(d_{k,l}) {-} A_{k,l}d_{k,l}$, and continuous variable $v_{k,l}$ and binary variable $u_{k,l}$ satisfy the following constraints:
\begin{equation}
s_k^{\sf m} = \sum\nolimits_{l=0}^{L-1} u_{k,l} \leq 1, \; \forall k\in \mathcal{B}, \text{ and } u_{k,l} d_{k,l} \leq v_{k,l} \leq u_{k,l+1} d_{k,l+1}, \; \forall k\in \mathcal{B},\; l{=}0{:}L{-}1. \label{eq36} 
\end{equation}
Then, the allocated backhaul resources due to user $k$ in \textbf{\textit{Mode 3}} are rewritten as $s_k^{\sf m}d_k{ = }\sum_{l=0}^{L-1} v_{k,l}$. Therefore, problem $(\mathcal{P}_{\sf{FV},\eta})$, which is used to determine the minimum total required fog computing resources for all users, is modified  in this extended case as follows
\begin{eqnarray}
\begin{aligned}
(\mathcal{P}_{\sf{FV},\eta}^{\sf{PLA}})  \;  &   \min\limits_{\check{\Omega}_3}  \hat{\mathcal{G}}_{\mathcal{B},\eta}^{\sf{PLA}} \Big(\check{\Omega}_3\Big) =  \sum\nolimits_{k \in \mathcal{B}} \left( s_k^{\sf{f}} f_{k}^{\sf{f,rq}} + \hat{\mathcal{F}}_{k,\eta}^{\sf{f,rq}}\left(V_{k}, U_{k}\right) \right) \nonumber \\[-0 pt]
\text{s.t.} \;
& (\check{\text{C}}3)^{\sf PLA}: s_{k}^{\sf{f}}, s_k^{\sf{c}}  , u_{k,l}\in \{0,1\}, \; \forall k,l;  \quad (\check{\text{C}}8\text{a})^{\sf PLA}: u_{k,l} d_{k,l} {\leq} v_{k,l} {\leq} u_{k,l+1} d_{k,l+1}, \; \forall k,l; \\
& (\check{\text{C}}4)^{\sf PLA}: s_{k}^{\sf{f}} + s_{k}^{\sf{c}} + \sum\nolimits_{l=0}^{L-1} u_{k,l} =1;   \quad (\check{\text{C}}8\text{b})^{\sf PLA}: \sum \nolimits_{k \in \mathcal{B}} \left( \sum\nolimits_{l=0}^{L-1} v_{k,l} + s_k^{\sf c}d_k^{\sf rq}  \right)\leq D^{\sf{max}}, \\[-0 pt]
\end{aligned}
\label{p511}	
\end{eqnarray}
where $\check{\Omega}_3 = \cup_{k \in \mathcal{B}} \big(s_k^{\sf{f}}\cup s_k^{\sf{c}} \cup U_{k} \cup V_{k}\big)$, constraints $(\check{\text{C}3})^{\sf PLA}$, $(\check{\text{C}4})^{\sf PLA}$, and $(\check{\text{C}}8\text{a})^{\sf PLA}{-}(\check{\text{C}}8\text{b})^{\sf PLA}$ are the transformed constraints of original constraints $(\check{\text{C}}3)$, $(\check{\text{C}}4)$, and $(\text{C}8)$, respectively.
This transformed problem is an MILP problem, which can be solved effectively by using the CVX solver.
The PLA~based~algorithm for verifying the feasibility of $(\mathcal{P}_{\mathcal{B}}^{\sf{ext}})$ is summarized in \textbf{Algorithm~\ref{alg_PLA}}, which can be integrated into \textbf{Algorithm~\ref{alg1}} to solve $(\mathcal{P}_2^{\sf{ext}})$. It is noted that if the value of $\mathcal{F}_{k,\eta}^{\sf{f,rq}}(d_{k,l})$ is unbounded for a given $d_{k,l}$, this infeasible point is removed when applying the PLA~based~algorithm.

\subsection{Two-stage Solution Approach (TSA)} \label{st5b}

In this section, two two-stage algorithms are developed by exploiting the fact that the decompression computational load (and therefore, the associated energy consumption) is almost independent from the compression ratio as can be seen in Fig.~\ref{P2-fig1}.
This implies that for a given $\eta$, the optimal values $f_k^{\sf{u}}$, $\omega_{k}^{\sf{u}}$, $p_k$, and $\rho_k$ for mobile user $k$ are 
similar for both $s_k^{\sf{f}}=1$ and $s_k^{\sf{c}}=1$.
Hence, in the first stage, after solving $(\mathcal{P}_3)_k$ and $(\mathcal{P}_4)_k$, $\forall k \in \mathcal{B}$, introduced in \textbf{Section~\ref{st3}}, we can set these variables to the corresponding optimal solution of $(\mathcal{P}_3)_k$, denoted as $f_{k,1}^{\sf{u}\star}$, $\omega_{k,1}^{\sf{u}\star}$, $p_{k,1}^\star$, and $\rho_{k,1}^\star$. In the second stage, we  find the remaining variables pertaining to the fog server $\Omega_4 = \cup_{k\in \mathcal{B}}\{s_k^{\sf{f}}, s_k^{\sf{c}}, s_k^{\sf{m}}, d_k, f_k^{\sf{f}}, \omega_k^{\sf{f}}\}$ 
by solving the following problem\footnote{We note that by reducing the number of optimization variables in $(\mathcal{P}_{\sf{FV},\eta}^{\sf{TSA}})$, the complexity of the resulting  
algorithms for feasibility verification of $(\mathcal{P}_{\mathcal{B}}^{\sf{ext}})$ 
	is lower than that of the PLA based algorithm.}:
\begin{eqnarray}
\begin{aligned}
(\mathcal{P}_{\sf{FV},\eta}^{\sf{TSA}})  \; \; &  \min\limits_{\Omega_4} \; \hat{\mathcal{G}}_{\mathcal{B},\eta}^{\sf{TSA}}(\Omega_4)=  \sum_{k \in \mathcal{B}} \left( s_{k}^{\sf{m}} f_k^{\sf{f}} + s_{k}^{\sf{f}} f_{k}^{\sf{f,rq}} \right) \nonumber \\[-0 pt]
\text{s.t.} \quad
&(\check{\text{C}}0\&9): 
s_{k}^{\sf{m}} \Big(\frac{ b_k^{\sf{out,f}}}{d_{k}} 
{+}  \frac{(c_{k}^{\sf{co},\sf{f}} {\plus} c_{k}^{\sf{de,u}})}{f_{k}^{\sf{f}}}\Big) {\leq} \nu_{k,0} , \\[-0 pt]
&(\check{\text{C}}8): \sum\nolimits_{k \in \mathcal{B}} \left( s_k^{\sf{m}}d_k + s_k^{\sf{c}}d_k^{\sf{rq}} \right) \leq D^{\sf{max}}, 
(\check{\text{C}}3), (\check{\text{C}}4), (\check{\text{C}}10),   \\[-0 pt]   
\end{aligned}
\label{p61}	
\end{eqnarray}
where $\nu_{k,0} = \min\{ (\eta-\Xi_{k,1})/w_k^{\sf{T}}, T_k^{\sf{max}}-T_{k,1}\} + {(c_{k,1} {+} c_{k}^{\sf{de}})}/{f_{k}^{\sf{f,rq}}} - T^{\sf{c}}$, and $\Xi_{k,1}$ and $T_{k,1}$ are the optimal values of $\Xi_k$ and $T_k$ in $(\mathcal{P}_3)_k$, respectively; $(\check{\text{C}}0\&9)$ is determined by the time delay constraint  as $\check{T}_k \leq \min(T_k^{\sf max}, (\eta - w_k^{\sf E}\xi_k)/w_k^{\sf T})$ which is
 equivalent to constraints $(\check{\text{C}}0)$ and $(\check{\text{C}}9)$. This constraint captures the fact that the application should be offloaded to
the cloud server if the resulting WEDC is smaller than that achieved when the application is executed at the fog server and the delay 
constraint $(\check{\text{C}}9)$ is not violated. 
Because $(\mathcal{P}_{\sf{FV},\eta}^{\sf{TSA}})$ is a difficult MINLP problem, we tackle it by reducing the set of variables based on the results in the following three propositions. In particular, \textbf{Propositions~\ref{prop3}--\ref{prop5}} are introduced to respectively rewrite variables $f_k^{\sf{f}}$, $\omega_k^{\sf{f}}$, and $d_k$, for all $k$ as functions of the remaining variables. Subsequently, two algorithms are proposed to solve for the remaining variables, one based on a one-dimensional search of the Lagrangian multiplier, and the other one based on an iterative update of the Lagrangian multiplier.
\begin{proposition} \label{prop3}
For any value of $d_k$'s satisfying $(\check{\text{C}}8)$, the optimal solution of $f_k^{\sf{f}}$ in $(\mathcal{P}_{\sf{FV},\eta}^{\sf{TSA}})$ can be determined as
$
 f_k^{\sf{f}\star} {=} s_k^{\sf{m}}  \frac{ (c_{k}^{\sf{co,f}} {\plus} c_{k}^{\sf{de,u}})}{\nu_{k,0} - { b_k^{\sf{out,f}}}/d_{k} }  {=} s_k^{\sf{m}}  \mathcal{H}_0 \big(\omega_{k}^{\sf{f}}, d_k\big),
$
where $ \mathcal{H}_0 \big(\omega_{k}^{\sf{f}}, d_k\big) {=}   \frac{\omega_k^{\sf{f}} d_{k}                                       \big[ \tilde{\gamma}_{k,1}^{\sf{co,f}} (\omega_k^{\sf{f}})^{\gamma_{k,2}^{\sf{co,f}}} + \tilde{\gamma}_{k,3}^{\sf{co,f}} \big]  }{\nu_{k,0} \omega_k^{\sf{f}} d_{k} - b_k^{\sf{in}} }$, $\tilde{\gamma}_{k,1}^{\sf{co,f}} =  \gamma_{k,0}^{\sf{f}} \gamma_{k,1}^{\sf{co,f}}$, and $\tilde{\gamma}_{k,3}^{\sf{co,f}} = \gamma_{k,0}^{\sf{f}} \gamma_{k,3}^{\sf{co,f}}                   {\plus} c_{k}^{\sf{de,u}}$.
\end{proposition}
\begin{proof}
When $s_k^{\sf{m}}{=}1$, the left-hand side of $(\check{\text{C}}1\&9)$ is inversely proportional to $f_k^{\sf{f}}$; thus, $f_k^{\sf{f}}$ is minimized if users spend the maximum possible resources. 
\end{proof}

\begin{proposition} \label{prop4}
When $s_k^{\sf{m}}=1$ and $d_k\geq \bar{d}_{k,1}$, the optimal value of  $\omega_{k}^{\sf{f}}$, denoted as $\omega_{k}^{{\sf{f}}\star}$, is  given as follows:
\begin{equation}
\omega_{k}^{{\sf{f}}\star} = 
\begin{cases} 
\omega_k^{\sf{max},\sf{f}} , \text{ if }\gamma_{k,2}^{\sf{co},\sf{f}} {\leq} 0 \cup \{\gamma_{k,2}^{\sf{co},\sf{f}} {\geq} 0, \bar{d}_{k,1} {<} d_k {\leq} \bar{d}_{k,2}\}, \\[-2pt]
\textup{inv}\big(\mathcal{H}_1 \big(d_k \big)\big) , \quad \text{if} \hspace{0.05 cm} \gamma_{k,2}^{\sf{co},\sf{f}} \geq 0, \bar{d}_{k,2} {<} d_k {\leq} \bar{d}_{k,3} , \\[-2pt]
\omega_k^{\sf{f,min}} , \hspace{0.5 cm} \text{if} \hspace{0.05 cm}  \gamma_{k,2}^{\sf{co},\sf{f}} \geq 0, d_k > \bar{d}_{k,3},
\end{cases} 
\label{eq54}
\end{equation}	
where $\bar{d}_{k,1} = {b_k^{\sf{in}}}/{(\nu_{k,0}\omega_{k}^{\sf{f}})}$, $\bar{d}_{k,2}   = \mathcal{H}_1 \big(\omega_{k}^{\sf{max},\sf{f}}\big)$,  $ \bar{d}_{k,3}= \mathcal{H}_1 \big(\omega_{k}^{\sf{f,min}}\big)$, and $\textup{inv}\big(\mathcal{H}_1 \big(d_k \big)\big)$ is the value of $\omega_{k}^{\sf{f}}$ for which
 $\mathcal{H}_1 \big(\omega_k^{\sf{f}} \big)$ is equal to $d_k$, and  $\mathcal{H}_1 \big(\omega_k^{\sf{f}} \big) \overset{\tinyDelta}{=} \frac{\tilde{\gamma}_{k,1}^{\sf{co},\sf{f}} b_k^{\sf{in}}(\gamma_{k,2}^{\sf{co},\sf{f}} {\plus} 1) {\big(\omega_{k}^{\sf{f}}\big)}^{\gamma_{k,2}^{\sf{co},\sf{f}}} {\plus} \tilde{\gamma}_{k,3}^{\sf{co},\sf{f}} b_k^{\sf{in}} }{\tilde{\gamma}_{k,1}^{\sf{co},\sf{f}} \nu_{k,0} \gamma_{k,2}^{\sf{co},\sf{f}} {\big(\omega_{k}^{\sf{f}}\big)}^{\gamma_{k,2}^{\sf{co},\sf{f}}+1} }.  $ 
\end{proposition}

\begin{IEEEproof}
	The proof is given in \textbf{Appendix~\ref{prf_prop4}}.
\end{IEEEproof}
Based on the results in \textbf{Propositions~\ref{prop3}} and \textbf{\ref{prop4}}, $(\mathcal{P}_{\sf{FV},\eta}^{\sf{TSA}})$ is equivalent to the
following problem:
\begin{eqnarray}
\begin{aligned}
(\mathcal{P}_{\sf{FV},\eta}^{\sf{TSAeq}})  \; \;   \min\limits_{\tilde{\Omega}_4 } \; \sum_{k \in \mathcal{B}} \big[ s_{k}^{\sf{m}}  \mathcal{H}_0\big(\omega_{k}^{{\sf{f}}\star}, d_k\big) + s_{k}^{\sf{f}} f_{k}^{\sf{f,rq}} \big]  \quad 
\text{s.t} \quad
 (\check{\text{C}}3),  (\check{\text{C}}4),  (\check{\text{C}}8),  \nonumber   
\end{aligned}
\label{p62}	
\end{eqnarray}
where $\tilde{\Omega}_4 = \cup_{k\in \mathcal{B}}\{s_k^{{\sf{c}}}, s_k^{{\sf{f}}}, s_k^{{\sf{m}}}, d_k\}$.

\begin{proposition} \label{prop5}
	The optimal value of $d_k$ for $(\mathcal{P}_{\sf{FV},\eta}^{\sf{TSAeq}})$, denoted as $d_k^\star$, is given as follows:	
	\begin{equation}
	d_k^\star = \begin{cases}
	0, \quad \text{if} \quad  s_k^{{\sf{f}}\star}=1, \\
	d_{k}^{\sf{rq}}, \quad \text{if} \quad  s_k^{{\sf{c}}\star}=1, \\
	\Big\{d_{k,\lambda} \Big| \Big(\frac{\partial \mathcal{H}_0(\omega_{k}^{{\sf{f}}\star}, d_k)}{\partial d_k}\big|_{d_k = d_{k,\lambda}} \Big) + \lambda = 0 \Big\}, \text{otherwise},
	\end{cases}
	\label{eq41}
	\end{equation}
	where $\lambda$ is the Lagrange multiplier of constraint $(\check{\text{C}}8)$. 
\end{proposition}

\begin{IEEEproof} 
The Lagrangian of problem $(\mathcal{P}_{\sf{FV},\eta}^{\sf{TSAeq}})$ can be expressed as
\begin{equation}
\mathcal{L}(\tilde{\Omega}_4,\lambda) = \sum_{k \in \mathcal{B}} \big[ s_{k}^{\sf{m}} \mathcal{H}_0\big(\omega_{k}^{{\sf{f}}\star}, d_k\big) + s_{k}^{\sf{f}} f_{k}^{\sf{f,rq}} \big] {+} 
 \lambda \big( \sum_{k \in \mathcal{B}} \big[ s_k^{\sf{m}} d_k + (1-s_k^{\sf{f}}-s_k^{\sf{m}})d_{k}^{\sf{rq}} \big]- D^{\sf{max}} \big). \nonumber
\end{equation}
When $s_k^{{\sf{m}} \star}=1$, the necessary conditions for the optimal solution $ f_k^{{\sf{f}}\star}, d_k^\star$ can be obtained by setting the derivatives of $\mathcal{L}$ 
with respect to these variables equal to zero as follows:
\begin{eqnarray}
&&\frac{\partial \mathcal{L} }{\partial d_k} = s_k^{\sf{m}} \Big( \frac{\partial \mathcal{H}_0\big(\omega_{k}^{{\sf{f}}\star}, d_k\big)}{\partial d_k} + \lambda \Big)=0, \label{eq39} \\
&&\lambda\left(\sum_{k \in \mathcal{B}} \big[ s_k^{\sf{m}}d_k + (1-s_k^{\sf{f}}-s_k^{\sf{m}})d_{k}^{\sf{rq}} \big]{-} D^{\sf{max}}\right) =0. \label{eq40}
\end{eqnarray}
Based on (\ref{eq39}), it can be verified that  $d_k^\star$ can be expressed as in (\ref{eq41}).
\end{IEEEproof}

\renewcommand{\baselinestretch}{1.2}
\setlength{\textfloatsep}{6 pt}
\begin{algorithm}[t]
	\footnotesize
	\caption{One-dimensional Search Based Feasibility Verification for $(\mathcal{P}_{\mathcal{B}}^{\sf{ext}})$} 
	\label{alg4}
	\begin{algorithmic}[1]
		\State \textbf{initialize}: $\Delta_\lambda$, $\lambda = 0$, Assign $(\mathcal{P}_{\mathcal{B}}^{\sf{ext}})$ is infeasible.
		\parState {Define $f_{k}^{\sf{f,rq}}$ and $d_k^{\sf{rq}}$ for all $k$ as in Step 2 and Step 3 of \textbf{Algorithm~\ref{alg2}}.}
		\Repeat
		\State  Assign $\lambda = \lambda + \Delta_\lambda$. Compute $d_{k,\lambda}$ as in (\ref{eq41}) and  solve $(\mathcal{P}_{\sf{FV},\eta}^{\sf{OSTS}})_{\lambda}$ to find $\tilde{\mathcal{G}}_{\mathcal{B},\eta}^{\sf{OSTS}}(\lambda)$.
		\If {$\tilde{\mathcal{G}}_{\mathcal{B},\eta}^{\sf{OSTS}}(\lambda) \leq F^{\sf{f,max}}$}
		Return $(\mathcal{P}_{\mathcal{B}}^{\sf{ext}})$ is feasible; \textbf{break}
		\EndIf
		\Until {$\lambda$ = $\lambda^{\sf{max}}$}
	\end{algorithmic}
\end{algorithm}
\renewcommand{\baselinestretch}{1.45}

\begin{lemma} \label{lemma_4}
	The gradient ${\partial \mathcal{H}_0(\omega_{k}^{{\sf{f}}\star}, d_k)}/{\partial d_k}$ is a monotonically increasing function of $d_k$.
\end{lemma}
\begin{IEEEproof} The proof is given in \textbf{Appendix~\ref{prf_lemma_4}} .
\end{IEEEproof}
As can be verified, if ${\partial \mathcal{H}_0(\omega_{k}^{{\sf{f}}\star}, d_k)}/{\partial d_k} \big|_{d_k = \bar{d}_{k,1}} +\lambda> 0$, then $d_k^\star =d_{k,\lambda} = 0, s_k^{{\sf{f}}\star} = 1$ will be the optimal solution. When $s_k^{{\sf{m}} \star}=1$, $\lambda$ must be positive because ${\partial \mathcal{H}_0(\omega_{k}^{{\sf{f}}\star}, d_k)}/{\partial d_k}$ is negative for all $d_k$. 
With the results in \textbf{Lemma~\ref{lemma_4}}, we can conclude that for a given $\lambda$, there exists at most one
 value of $d_k$ satisfying ${\partial \mathcal{H}_0(\omega_{k}^{{\sf{f}}\star}, d_k)}/{\partial d_k}+ \lambda = 0 $. This means if the optimal $\lambda$ is known, problem $(\mathcal{P}_{\sf{FV},\eta}^{\sf{TSAeq}})$ can be solved effectively. Therefore, as described in the following, to solve $(\mathcal{P}_{\sf{FV},\eta}^{\sf{TSAeq}})$, we propose two algorithms: one is based on a one-dimensional search for $\lambda$, and the other one is based on  iterative 
updating $\lambda$.

\subsubsection{One-dimensional $\lambda$-search based two-stage algorithm (OSTS Alg.)} \label{P61}
For a given $\lambda$, 
suppose that $d_{k,\lambda}$ satisfies ${\partial \mathcal{H}_0(\omega_{k}^{{\sf{f}}\star}, d_k)}/{\partial d_k}\big|_{d_k = d_{k,\lambda}}  + \lambda = 0$.
 By defining $f_{k,\lambda} = \mathcal{H}_0\big(\omega_{k}^{{\sf{f}}\star}, d_k\big) \big|_{d_k = d_{k,\lambda}}$, $\mu_{k,\lambda} = s_k^{\sf{m}}$, $\mu_{k,\lambda} = 1-s_k^{\sf{c}}$, and $\mu_{k,\lambda}= s_k^{\sf{c}}(1-x_k)$, we can 
find the optimal solution of $\cup_{k\in \mathcal{B}}\{s_k^{{\sf{c}}}, x_k, d_k\}$ by solving the following problem:
\begin{eqnarray}
\begin{aligned}
(\mathcal{P}_{\sf{FV},\eta}^{\sf{OSTS}})_{\lambda}  \quad &  \tilde{\mathcal{G}}_{\mathcal{B},\eta}^{\sf{OSTS}}(\lambda) = \min\limits_{\cup_{k \in \mathcal{B}}s_{k,\lambda}} \; \sum_{k \in \mathcal{B}} \big[ s_{k,\lambda}^{\sf m} f_{k,\lambda}  + s_{k,\lambda}^{\sf f} f_{k}^{\sf{f,rq}} \big]  \nonumber \\[-0 pt]
\text{s.t.} \quad
& (\check{\text{C}}8)_{\lambda}: \sum_{k \in \mathcal{B}}  s_{k,\lambda}^{\sf m} d_{k,\lambda} + (1-s_{k,\lambda}^{\sf f}-s_{k,\lambda}^{\sf m}) d_{k}^{\sf{rq}} \leq D^{\sf{max}},  \quad \{s_{k,\lambda}^{\sf m},s_{k,\lambda}^{\sf f} \}\in \{0,1\},
\end{aligned}
\label{p63}	
\end{eqnarray}
where $s_{k,\lambda} {=} \{s_{k,\lambda}^{\sf f}, s_{k,\lambda}^{\sf m} \}$. The above transformed problem is an integer linear programming (ILP) problem, which can be solved effectively by CVX. Let $\tilde{\mathcal{G}}_{\mathcal{B},\eta}^{\sf{OSTS}}(\lambda)$ be the optimum of $(\mathcal{P}_{\sf{FV},\eta}^{\sf{OSTS}})_{\lambda} $, then we can find the optimum of  
$(\mathcal{P}_{\sf{FV},\eta}^{\sf{TSAeq}})$ as $\tilde{G}_{\mathcal{B},\eta}^{\sf{OSTS}\star} {=} \min_{\lambda} \tilde{\mathcal{G}}_{\mathcal{B},\eta}^{\sf{OSTS}}(\lambda) $. Moreover, it can be shown that when we increase $\lambda$, all $d_{k,\lambda}$ will decrease. Therefore, the maximum value of $\lambda$ is $\lambda^{\sf{max}}$ satisfying $\mathcal{H}_0(\omega_{k}^{\sf{f}}, d_{k,\lambda^{\sf{max}}}) \geq f_{k}^{\sf{f,rq}}$, $\forall k \in \mathcal{B}$ and  $\sum_{k \in \mathcal{B}} d_{k,\lambda^{\sf{max}}} \leq D^{\sf{max}}$. Note that we can stop the search process when there exists a $\lambda$ such that $\tilde{\mathcal{G}}_{\mathcal{B},\eta}^{\sf{OSTS}}(\lambda) \leq F^{\sf{f,max}}$. When the bisection search for $\eta$ converges, we can find the optimum $\lambda^\star = \argmin_{\lambda} \tilde{\mathcal{G}}_{\mathcal{B},\eta}^{\sf{OSTS}}(\lambda)$, and the optimal variables $s_{k}^{\sf m \star}=s_{k,\lambda^{\star}}^{\sf m}$, $s_{k}^{\sf f \star}=s_{k,\lambda^{\star}}^{\sf f}$, $s_{k}^{\sf c \star}=1-s_{k}^{\sf m \star}-s_{k}^{\sf f \star}$, 
$
f_k^{{\sf{f}}\star} =  s_{k,\lambda^{\star}}^{\sf m} f_{k,\lambda^\star}  + s_{k,\lambda^{\star}}^{\sf f} f_{k}^{\sf f,rq}$, and $d_k^{\star}=s_{k,\lambda^{\star}}^{\sf m} d_{k,\lambda^{\star}} + (1-s_{k,\lambda^{\star}}^{\sf f}-s_{k,\lambda^{\star}}^{\sf m}) d_{k}^{\sf{rq}}, \forall k \in \mathcal{B}$.
The OSTS algorithm for feasibility verification of $(\mathcal{P}_{\mathcal{B}}^{\sf{ext}})$ is summarized in \textbf{Algorithm~\ref{alg4}}.

\subsubsection{Iterative $\lambda$-update based two-stage algorithm (IUTS Alg.)} \label{P62}
 This method can solve  $(\mathcal{P}_{\sf{FV},\eta}^{\sf{TSAeq}})$ with very low complexity via Lagrangian dual updates. Specifically,  the dual function of $(\mathcal{P}_{\sf{FV},\eta}^{\sf{TSAeq}})$ can be defined as $\mathcal{G}^{\sf{o}}(\lambda) = \min_{\tilde{\Omega}_4}  \mathcal{L}(\tilde{\Omega}_4,\lambda)$, and the dual problem can be stated as
\begin{eqnarray}
\max_{\lambda} \mathcal{G}^{\sf{o}}(\lambda)  \text{  s.t.  } \lambda \geq 0.
\end{eqnarray}

Since the dual problem is always convex, $\mathcal{G}^{\sf{o}}(\lambda)$ can be maximized by using the standard sub-gradient 
method where the dual variable $\lambda$ is iteratively updated as follows:
\begin{equation}
\lambda_{n} = \Big[ \lambda_{n-1} + \delta_{n} \big( \sum\nolimits_{k \in \mathcal{B}} \left( s_{k,\lambda_{n-1}}^{\sf m}d_{k,\lambda_{n-1}} + s_{k,\lambda_{n-1}}^{\sf c}d_{k}^{\sf{rq}} \right)- D^{\sf{max}} \big)  \Big]^+,
\end{equation}
where $n$ denotes the iteration index, $\delta_{n}$ represents the step size, and $[a]^+$ is defined as $\max(0,a)$. The 
sub-gradient method is guaranteed to converge to the optimal value of $\lambda$ for an initial primal point $\Omega_4$ if 
the step size $\delta_{n}$ is chosen appropriately, e.g., $\delta_{n} \rightarrow 0$ when $n \rightarrow \infty$, which is met by 
setting $\delta_{n}  = 1/\sqrt{n}$.

For a given $\lambda_{n}$, we can determine the primal variable $d_{k,\lambda_{n}} = \textup{inv}(\mathcal{H}_2(\lambda_{n}))$. For given 
$\lambda_{n}$ and $d_{k,\lambda_{n}}$, the primal problem becomes a linear program in $s_{k,\lambda_n}, \forall k\in\mathcal{B}$, which can be solved effectively by using standard linear optimization techniques. Moreover, the vertices in this problem are the points where the $s_{k,\lambda_n}^{\sf m}$'s, $s_{k,\lambda_n}^{\sf f}$'s, and $s_{k,\lambda_n}^{\sf c}$'s are either $0$ or $1$. Thus, \textit{solving the relaxed
 problem will also return binary values $0$ or $1$}. However, once the $s_{k,\lambda_n}^{\sf m}$'s, $s_{k,\lambda_n}^{\sf f}$'s, and $s_{k,\lambda_n}^{\sf c}$'s take values of $0$ or $1$, the
decision on the application execution  location (fog or cloud) may be trapped at a local optimal solution
such that the required fog computing resources cannot be updated to improve the solution. To overcome this critical issue, the gradient projection method can be adopted to slowly update variables $s_{k,\lambda_n}^{\sf m}$'s, $s_{k,\lambda_n}^{\sf f}$'s, and $s_{k,\lambda_n}^{\sf c}$'s  as $
\mb{s}_k^{(n+1)} = \mathbbm{P}_{\Phi_k}\left( \mb{s}_k^{(n)} -\check{\delta} \nabla\mb{s}_k^{(n)} \right),
$
where $\mb{s}_k^{(n)} {=} [s_{k,\lambda_n}^{\sf m},s_{k,\lambda_n}^{\sf f},s_{k,\lambda_n}^{\sf c}]$, $\check{\delta}$ is the step size, 
$ \nabla\mb{s}_k^{(n)} = [\mathcal{H}_0(\omega_{k}^{\sf{f}\star},d_{k,\lambda_{n}}) + \lambda_{n}d_{k,\lambda_{n}}, 
\lambda_{n} f_{k}^{\sf{f,rq}},  \lambda_{n} d_{k}^{\sf{rq}} ]
$, and $\mathbbm{P}_{\Phi_k}(.)$ is the projection onto the set $\Phi_k = \left\lbrace \mb{s}_k | \mb{s}_k\geq 0  ,\; s_{k,\lambda_n}^{\sf f}{+}s_{k,\lambda_n}^{\sf c}{+}s_{k,\lambda_n}^{\sf m} \leq 1 \right\rbrace $.
Finally, it can be verified that this iterative mechanism always converges \cite{sanjabi2014optimal}.

\subsection{Complexity Analysis}

The overall complexity of the PLA based algorithm for solving the extended problem is $|\mathcal{B}|\mathcal{O}((\mathcal{P}_3)_k) + L|\mathcal{B}|\mathcal{O}((\mathcal{P}_4)_k) + \mathcal{O}(\mathcal{P}_{\sf{FV},\eta}^{\sf{PLA}})$. Moreover,
 $(\mathcal{P}_{\sf{FV},\eta}^{\sf{PLA}})$ is an NP-hard problem, solving it via an optimal exhaustive search entails the complexity
 $\mathcal{O}(2^{{(l+1)}^{|\mathcal{B}|}})$.
 

The proposed two-stage IUTS and OSTS based algorithms have the overall complexity of $|\mathcal{B}|\mathcal{O}((\mathcal{P}_3)_k) + |\mathcal{B}|\mathcal{O}((\mathcal{P}_4)_k) + \mathcal{O}(\mathcal{P}_{\sf{FV},\eta}^{\sf{TSA}})$.
In \textbf{Section~\ref{P61}}, problem  $((\mathcal{P}_{\sf{FV},\eta}^{\sf{OSTS}})_{\lambda})$ can be transformed to the standard knapsack problem  as in \cite{martello1990knapsack}, while the optimal $d_{k}$ and $\omega_{k}$ can be computed directly for a given value of $\lambda$. Therefore, the complexity of \textbf{Algorithm \ref{alg4}} to solve $(\mathcal{P}_{\sf{FV},\eta}^{\sf{TSA}})$ by the OSTS method is $\mathcal{O}(\frac{\lambda^{\sf{max}}}{\Delta_\lambda} \nu_2|2\mathcal{B}|)$. For the IUTS based algorithm presented in \textbf{Section~\ref{P62}}, we can directly update $\lambda_{n}, d_{k,\lambda_{n}}, \mu_{k,i,\lambda_n}, \forall i,k,n$; which means that $(\mathcal{P}_{\sf{FV},\eta}^{\sf{TSA}})$ has a complexity of $\mathcal{O}(N)$, where $N$ is the number of iterations. It is worth noting that  $\mathcal{O}((\mathcal{P}_3)_k)$ and $\mathcal{O}((\mathcal{P}_4)_k)$ are given in \textbf{Section \ref{sec:ComAna}}.


\section{Numerical Results} \label{st5}
\subsection{Simulation Setup}
\setlength{\textfloatsep}{5 pt}
\small{
\begin{table}[htb]
	\caption{Simulation Parameter Settings}
	\centering
	\label{J2t:observed_para1}
	\begin{tabular}{|l|l||l|l|}
		\hline
		Parameters & Setting& Parameters & Setting \\
		\hline
		\hline
		Path loss, $\beta_k$ & $128.1{+}37.6 \log_{10} (dist_k (\text{km})) $  
		&
		Cell radius & $800$ meters  \\
		\hline
		Noise power density, $\sigma_{\sf bs}$  & $3.18 \times 10^{-20} $ W/Hz
		&
		Number of users $K$ & $10$ \\
		\hline
		Beamforming gain $M_0$ & $5$ 
		&
		Max. transmission bandwidth $\rho_k^{\sf max}$ & $1$ MHz \\
		\hline		
		Max. delay time $T_k^{\sf max}$& $1$ second 
		&
		Max. clock speed $F_k^{\sf max}$& $2.4$ GHz \\
		\hline
		Max. transmit power $P_k^{\sf max}$& $0.22$ W 
		&
		Circuit power $p_{k,0}$ & $22$ nW/Hz \\
		\hline
		User computation demand  & $c_k \in [1.8-2.4]$ Gcycles 
		&
		Offloadable load & $c_{k,1} = 0.9 c_k$ \\
		\hline
		Energy coefficient $\alpha_{k}$ & $0.1 \times 10^{-27}$ 
		&
		Time $T^{\sf c}$ & $0.2$ second \\
		\hline
		Raw data size $b_k^{\sf in}$ & $4$ Mbits 
		&
		Max. fog computing resource $F^{\sf f,max}$ & $15$ GHz \\
		\hline
		Max. backhaul capacity $D^{\sf max}$ & $20$ Mbps 
		&
		User compression ratio range: $\omega_{k}^{\sf u}$ & $[2.3, 2.9]$\\
		\hline
		Coefficient $\kappa$ & $50$ 
		&
	    Fog compression ratio range: $\omega_{k}^{\sf f}$  & $[3.4, 11.2]$ \\
		\hline
	\end{tabular}
\end{table}
}
\normalsize

We consider a hierarchical fog-cloud system consisting of $K$=10 users (except for Fig.~\ref{P2-fig10}) where the users are randomly distributed in the cell coverage area with a radius of $800$ m and the BS is located at the cell center. 
Detailed simulation parameter settings are summarized in Table~\ref{J2t:observed_para1} unless otherwise stated. Particularly, the path-loss is calculated as $\beta_k (\text{dB})=128.1 + 37.6\log_{10}(dist_k)$, where $dist_k$ is the geographical distance between user $k$ and the BS (in km) \cite{3GPP1}. We further set the beamforming gain as $M_0=5$, the maximum transmission bandwidth as $\rho_k^{\sf{max}}= 1$ MHz, and the noise power density as $\sigma_{\sf{bs}} = 1.381\times 10^{-23} \times 290 \times 10^{0.9}$ W/Hz \cite{ngo2017cell}.
All users are assumed to have the same maximum clock speed of $2.4$ GHz, a maximum transmit power of $P_k^{\sf max} = 0.22$ W, and the circuit power is set to $p_{k,0} = 22$ nW/Hz.  
We assume that the number of transmission bits incurred to support computation offloading $b_k^{\sf{in}}$ is the same for all users.

Moreover, the computation demands of the 10 different users $\{c_{1}, c_{2},..., c_{9},c_{10}\}$ are set randomly in the range $1.8-2.4$ Gcycles 
while the maximum delay time is $T_k^{\sf{max}} = 1$ second, the non-offloadable load is $c_{k,0} = 0.1 c_{k}$, and the offloadable load is $c_{k,1} = 0.9 c_{k}$ for all users.
We also set the energy coefficient as $\alpha_k = 0.1 \times 10^{-27}$ and the computing time at the cloud server as $T^{\sf{c}}=T_k^{\sf{max}}/5$. 
For the DC algorithm, we set the parameters according to the top-left sub-figure in Fig.~\ref{P2-fig1} as follows: $\gamma_{k,1}^{\sf{co}} = 0.03 \times 2.6^{32.28}$, $\gamma_{k,2}^{\sf{co}}=32.28$, $\gamma_{k,3}^{\sf{co}}=0.3$,  $\gamma_{k,1}^{\sf{de}} = 0.115$, $\gamma_{k,2}^{\sf{de}}=-0.9179$, $\gamma_{k,3}^{\sf{de}}=0.046, \forall k$, $\omega_k^{\sf{u,min}} = 2.3$, and $\omega_k^{\sf{u,max}} = 2.9$. The energy and delay weights are chosen so that $w_k^{\sf{E}} + w_k^{\sf{T}}=1, \forall k$. 
Simulation results are obtained by averaging over 100 realizations of the random locations of the users. 
Finally, for all figures, we set the raw data size as $b_k^{\sf{in}} = 4$ Mbits (except for Figs.~\ref{P2-fig2}, \ref{P2-fig3} and \ref{P2-fig10}),  $w_k^{\sf{E}} = 2w_k^{\sf{T}}$, $\forall k$
(except for Fig.~\ref{P2-fig4}), the maximum fog computing resource as $F^{\sf{f,max}}=15$ GHz, the maximum backhaul capacity as $D^{\sf{max}} = 20$ Mbps (except for Figs.~\ref{P2-fig3} 
and \ref{P2-fig4}), and  $\kappa = 50$ (except for Figs.~\ref{P2-fig2} and ~\ref{P2-fig5}), where  $\kappa$ 
captures the relationship between $\gamma_{k,0}^{\sf{u}}$ in (\ref{eqcomplex}) and
the raw data size as $\gamma_{k,0}^{\sf{u}} = \kappa b_k^{\sf{in}}$ \cite{yu2006information}.

In practice, a fog server can support more powerful DC algorithms compared to the users.
This implies that the compression ratio for the fog server is much larger than that for the users. Therefore, when the fog server decompresses and re-compresses data, we set the parameters according to the top-middle sub-figure in Fig.~\ref{P2-fig1} as follows: $\gamma_{k,1}^{\sf{co,f}} = 0.076, \gamma_{k,2}^{\sf{co,f}} = 0.7116, \gamma_{k,3}^{\sf{co,f}} = 0.5794$, $\omega_{k}^{\sf{f,min}} = 3.4$ and $\omega_{k}^{\sf{f,max}} = 11.2$.  The step size is set as $\check{\delta} = 0.1$. 
For the proposed algorithms presented in \textbf{Section~\ref{st3}} and \textbf{Section~\ref{st4a}}, numerical results 
are shown in Figs.~\ref{P2-fig2}--\ref{P2-fig10} and $\text{Figs.~\ref{P2-fig7}--\ref{P2-fig11}}$, respectively.
\subsection{Results for DC at only Mobile Users}

In Fig.~\ref{P2-fig2}, we show the significant benefits of DC in computation offloading where the min-max WEDC
(called WEDC for brevity) vs. $b_k^{\sf{in}}$ is plotted  for six different schemes: the `Local-execution' scheme in which
 all users' applications are executed locally; the `Alg. in \cite{du2017computation} (w/o Comp)' scheme 
in which the benchmark algorithm in \cite{du2017computation} is applied with $\omega_{k}^{\sf u}=1, \forall k$, and no DC \footnote{As discussed in Section I, our current work is the first study on joint data compression and computation offloading in hierarchical fog-cloud systems to minimize the maximum weighted energy and service delay cost. 
Therefore, \cite{du2017computation}, a recent work on computation offloading in hierarchical fog-cloud systems without exploiting data compression, is selected as a benchmark work for comparison purpose.}; the `JCORA Alg. w/o Comp' in which the proposed JCORA algorithm is applied with $\omega_{k}^{\sf u}=1, \forall k$, and no DC (the other variables 
are optimized as in the JCORA algorithm); and three other instances of the proposed JCORA algorithm with
DC and three different values of $\kappa = 50, 100, 200$ ($\kappa = \gamma_{k,0}^{\sf{u}}/b_k^{\sf{in}}$). To guarantee 
a fair comparison between 
the `Alg. in \cite{du2017computation} (w/o Comp)' scheme and our proposed schemes, we also apply MIMO and 
optimize the offloading decision and the allocation of the fog computing resources,  transmit power, bandwidth, and local CPU clock speed 
for the `Alg. in \cite{du2017computation} (w/o Comp)' scheme. In addition, for the remaining
 variable $d_k$, we  allocate the backhaul capacity equally to the users that offload their tasks to the cloud server.

\begin{figure}[t]
	\centering
	\includegraphics[width=0.7\textwidth]{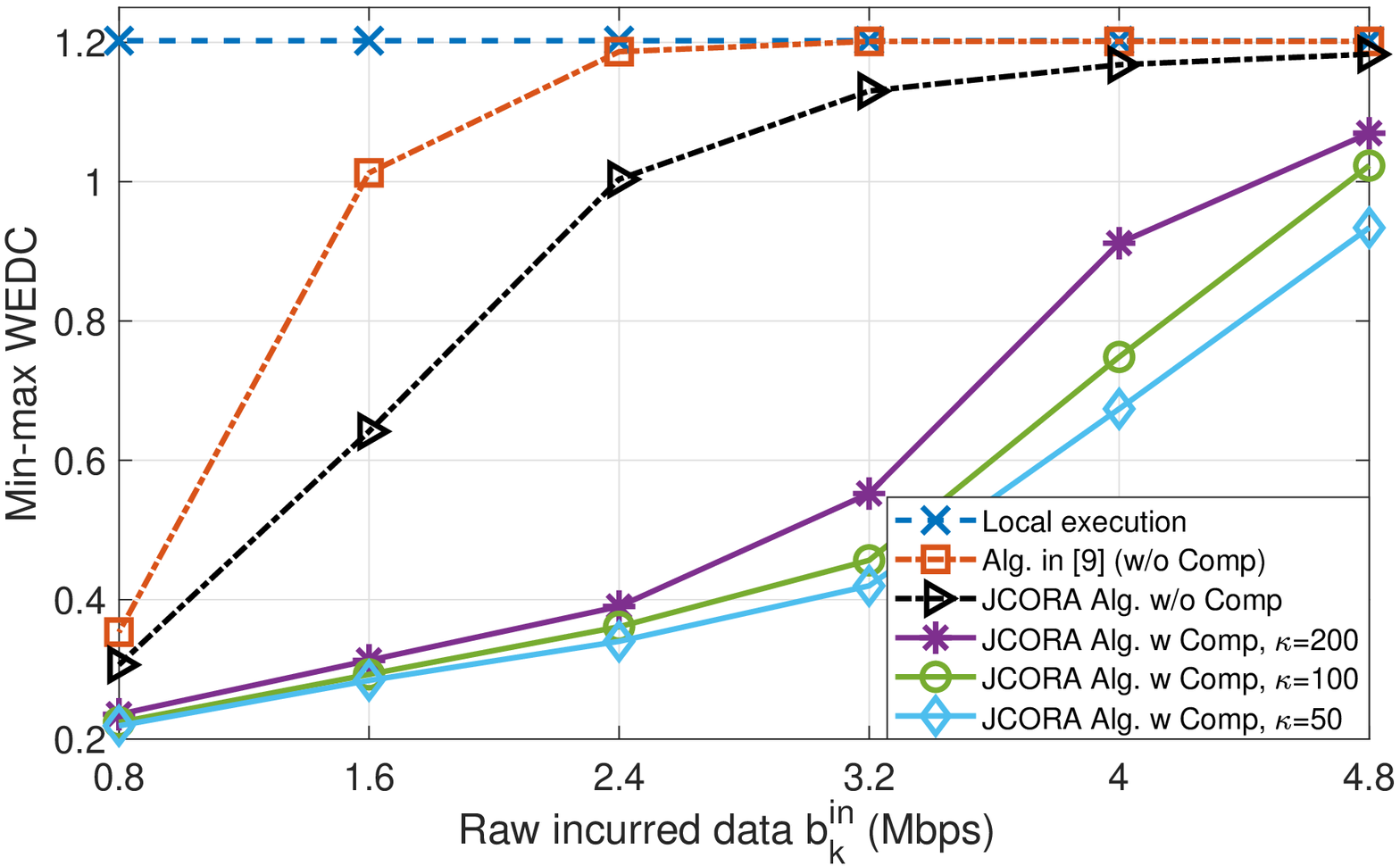}
	\caption{Min-max WEDC vs. $b_k^{\sf{in}}$.}
	\label{P2-fig2}
\end{figure}

\begin{figure}[t]
	\centering
		\includegraphics[width=0.7\textwidth]{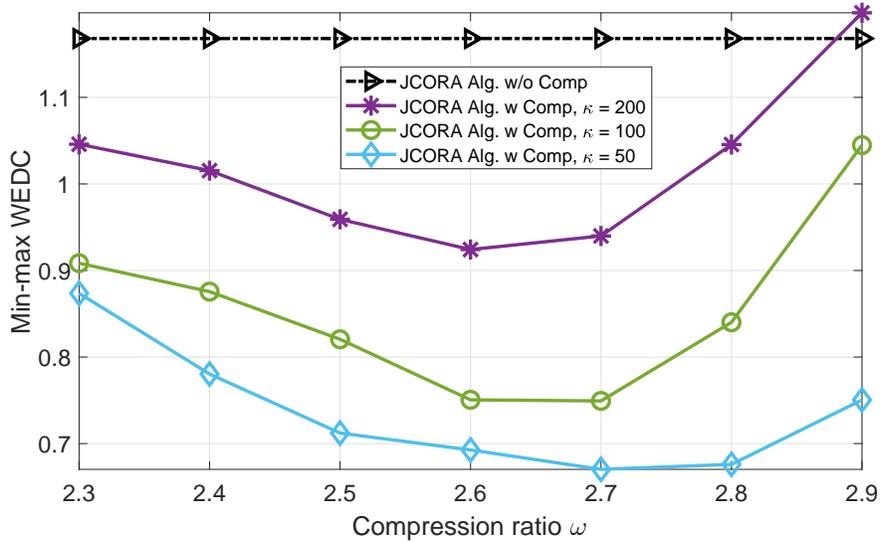}
		\caption{Min-max WEDC vs. compression ratio.}
		\label{P2-fig5}
\end{figure}

As can be observed from Fig.~\ref{P2-fig2}, computation offloading can greatly
 improve the WEDC when there are sufficient radio and computing resources to support 
the offloading (e.g., the incurred amount of data is not too large). Specifically, computation offloading even without DC can result in a significant reduction of the WEDC compared to local execution, especially when the incurred amount of data $b_k^{\sf{in}}$ is small due to the constrained radio resources. Furthermore, even without exploiting DC, our proposed algorithm (JCORA Alg. w/o Comp) 
results in much better performance than the algorithm proposed in \cite{du2017computation}. 
This is because our proposed design jointly optimizes the offloading decisions and the computing and radio resource allocation,
while in \cite{du2017computation}, the offloading decisions are found nearly independent of the computing and radio resource
allocation. In particular, the semidefinite relaxation technique employed in \cite{du2017computation} may not always
 guarantee the rank-1 condition for the optimized matrix. 
Joint optimization of DC, computation offloading, and resource allocation can lead to a further significant 
reduction of the WEDC for a larger range of $b_k^{\sf{in}}$ (e.g., when $b_k^{\sf in} = 2.4 $ Mbps, the min-max WEDC is reduced by up to 65\%). However, the energy and time consumed for (de)compression  also affect the 
achievable  min-max WEDC, and their impact tends to become stronger for larger $\gamma_{k,0}^{\sf{u}}$ and when the available radio 
resource is more limited.

In Fig.~\ref{P2-fig5}, we investigate the impact of the compression ratio on the min-max WEDC for the JCOCA scheme with and without DC for different values of $\omega_k^{\sf{u}} = \omega, \forall k$ (i.e., the compression ratio $\omega_k^{\sf{u}}$ is fixed while the remaining variables are optimized as in the JCOCA scheme). As can be seen, there is an optimal $\omega$ that achieves the minimum WEDC. Moreover, the optimal value of $\omega$ tends to decrease for increasing computational load because the optimal compression ratio has to efficiently balance the demand on the radio and computing resources. 
In fact, for the right choice of $\omega$, the \textit{``JCORA Alg. w Comp''} scheme
greatly outperforms the \textit{``JCORA Alg. w/o Comp''} scheme. 
Moreover, this figure shows that for the optimal $\omega$, 29\% reduction in the min-max WEDC   
can be achieved compared to  the worst choice of $\omega$. 

\begin{figure}[t]
	\centering
		\includegraphics[width=0.7\textwidth]{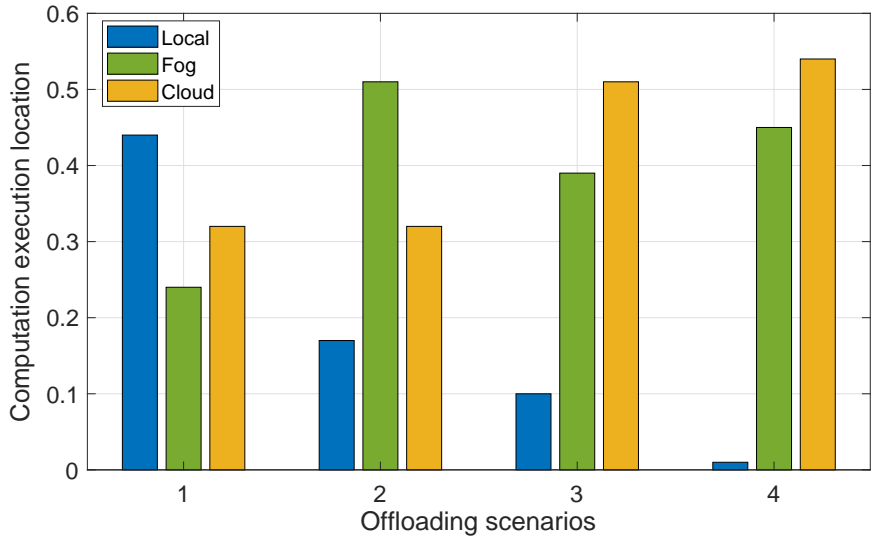}
		\caption{User, fog, and cloud computational load processing.}
		\label{P2-fig3}
\end{figure}

\begin{figure}[t]
	\centering
		\includegraphics[width=0.7\textwidth]{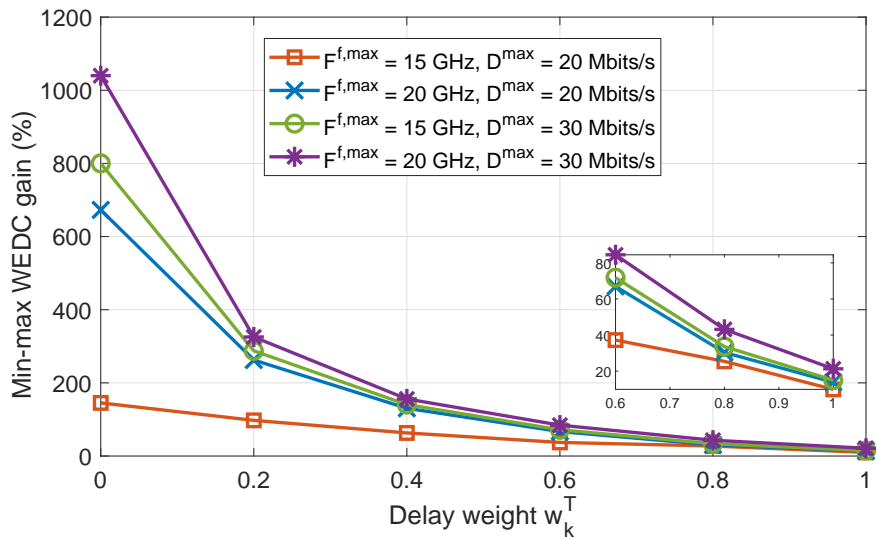}
		\caption{Min-max WEDC gain vs. delay weight.}
		\label{P2-fig4} 
\end{figure}

Fig.~\ref{P2-fig3} shows the computational loads processed locally as well as in the fog and cloud servers when $b_k^{\sf{in}}=4.8$ Mbits for four different scenarios: 
1) $F^{\sf{f,max}} = 15$ GHz, $D^{\sf{max}} = 20$ Mbps; 2) $F^{\sf{f,max}} = 20$ GHz, $D^{\sf{max}} = 20$ Mbps; 3) $F^{\sf{f,max}} = 15$ GHz, $D^{\sf{max}} = 30$ Mbps; and 4) $F^{\sf{f,max}} = 20$ GHz, $D^{\sf{max}} = 30$ Mbps. The results shown in Fig.~\ref{P2-fig3} suggest that
more of the users' computational load should be offloaded and executed at the fog and cloud servers if there are sufficient resources to support the
offloading process. Particularly, nearly all users offload their computation tasks in Scenario 4, while in Scenario 1, 
about half of the users  offload their computation demand.  



In Fig.~\ref{P2-fig4}, we show the min-max WEDC gain due to DC as a function of the delay weight $w_k^T$. 
The  min-max WEDC gain is computed as $\frac{\eta^{\sf{NoComp}\star} -\eta^{\sf{Comp}\star}}{\eta^{\sf{Comp}\star} } \times 100$ 
(\%) where $\eta^{\sf{Comp}\star}$ and $\eta^{\sf{NoComp}\star}$  denote the optimal min-max WEDCs with and without DC
under the JCORA framework. 
When energy saving is the only concern for the mobile devices ($w_k^{\sf{T}} = 0, w_k^{\sf{E}} = 1$), 
this figure confirms that JCORA with DC can save more than 170\% of energy compared with 
JCORA without DC even for the scenario with $F^{\sf{f,max}} = 15$ GHz and $D^{\sf{max}} = 20$ Mbps. 
The min-max WEDC gain decreases
when we focus more on latency (i.e., for higher delay weight $w_k^T$). Moreover, for $w_k^{\sf{T}}=1$, DC 
results in a 15\% reduction of the execution delay for $F^{\sf{f,max}} = 15$ GHz, $D^{\sf{max}} = 20$ Mbps, 
and about 25\% delay reduction for $F^{\sf{f,max}} = 20$ GHz, $D^{\sf{max}} = 30$ Mbps.

In Fig.~\ref{P2-fig10}, we show the min-max WEDC vs. the number of users in the system for $b_k^{\sf{in}} = 2.4$ Mbps, $\forall k$. When there are more users that may offload their computational loads to the fog and cloud servers, the available resources that can be allocated to each user become smaller; therefore, the min-max WEDC increases. However, the proposed JCORA
scheme still achieves the optimal performance in the multi-user hierarchical fog-cloud system. 

\begin{figure}[t]
	\centering
	\includegraphics[width=0.7\textwidth]{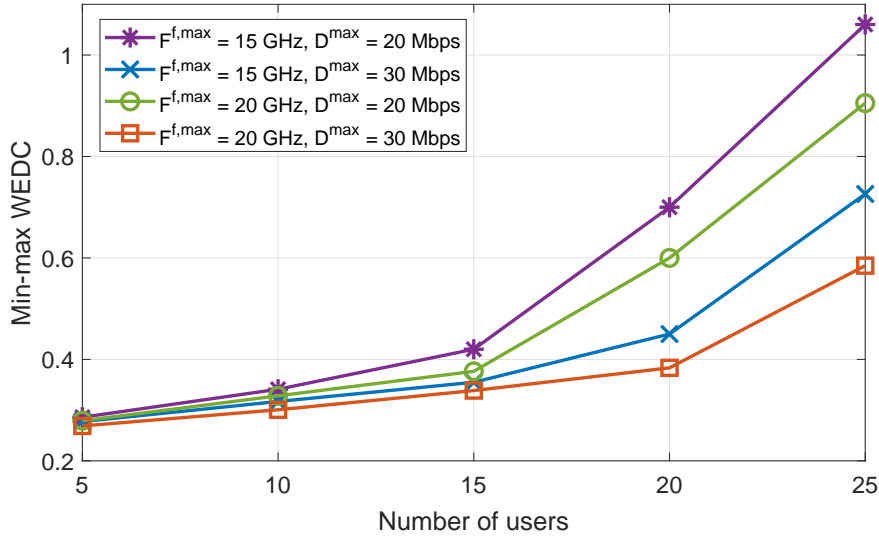}
	\caption{Min-max WEDC vs. number of users.}
	\label{P2-fig10}	
\end{figure}
\begin{figure}[t]
	\centering
	\includegraphics[width=0.7\textwidth]{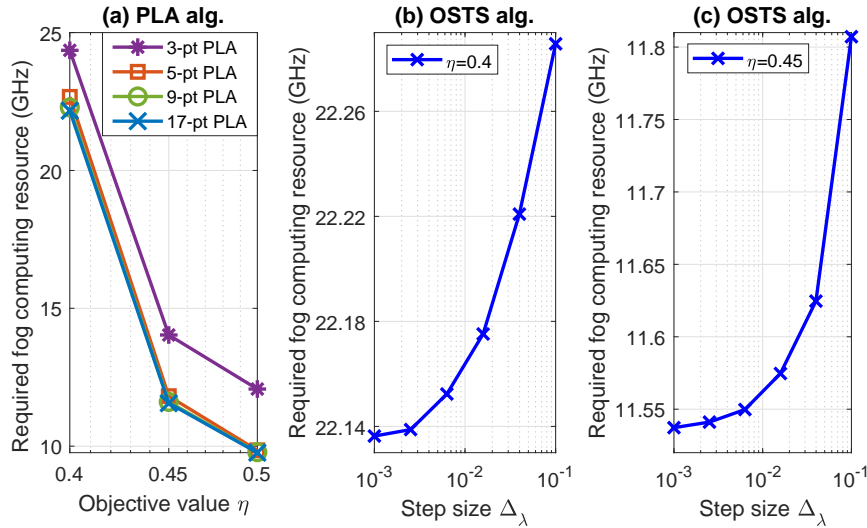}
	\caption{Accuracy of proposed PLA and OSTS algs.}
	\label{P2-fig7}
\end{figure}
	\begin{figure}[t]
	\centering
	\includegraphics[width=0.7\textwidth]{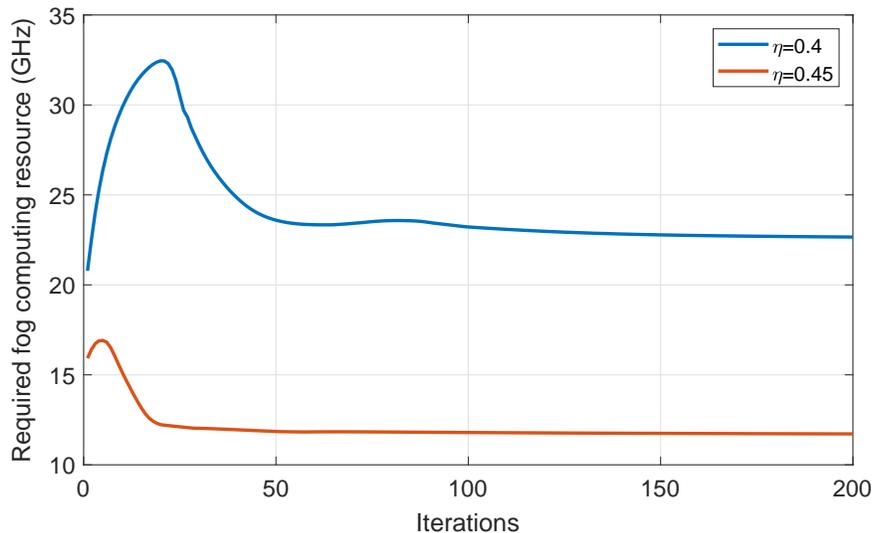}
	\caption{Convergence of proposed IUTS alg.}
	\label{P2-fig8}
\end{figure}
\subsection{Results for DC at both Mobile Users and Fog Server}

To evaluate the system performance when  DC is performed  at both the mobile users and the fog server, 
we consider the following parameter setting: $ \gamma_{k,0}^{\sf{f}} =\gamma_{k,0}^{\sf{u}}$ (except for Fig.~\ref{P2-fig11}),
$F^{\sf{f,max}} = 15$ GHz, and $D^{\sf{max}} = 20$ Mbps. 
In Fig.~\ref{P2-fig7}, we show the required computing resources for 
 the proposed PLA and OSTS algorithms when solving the extended problem. In Fig.~\ref{P2-fig7}-(a), `$n$-pt PLA' corresponds to the
$n$-point PLA method. 
In the PLA method,  when the number of points used to approximate the actual function is sufficiently large,  the difference between the actual and approximated functions becomes negligible. As shown in Fig~\ref{P2-fig7}-(a), there is only a 
small difference in the required fog computing resources when the number of points increases from $5$ 
to $9$. In addition, these required resources are nearly identical for both  the $9$-point and $17$-point curves.
Therefore, we use  `9-pt PLA' as a benchmark method to evaluate the performance of the OSTS and IUTS algorithms. 
The middle and right sub-figures illustrate the accuracy of the OSTS algorithm in solving problem $(\mathcal{P}_{\sf{FV},\eta}^{\sf{TSA}})$ vs.
the step size $\Delta_{\lambda}$. Specifically, these figures show that the value of  $G_{\mathcal{B},\eta}^{\sf{OSTS}\star}$ becomes stable when $\Delta_{\lambda}$ is about $5\times10^{-3}$. Moreover, the value of $G_{\mathcal{B},\eta}^{\sf{OSTS}\star}$ achieved with the OSTS algorithm at $\Delta_{\lambda} = 5\times10^{-3}$ is almost the same as the value of  $\hat{G}_{\mathcal{B},\eta}^{\sf{PLA}\star}$ achieved with `17-pt PLA', which means that the approximated problem 
$(\mathcal{P}_{\sf{FV},\eta}^{\sf{TSA}})$ can be used to find a close-to-optimal solution of the extended problem. Besides, the difference in $G_{\mathcal{B},\eta}^{\sf{OSTS}\star}$ for $\Delta_{\lambda} = 0.1$  and $\Delta_{\lambda} = 0.001$ is less than 2\%, which means that a large step size ($\Delta_{\lambda} = 0.1$) can be used to make the OSTS algorithm converge quickly while still guaranteeing
good system performance.

 The convergence of the proposed IUTS algorithm is illustrated in Fig.~\ref{P2-fig8}. The initial conditions are 
		set as follows: $\lambda=1$, and $s_{k,\lambda}^{\sf x}=1/3, \forall k \in \mathcal{B}, \sf{x} \in \{\sf{f,c,m}\}$.  It can be observed that the proposed IUTS algorithm converges after about 100 iterations even with small feasible set when $\eta$ closes to the optimal value.

The benefits of data re-compression at the fog are shown in Fig.~\ref{P2-fig6} where we plot the min-max WEDC vs. 
$b_k^{\sf{in}}$ for four different schemes: the `JCORA Alg. w Comp' scheme in which data are compressed only at
 users while the three remaining schemes correspond to the proposed algorithms for the extended case. In particular, `9-pt PLA Alg. w Fog Comp', `OSTS Alg. w Fog Comp', and `IUTS Alg. w Fog Comp' correspond to the 
9-point PLA, OSTS, and IUTS algorithms, respectively, which perform compression at both the users and the fog server.
For $b_k^{\sf{in}} = 4$ Mbits, an additional min-max WEDC reduction of 35\% can be achieved by performing DC at both the users and the fog server. 
Moreover, the required radio resources decrease with decreasing  $b_k^{\sf{in}}$; therefore, the gain is reduced due to the decreasing demand for data transmission.  When $b_k^{\sf{in}}$ increases, the main bottleneck for  computation offloading  are the limited radio resources available to support data transmissions between the users and the fog server; 
therefore, the gain due to data re-compression at the fog server becomes less significant. This figure 
also confirms that the `9-pt PLA', `OSTS', and `IUTS' schemes achieve almost the same min-max WEDC.

\begin{figure}[t]
	\centering
	\includegraphics[width=0.7\textwidth]{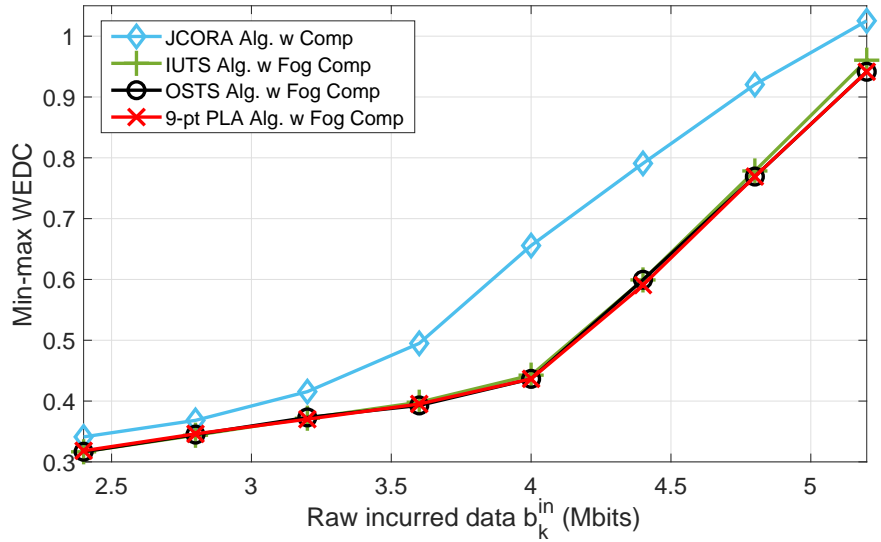}
	\caption{Min-max WEDC in general design scenario.}
	\label{P2-fig6}
\end{figure}

\begin{figure}[t]
	\centering
	\includegraphics[width=0.7\textwidth]{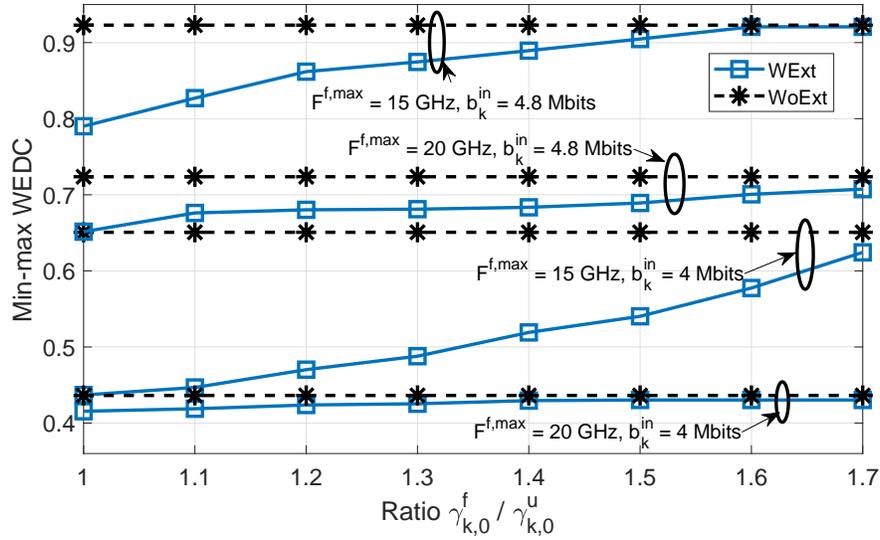}
	\caption{Min-max WEDC vs. $\gamma_{k,0}^{\sf{f}}/\gamma_{k,0}^{\sf{u}}$.}
	\label{P2-fig11}
\end{figure}

In Fig.~\ref{P2-fig11}, we plot the min-max WEDC vs. the ratio between the maximum computational loads (in CPU cycles) required 
to compress data at the fog server ($\gamma_{k,0}^{\sf{f}}$) and the user ($\gamma_{k,0}^{\sf{u}}$) 
 for different values  of $F^{\sf{f,max}}$ and $b_k^{\sf{in}}$. The  `WoExt' and `WExt' correspond to the JCORA and OSTS algorithms presented in \textbf{Section~\ref{st3}} and \textbf{\ref{st4a}}, respectively.
This figure shows that  data re-compression at the fog server can bring additional performance benefits, especially in scenarios with limited fog
 computing resources (i.e., $F^{\sf{f,max}}=15$ GHz). As the compression ratio adopted at the fog server 
could be much larger than that at the users, a better performance can be obtained by applying DC at both the users and the fog server when 
$\gamma_{k,0}^{\sf{f}}$ is not much larger than $\gamma_{k,0}^{\sf{u}}$.  
Otherwise, if the cost due to data re-compression becomes larger, 
the benefits of adopting \textbf{\textit{Mode 3}} are less significant (i.e., for $\gamma_{k,0}^{\sf{f}} = 1.7\gamma_{k,0}^{\sf{u}}$). 

\begin{remark}
\label{mrk:discuss_fig}
Although certain variation patterns or performance improvements obtained 
		by the newly proposed designs for certain parameter variations can be predicted or foreseen through careful analysis.
		However, presenting this kind of simulation results for the proposed designs does still provide valuable insights in many cases.
		In particular, these simulation results enable us to confirm/validate the superior performance and
		quantify the performance gains of the proposed algorithms compared to  state-of-the-art designs/algorithms
		for practical settings. 
\end{remark}

\section{Conclusion}
\label{st6}
In this paper, we have proposed novel and efficient algorithms for joint DC and computation offloading in
 hierarchical fog-cloud systems which minimize the weighted energy and delay cost while maintaining user fairness. 
 Specifically, we have considered the cases where DC is leveraged at only the mobile users and
 at both the mobile users and the fog server, respectively. Numerical results have confirmed the significant performance gains of 
the proposed algorithms compared to conventional schemes not using DC. 
Particularly, the following key observations can be drawn from our numerical studies: 1) Joint DC and computation offloading  can result in min-max WEDC reductions of up to 65\% compared to optimal computation
offloading without DC; 2) the proposed JCORA scheme can efficiently distribute the computational load among the mobile users, the fog server, and the cloud server and exploits the available system resources in an optimal manner;
3) when energy saving is the only concern for the mobile users, the JCORA scheme can achieve an energy saving 
gain of up to  a few hundred percent compared to optimal computation offloading  without DC; 
and 4) an additional min-max WEDC reduction of up to 35\% can be achieved by further employing DC at the fog server.
In future work, we plan to extend our designs to multi-task offloading and systems with multiple fog servers.

\appendices

\section{Proof of \textbf{Lemma~\ref{lemma_1}}}
\label{prf_lemma_1} 
The three statements in \textbf{Lemma~\ref{lemma_1}} can be proved as follows.
\begin{enumerate}
	\item As can be observed, merging the optimal solutions of $(\mathcal{P}_{\mathcal{A}})$ and $(\mathcal{P}_{\mathcal{B}})$ will result in a feasible solution of $(\mathcal{P}_2)$. In addition, the objective function corresponding to this feasible solution can be expressed as $\max(\eta_{\mathcal{A}}^{\star},\eta_{\mathcal{B}}^{\star})$. Therefore, we can conclude that $\eta^{\star} \leq \max(\eta_{\mathcal{A}}^{\star},\eta_{\mathcal{B}}^{\star})$ for any user classification $(\mathcal{A},\mathcal{B})$.
	
	\item The merged optimal solutions of $(\mathcal{P}_{\mathcal{A}})$ and $(\mathcal{P}_{\mathcal{B}})$ is a feasible solution of $(\mathcal{P}_2)$, and its objective function is $\max(\eta_{\mathcal{A}}^{\star},\eta_{\mathcal{B}}^{\star})$. Hence, if $\eta^{\star} = \max(\eta_{\mathcal{A}}^{\star},\eta_{\mathcal{B}}^{\star})$, this feasible solution is also an optimal solution.
	
	\item  Let $f_{k}^{{\sf{f}}\star}$ and $d_{k}^\star$ denote the optimal solution of user $k \in \mathcal{B}$. It is easy to see that $\{f_{k}^{{\sf{f}}\star}, \; d_{k}^\star\}$ is also the feasible point when $k \in \mathcal{B'}$. This means that $\eta_{\mathcal{B'}}$ cannot be greater than $\eta_{\mathcal{B}}$. 
	On the other hand, to satisfy the delay constraint, we must allocate either the fog computing resource or the backhaul resource for users $k' \in \mathcal{B}\backslash \mathcal{B'}$.
	Therefore, we can allocate these resources $\{f_{k'}^{{\sf{f}}\star}, \; d_{k'}^\star\}$ for users $k \in \mathcal{B'}$ with the WEDC satisfying $\Xi_k = \eta_{\mathcal{B}}$. Let $f^{\sf{f}} = \sum_{k' \in \mathcal{B}\backslash \mathcal{B'}} f_{k'}^{{\sf{f}}\star}$, $d= \sum_{k' \in \mathcal{B}\backslash \mathcal{B'}} d_{k'}^{{\sf{f}}\star}$, $N_1 = \sum_{k \in \mathcal{B'}} \mathbbm{1}_{\Xi_{k}=\eta_{\mathcal{B}}} \mathbbm{1}_{s_k^{\sf{f}}=1}$ and $N_2 = \sum_{k \in \mathcal{B'}} \mathbbm{1}_{\Xi_{k}=\eta_{\mathcal{B}}} \mathbbm{1}_{s_k^{\sf{c}}=1}$. 
	If both $f^{\sf{f}}$ and $d$ are positive, 
	we can allocate $d_{k} = d_{k}^\star + d/N_2$ if $s_k^{\sf{c}}=1$, or $f_{k}^{{\sf{f}}} = f_{k}^{{\sf{f}}\star} + f^{\sf{f}}/N_1$ if $s_k^{\sf{c}}=0$ for these users.  As the WEDC $\Xi_k$ and the total delay $T_k$ are inversely proportional to the $f_k^{\sf{f}}$ and $d_k$, $\Xi_k$ will decrease and less than to $\eta_{\mathcal{B}}$. 
	Therefore, $\eta_{\mathcal{B'}} \leq \eta_{\mathcal{B}}$.
\end{enumerate}
\section{Proof of \textbf{Proposition~\ref{prop_PA}}}
\label{prf_prop_PA}
Since the variables of all users in $\mathcal{A}$ are independent in $(\mathcal{P}_{\mathcal{A}})$, this optimization problem can be solved by checking the feasibility condition of each variable.
Specifically, considering constraints (CA0) and (CA2), the solution for $f_k^{\sf{u}}$ is feasible if and only if 
$
\min_{f_k^{\sf{u,min}} \leq f_{k}^{\sf{u}} \leq F_k^{\sf{max}}}\mathcal{Q}_{k,0}(f_k^{\sf{u}}) \leq \eta_{\mathcal{A}}.
$

As $\mathcal{Q}_{k,0}$ is convex with respect to $f_k^{\sf{u}}$, one can easily determine the only stationary point of $\mathcal{Q}_{k,0}(f_k^{\sf{u}})$ as $f_k^{\sf{u,sta}} = \sqrt[3]{\frac{w_k^\text{\tiny T}}{2w_k^\text{\tiny E}  \alpha_k}}$ by taking the derivative of $\mathcal{Q}_{k,0}(f_k^{\sf{u}})$ with respect to $f_k^{\sf{u}}$, and setting the resulting derivative to zero.
Then, the minimum value of $\mathcal{Q}_{k,0}(f_k^{\sf{u}})$ over the range $f_k^{\sf{u,min}}\leq f_{k}^{\sf{u}} \leq F_k^{\sf{max}}$ can be determined as
in \eqref{eq20}.
Using these results for all locally executing users, the feasibility condition of problem $(\mathcal{P}_{\mathcal{A}})$ can be 
written as
$
\max_{k \in \mathcal{A}} \eta_k^{\sf{lo}} \leq \eta_{\mathcal{A}}.
$
Hence, the optimum value of $(\mathcal{P}_{\mathcal{A}})$ must be $\eta_{\mathcal{A}}^{\star}=\max_{k \in \mathcal{A}} \eta_k^{\sf{lo}}$ where the optimal solution for each user $k$ can be defined as the point in $\big[f_k^{\sf{u,min}},F_k^{\sf{max}}\big]$ corresponding to the minimum value of $\mathcal{Q}_{k,0}(f_k^{\sf{u}})$.

\section{Proof of \textbf{Theorem~\ref{thrm_opt_class}}}
\label{prf_thrm_opt_class}
Assume that $(\mathcal{A}^{\prime},\mathcal{B}^{\prime})$ is an optimal classification corresponding to the optimum value $\eta^{\star}$. Due to Statement 2 in \textbf{Lemma~\ref{lemma_1}} and \textbf{Proposition~\ref{prop_PA}}, we have the following results:
\beq
\max(\eta_{\mathcal{A}^{\prime}},\eta_{\mathcal{B}^{\prime}}) = \eta^{\star}, \\
\quad \eta_{\mathcal{A}^{\prime}} = \max_{k \in \mathcal{A}^{\prime}} \eta_k^{\sf{lo}}. \label{A_prime}
\eeq

If there is no user $k$ in $\mathcal{B}$ whose $\eta_k^{\sf{lo}}$ is less than or equal to $\eta^{\star}$, we can conclude that $(\mathcal{A}^{\prime},\mathcal{B}^{\prime}) \equiv (\mathcal{A}^{\star},\mathcal{B}^{\star})$. Then, $(\mathcal{A}^{\star},\mathcal{B}^{\star})$ must be an optimal classification.

Conversely, if there exists a user $k$ in $\mathcal{B}$ such that $\eta_k^{\sf{lo}} \leq \eta^{\star}$, we will prove that the user classification determined in \textbf{Theorem~\ref{thrm_opt_class}} is also an optimal classification.
Let $\mathcal{C}=\lbrace k \in \mathcal{B}^{\prime}\vert \eta_k^{\sf{lo}} \leq \eta^{\star} \rbrace$. Then, it is easy to see that
$
\mathcal{A}^{\star} = \mathcal{A}^{\prime} \cup \mathcal{C} \;\; \text{ and } 
\mathcal{B}^{\star} = \mathcal{B}^{\prime}/\mathcal{C}.
$
According to the definition of $\mathcal{C}$, \eqref{A_prime}, and the result in \textbf{Proposition~\ref{prop_PA}}, we have $\eta_{\mathcal{A}^{\star}} \leq \eta^{\star}$.
In addition, since $\mathcal{B}^{\star} \subset \mathcal{B}^{\prime}$, because of Statement 3 in \textbf{Lemma~\ref{lemma_1}}, we can conclude that $\eta_{\mathcal{B}^{\star}} \leq \eta_{\mathcal{B}^{\prime}} \leq \eta^{\star}$.
Using these results, we can conclude that $(\mathcal{A}^{\star},\mathcal{B}^{\star})$ is an optimal classification.

\section{ Proof of \textbf{Proposition~\ref{prop2}}} \label{prf_prop2}
Functions $\mathcal{Q}_{k,1}$ and $\mathcal{Q}_{k,2}$ are sums of exponential terms with positive coefficients; therefore, they are convex with respect to the variables in set $\tilde{\Omega}_{2,k}$ as proven in \cite{boyd2004convex}. On the other hand, the first term of the WEDC and the total delay can be represented via function $\mathcal{H}({\tilde{p}_k},y_k) = \frac{a_{k,0}  \e^{a_{k,1}{\tilde{p}_k} + a_{k,2} y_k }}{\log \big(1{\plus} \beta_{k,0} \e^{{\tilde{p}_k}}  \big)}$, where $y_k \in \{ {\tilde{\omega}_k^{\sf{u}}}, \tilde{\rho}_{k}, \tilde{l}_{k} \}$,  $a_{k,0} > 0$, $a_{k,1} = \{0,1\}$, and $\beta_{k,0}\e^{{\tilde{p}_k}} > 0$ due to the required positive data rate  when users decide to offload their computational load. 

Now, we will show that $\mathcal{H}({\tilde{p}_k},y_k)$ is a convex function of ${\tilde{p}_k}$ and $y_k$. Firstly, $\mathcal{H}({\tilde{p}_k},y_k)$ is convex with respect to $y_k$. 
Now, we need to prove that ${\partial^2 \mathcal{H}({\tilde{p}_k},y_k)}/{\partial {\tilde{p}_k}^2} \geq  0$ and the determinant $|H({\tilde{p}_k},y_k)| > 0$, where $H({\tilde{p}_k},y)$ is the Hessian matrix of $\mathcal{H}({\tilde{p}_k},y_k)$. 

Because we have $u_k =\beta_{k,0} \e^{{\tilde{p}_k}} > 0$ and the fact that $\log(1+u_k) < u_k, \forall u_k >0$, it can be verified that $|H({\tilde{p}_k},y)| {=}  \frac{a_{k,0} a_{k,2}^2 \beta_{k,0} [u_k {\minus} \log (1{\plus}u_k)] \e^{(2a_{k,1}{\plus}1) {\tilde{p}_k} {\plus} 2 a_{k,2} y }}{(1+u_k)^2 \log^4 (1+u_k )} {>} 0 $. In addition, we have
\begin{equation}
\frac{\partial^2 \mathcal{H}({\tilde{p}_k},y_k)}{\partial {\tilde{p}_k}^2} = 
\begin{cases}
\frac{u_k[2u_k - \log (1+u_k) ]}{ (1+u_k)^2 \log^3 (1+u_k)}, &  \text{if} \quad a_{k,1} = 0, \\[2pt] 
\frac{a_{k,0} \e^{a_{k,2}y} \e^{{\tilde{p}_k}} \mathcal{H}_\text{a} (u_k)}{(1+u_k)^2 \log^3 (1+u_k)},  &  \text{if} \quad a_{k,1} = 1,
\end{cases}
\label{eq25a}
\end{equation}
where $\mathcal{H}_\text{a} (u_k) = (1{\plus} u_k)^2 \log^2(1{\plus} u_k) 
{\minus} (3u_k  {\plus} 2u_k^2) \log(1{\plus}u_k )  {\plus} 2u_k^2.$
From (\ref{eq25a}), it can be verified that ${\partial^2 \mathcal{H}({\tilde{p}_k},y_k)}/{\partial {\tilde{p}_k}^2} > 0, \forall u_k > 0$ when $a_{k,1} = 0$. For the case with $a_{k,1} = 1$, since $\mathcal{H}_\text{a} (u_k) $ is a quadratic function of $\log(1+u_k)$, the discriminant of $\mathcal{H}_\text{a}(u_k) $ is $u_k^2\big[ 2- \big(1+2u_k \big)^2 \big]$, which leads to $\mathcal{H}_\text{a} (u_k) = (1{+}u_k)^2 \prod\limits_{j=\{-1,1\}} \big( \log(1{+}u_k) {-} u_{k,j}\big)$ if $u_k \leq  \frac{\sqrt{2} - 1}{2}$, where $u_{k,j} = \frac{u_k(3+2u_k) + j u_k \sqrt{2-(1+2u_k)^2}}{2(1+u_k)^2}$, $j=\{-1,1\}$. Otherwise, $\mathcal{H}_\text{a} (u_k)$ will be positive.
Using again $\log(1+u_k) < u_k, 
\forall u_k >0$, we have
$
u_{k,\{1\}} - \log(1+u_k) \geq  u_{k,\{-1\}} - \log(1+u_k) 
\geq u_{k,\{-1\}} - u_k > 0, \forall u_k > 0.
\label{eq28a}
$
This implies that $\mathcal{H}_\text{a} (u_k) >0, \forall u_k >0$, and we can conclude that ${\partial^2 \mathcal{H}({\tilde{p}_k},y_k)}/{\partial {\tilde{p}_k}^2} > 0$ as shown in (\ref{eq25a}).
As $\mathcal{H}({\tilde{p}_k},y_k)$ is a convex function, $\Xi_{k}$ and $T_k$ are also convex. Furthermore, $(\text{C}6)_k$ can be easily transformed to a linear constraint as ${\tilde{\rho}_k} + {\tilde{p}_k} \leq \log({P_k^{\sf{max}}})$, while  
$(\text{C}1)_k$, $(\text{C}5)_k$,  and $(\text{C}7)_k$  can be converted to  box constraints for ${\tilde{f}_k^{\sf{u}}}$, ${\tilde{\omega}_k^{\sf{u}}}$, and ${\tilde{\rho}_k}$, respectively. Therefore, $(\mathcal{P}_3)_k$ is a convex optimization problem with respect to $\tilde{\Omega}_{2,k} \cup {\tilde{l}_k}$.

\section{Proof of \textbf{Proposition~\ref{prop4}}} \label{prf_prop4}
We have the derivative $
{\partial \mathcal{H}_0 \big(\omega_{k}^{\sf{f}}, d_k\big) }/{\partial \omega_k^{\sf{f}}} = { \mathcal{H}_3(\omega_k^{\sf{f}}, d_k)}/{(\nu_{k,0} {\omega_{k}^{\sf{f}}} d_k {\minus} b_k^{\sf{in}})^2}, 
$
where $\mathcal{H}_3(\omega_k^{\sf{f}}, d_k) = d_k \big[ - \tilde{\gamma}_{k,1}^{\sf{co},\sf{f}} b_k^{\sf{in}} (\gamma_{k,2}^{\sf{co},\sf{f}} +1) {\big(\omega_{k}^{\sf{f}}\big)}^{\gamma_{k,2}^{\sf{co},\sf{f}}} $ $- \tilde{\gamma}_{k,3}^{\sf{co},\sf{f}} b_k^{\sf{in}} + \tilde{\gamma}_{k,1}^{\sf{co},\sf{f}}  \nu_{k,0} \gamma_{k,2}^{\sf{co},\sf{f}} d_k {\big(\omega_{k}^{\sf{f}}\big)}^{\gamma_{k,2}^{\sf{co},\sf{f}}+1}  \big]$.
As $\mathcal{H}_0 \big(\omega_{k}^{\sf{f}}, d_k\big)$ is positive when $s_k^{\sf{m}}=1$, it implies that $\nu_{k,0} \omega_k^{\sf{f}} d_k > b_k^{\sf{in}}$. Therefore, we can infer that $\mathcal{H}_3(\omega^{\sf{f}}, d_k) {\leq} -\tilde{\gamma}_{k,1}^{\sf{co},\sf{f}} \big(\omega^{\sf{f}} \big)^{\gamma_{k,2}^{\sf{co},\sf{f}}} b_k^{\sf{in}} d_k - \tilde{\gamma}_{k,3}^{\sf{co},\sf{f}} b_k^{\sf{in}} d_k {<} 0, \forall \omega^{\sf{f}}, d_k$ if $\gamma_{k,2}^{\sf{co},\sf{f}}\leq 0$. Hence,  $ \mathcal{H}_0 \big(\omega_{k}^{\sf{f}}, d_k\big)$ achieves its minimal value at $\omega_{k}^{{\sf{f}}\star} = \omega_{k}^{\sf{max},\sf{f}}$ when $\gamma_{k,2}^{\sf{co},\sf{f}}\leq 0$.
When $\gamma_{k,2}^{\sf{co},\sf{f}} > 0$, it can be verified that $
\mathcal{H}_3 \big(\omega_{k}^{{\sf{f}}\star}, d_k \big) = 0$ if and only if $ d_ k =    \mathcal{H}_1 \big(\omega_{k}^{{\sf{f}}\star}\big).
$
On the other hand, the derivative of $ \mathcal{H}_1 \big(\omega_{k}^{\sf{f}}\big)$ is
$
\frac{\partial \mathcal{H}_1 \big(\omega_{k}^{\sf{f}}\big) }{\partial \omega_k^{\sf{f}}} = 
- \frac{(\gamma_{k,2}^{\sf{co},\sf{f}}+1)}{\gamma_{k,2}^{\sf{co},\sf{f}}} \frac{ b_k^{\sf{in}} (\tilde{\gamma}_{k,1}^{\sf{co},\sf{f}} {\big(\omega_{k}^{\sf{f}}\big)}^{\gamma_{k,2}^{\sf{co},\sf{f}}} + \tilde{\gamma}_{k,3}^{\sf{co},\sf{f}} ))}{\nu_{k,0} {\big(\omega_{k}^{\sf{f}}\big)}^{\gamma_{k,2}^{\sf{co},\sf{f}}+2}} < 0. 
$
So, $\mathcal{H}_1 \big(\omega_{k}^{\sf{f}}\big)$ is a monotonically decreasing function with respect to $\omega_{k}^{\sf{f}}$. Therefore, $\mathcal{H}_0 \big(\omega_{k}^{\sf{f}}, d_k\big)$ is minimized if $\omega_{k}^{\sf{f}} = \omega_{k}^{{\sf{f}}\star}$ satisfies (\ref{eq54}).


%
\section{Proof of \textbf{Lemma~\ref{lemma_4}}} \label{prf_lemma_4}
First, it can be verified that $
\frac{\partial \mathcal{H}_0\big(\omega_{k}^{{\sf{f}}}, d_k\big)  }{\partial d_k} {=} {-} \frac{b_k^{\sf{in}}  \omega_{k}^{\sf{f}}     \big[ \tilde{\gamma}_{k,1}^{\sf{co},\sf{f}} \big(\omega_{k}^{\sf{f}} \big)  ^{\gamma_{k,2}^{\sf{co},\sf{f}}} {+} \tilde{\gamma}_{k,3}^{\sf{co},\sf{f}} \big] }{\big(\nu_{k,0} \big(\omega_{k}^{\sf{f}} \big)  d_{k} {-} b_k^{\sf{in}} \big)^2}{=}
\mathcal{H}_2 \big(\omega_{k}^{\sf{f}}, d_k \big).
$
As $\frac{\partial \mathcal{H}_1 \big(\omega_{k}^{\sf{f}}\big) }{\partial \omega_k^{\sf{f}}} < 0$ for all $\omega_k^{\sf{f}}$, $\omega_{k}^{{\sf{f}}\star}$ will not increase when $d_k > \bar{d}_{k,1}$ increases. When $\gamma_{k,2}^{\sf{co},\sf{f}} \leq 0$, $\omega_{k}^{{\sf{f}}\star} = \omega_{k}^{\sf{max},\sf{f}}$ as proved in \textbf{Proposition~\ref{prop4}}. Therefore,  $\mathcal{H}_2 \big( \omega_k^{\sf{max},\sf{f}}, d_k \big)$ increases  with respect to $d_k$. When $\gamma_{k,2}^{\sf{co},\sf{f}} > 0$, we will show that
$
{\mathcal{H}_2\big(\omega_{k,1}^{{\sf{f}}\star}, d_k\big)  }\big|_{d_k = d_{k,1}} 
< {\mathcal{H}_2\big(\omega_{k,2}^{{\sf{f}}\star}, d_k\big)  }\big|_{d_k = d_{k,2}}
$, where  $\bar{d}_{k,1} < d_{k,1} < d_{k,2}$ and $\omega_{k,i}^{{\sf{f}} \star}$ denotes the optimal value of $\omega_k^{\sf{f}}$  when $d_k$ is equal to $d_{k,i}$, for $i=1,2$. 

Indeed, when $\omega_{k}^{{\sf{f}}}$ is fixed, 
$
{\mathcal{H}_2\big(\omega_{k}^{{\sf{f}}}, d_k\big)  }$ is an increasing function of $d_k$. The second derivative of  $\mathcal{H}_0\big(\omega_{k}^{{\sf{f}}}, d_k\big) $ when substituting $d_k {=} \mathcal{H}_1\big(\omega_{k}^{{\sf{f}}}\big) $ is given as
$
\frac{\partial {\mathcal{H}_2\big(\omega_{k}^{{\sf{f}}}, d_k\big)  } }{\partial \omega_{k}^{{\sf{f}}}} =  -\frac{\mathcal{H}_4\big( \omega_{k}^{{\sf{f}}}\big)}{\big( \tilde{\gamma}_{k,1}^{\sf{co},\sf{f}} \big(\omega_{k}^{{\sf{f}}} \big)  ^{\gamma_{k,2}^{\sf{co},\sf{f}}} + \tilde{\gamma}_{k,3}^{\sf{co},\sf{f}} \big)^2},
$
where $\mathcal{H}_4\big( \omega_{k}^{{\sf{f}}}\big) {=} (\tilde{\gamma}_{k,1}^{\sf{co},\sf{f}})^2 \big({\gamma_{k,2}^{\sf{co},\sf{f}}}\big)^2 \big(\omega_{k}^{{\sf{f}}} \big)^{2 {\gamma_{k,2}^{\sf{co},\sf{f}}}} \big( \tilde{\gamma}_{k,1}^{\sf{co},\sf{f}} \big( {\gamma_{k,2}^{\sf{co},\sf{f}}} {+}1\big) \big(\omega_{k}^{{\sf{f}}} \big)^{{\gamma_{k,2}^{\sf{co},\sf{f}}}} {+} \tilde{\gamma}_{k,3}^{\sf{co},\sf{f}} \big(2{\gamma_{k,2}^{\sf{co},\sf{f}}} {+}1\big) \big) {>} 0$, for all $\omega_{k}^{{\sf{f}}}$ when $\gamma_{k,2}^{\sf{co},\sf{f}} {>} 0$. 
Thus, it can be concluded that ${\mathcal{H}_2\big(\omega_{k}^{{\sf{f}}}, d_k\big)  }$ is a decreasing function of  $\omega_{k}^{{\sf{f}}}$. Furthermore,  the optimal solution $\omega_{k}^{{\sf{f}}\star}$ monotonically decreases as $d_k$ increases as shown in  (\ref{eq54}); hence,  $\omega_{k,1}^{{\sf{f}}\star} {\geq} \omega_{k,2}^{{\sf{f}}\star}$. Therefore, we have
$
{\mathcal{H}_2\big(\omega_{k,1}^{{\sf{f}}\star}, d_k\big)  }\big|_{d_k = d_{k,1}} {\leq} {\mathcal{H}_2\big(\omega_{k,2}^{{\sf{f}}\star}, d_k\big)  }\big|_{d_k = d_{k,1}} 
{<} {\mathcal{H}_2\big(\omega_{k,2}^{{\sf{f}}\star}, d_k\big)  }\big|_{d_k = d_{k,2}}.
$

\section{Verification of data compression model}
To the best of our knowledge, in the literature, there is no theoretical result regarding a mathematical model for the computational workload required
for data compression process, i.e., a function expressing the relationship between the computation workload of
the compression process and the compression ratio does not seem to be known. This justifies our selection of a general
non-linear function to represent the relationship between the computation load and the compression ratio. The parameters of the proposed function are obtained 
by minimizing the difference between the experimental and the modeled data using a fitting approach.
In fact, this fitting method based on experimental data has been employed to model different metrics in the wireless communication literature including the  energy per operation \cite{Zhang13} and the channel modeling \cite{sun2016air}. 
In particular, \cite{Zhang13} proposed an energy model which can fit well with the measured data in \cite{miettinen2010energy}. 
The authors in \cite{sun2016air} proposed the air-to-ground channel models which fit well with their collected data.

Because there is no theoretical result available regarding this modeling issue, we have employed a practical data-fitting approach to 
capture the compression computational load, decompression computational load, and compression quality as non-linear
functions of the compression ratio.
In general, when different models based on data fitting  are chosen to optimize certain
system parameters then they can affect the system performance differently.
In addition, a more accurate fitting model can result in a more reliable design.

In our current work, we first collected the experimental data by running different compression algorithms, namely GZIP, BZ2, and JPEG, and measured on the execution time of the corresponding algorithms for different values of the compression ratio.
Specifically, we run algorithms GZIP, BZ2, and JPEG for data compression
in Python 3.0 via a Linux terminal using Ubuntu 18.04.1 LTS on a computer equipped with CPU chipset Intel(R) core(TM) i7-4790, and 12 GB RAM.
In these experiments, we keep the CPU clock speed unchanged by employing \textit{``Linux cpupower tool''} and turning off all applications except the ``Linux terminal'' program executing the compression and decompression algorithms. 
The experimental data were obtained by averaging over 1000 realizations.
Since the computational load in CPU cycles is linearly proportional to the execution time for a fixed CPU clock speed, the normalized execution time is proportional to the normalized computational load.
It is worth noting that the same approach for collecting experimental data of the normalized computational load was employed
in \cite{Clearlinux},  in which different compression algorithms, e.g., XZ and ZLIB, were run on computers using two different operating systems: Ubuntu and Clear linux (CL).
After collecting the experimental data, we used them to fit the parameters of the functions shown in equation (1) in the manuscript 
by using the Matlab curve fitting tool \cite{matlabfit}.

To validate the quality of the proposed fitting model, we compare the normalized execution time obtained by our proposed model fitted by
using the experimental data and the experimental data given in \cite{Clearlinux}.
The root mean square error is computed as 
\beq
\text{RMSE} = \sqrt{\frac{\sum_{i=1}^{N} (\hat{y}_i-y_i)^2}{N}},
\eeq
where $y = [y_1, y_2,...,y_N]$ are our collected experimental data or the experimental data obtained in \cite{Clearlinux}, and $\hat{y}=[\hat{y}_1, \hat{y}_2,...,\hat{y}_N]$ are the values obtained from the fitted function.
Additionally, two other existing models in the literature, namely the linear model adopted in \cite{zhang2017mobile, yu2006information} and the non-linear model proposed in \cite{li2019wirelessly} are compared with our proposed model in the additional numerical studies
shown in Figs.~\ref{fig1res}, \ref{fig2res}, \ref{fig3res}, \ref{fig4res}, and \ref{fig5res}. These additional results allow us
to verify the quality of the proposed functions.
Particularly, we compare the proposed model introduced in Section II.A, $C_{\sf our} = \gamma_1 \omega^{\gamma_2} + \gamma_3$, with the linear model \cite{yu2006information}, $C_{\sl li} =  \beta_1 \omega + \beta_2$, and the non-linear model  \cite{li2019wirelessly}, $C_{\sf nl} = \epsilon_1 (\exp (\epsilon_2 \omega) - \exp(\epsilon_2))$, where $\epsilon_1, \epsilon_2, \gamma_1, \gamma_3$ are positive constants and $\omega$ is the compression ratio. The parameters $\epsilon_1, \epsilon_2, \gamma_1, \gamma_3, \beta_1$, and $\beta_2$ used for Figs.~\ref{fig1res}, \ref{fig2res}, \ref{fig3res}, \ref{fig4res}, and \ref{fig5res} are given in Table~\ref{t:fitcurve}.

These figures confirm that the proposed model is superior to the existing models.
Specifically, Figs.~\ref{fig1res} and \ref{fig2res} illustrate the fitted curves using our experimental data when running compression algorithms  GZIP and BZ2 for the benchmark text files ``alice.txt" and ``asyoulik.txt" while Figs.~\ref{fig3res} and \ref{fig4res} show the fitted curves using the experimental data obtained in \cite{Clearlinux}. Fig.~\ref{fig5res} shows the root mean square error (RMSE) of the fitted curves and the true data (our experimental data and the experimental data in \cite{Clearlinux}). The results obtained with
our proposed model and existing models are presented in these figures for the comparison purpose.

\begin{table}[t]
	\caption{Fitting Curves From Different Models}
	\centering
	\label{t:fitcurve}
	\begin{tabular}{|l|l|l|l|}
		\hline
		Schemes & Linear Model& Model [7]& Our Proposed Model\\
		\hline
		GZIP-Alice & $y {=}1.7189x{-}3.9237$ & $y{=}1.156{\times} 10^{-5}(\exp(3.969x){-}\exp(3.969)) $& $y{=}1.207{\times} 10^{-15} x^{32.28} {+} 0.3$ \\
		\hline
		GZIP-Asyoulik & $y {=} 2.1605x-4.6807$ & $y{=}2.146{\times} 10^{-6}(\exp(4.964x){-}\exp(4.964)) $& $y{=}6.497{\times} 10^{-19} x^{42.94} {+} 0.303$ \\
		\hline
		BZ2-Alice & $y {=} 0.0297x+0.6657$ & $y{=}1029(\exp(0.00012x)-\exp(0.00012)) $& $y{=}0.076 x^{0.7117} {+} 0.579$ \\
		\hline
		BZ2-Asyoulik &$y {=} 0.0297x-0.6661$ & $y{=}1450 (\exp(8.47{\times}10^{-5}x){-}\exp(8.47{\times}10^{-5})) $& $y{=}0.178 x^{0.478} {+} 0.437$ \\
		\hline
		XZ-Ubuntu & $y {=}0.4767x{-}2.5884$ & $y{=}4.703{\times} 10^{-4}(\exp(1.018x){-}\exp(1.018)) $& $y{=}6.441{\times} 10^{-7} x^{7.062} + 10^{-7}$ \\
		\hline
		XZ\_CL & $y {=}0.4767x{-}2.5884$ & $y{=}7.255{\times} 10^{-4}(\exp(0.961x){-}\exp(0.961)) $& $y{=}1.492{\times} 10^{-6} x^{6.646} + 10^{-6}$ \\
		\hline
		ZLIB-Ubuntu & $y {=}1.032x{-}4.0726$ & $y{=}2.002{\times} 10^{-20}(\exp(9.228x){-}\exp(9.228)) $& $y{=}6.019{\times} 10^{-76} x^{108.6} {+} 0.240$ \\
		\hline
		ZLIB\_CL & $y {=}1.177x{-}4.7851$ & $y{=}1.1557{\times} 10^{-5}(\exp(9.769x){-}\exp(9.769)) $& $y{=}1.436{\times} 10^{-68} x^{97.76} {+} 0.205$ \\
		\hline
	\end{tabular}
\end{table}

\begin{figure*}[h!]
	\centering
	\includegraphics[height=3.9in]{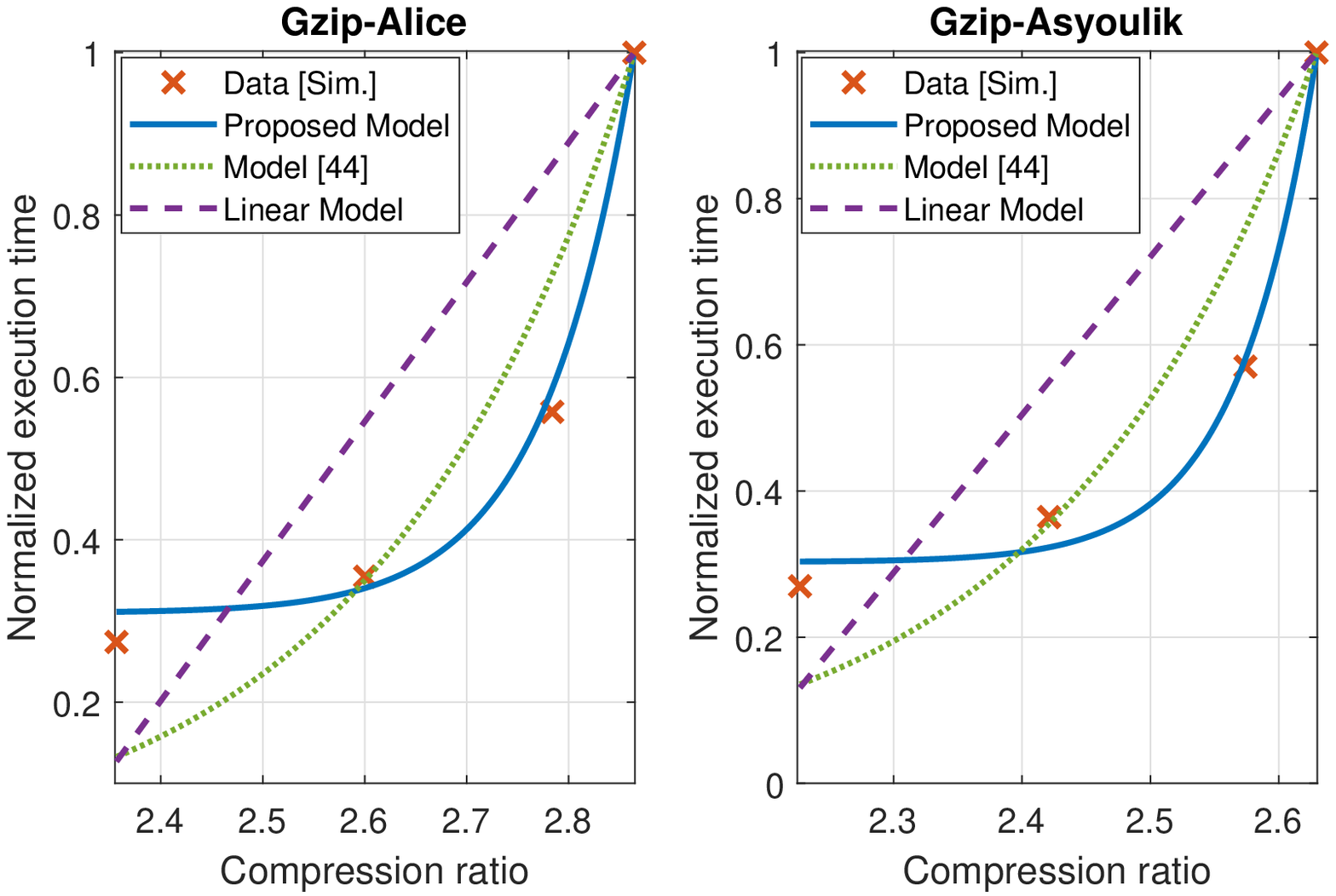}
	\caption{Comparison of different fitting models for the GZIP algorithm.}
	\label{fig1res}
\end{figure*}

\begin{figure*}[h!]
	\centering
	\includegraphics[height=3.9in]{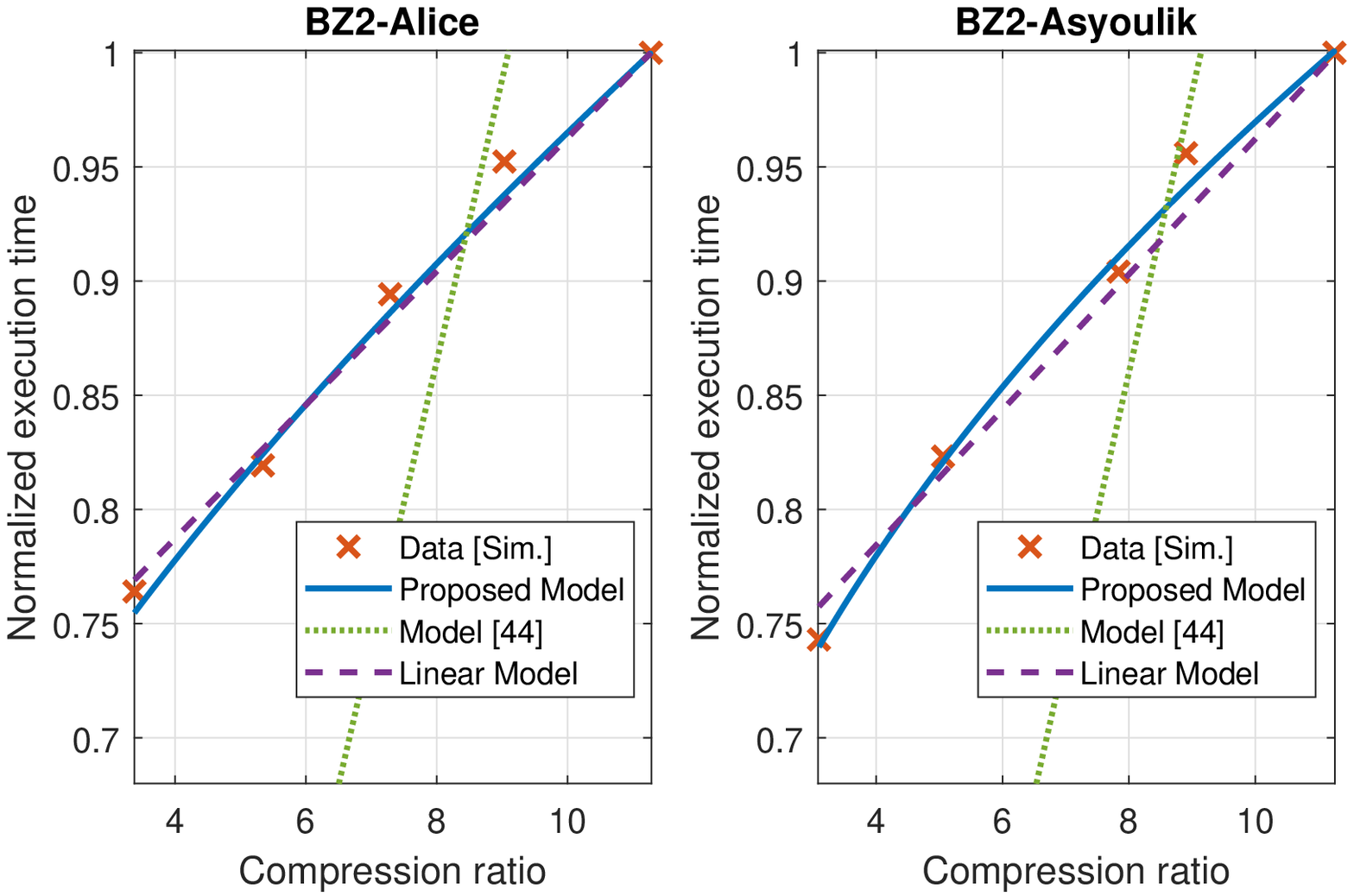}
	\caption{Comparison of different fitting models for the BZ2 algorithm.}
	\label{fig2res}
\end{figure*}

\begin{figure*}[h!]
	\centering
	\includegraphics[height=3.9in]{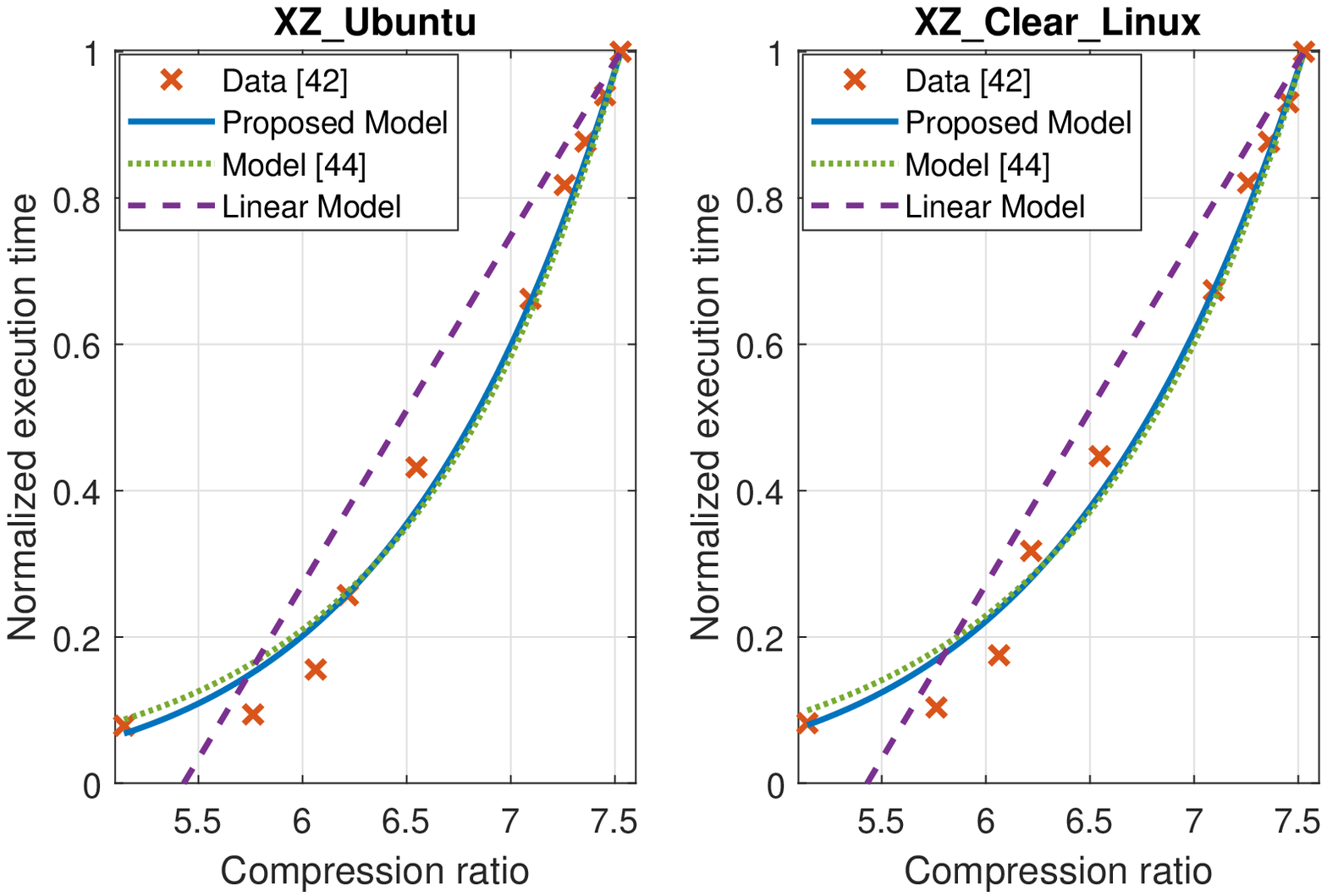}
	\caption{Comparison of different fitting models for the XZ algorithm.}
	\label{fig3res}
\end{figure*}

\begin{figure*}[h!]
	\centering
	\includegraphics[height=3.9in]{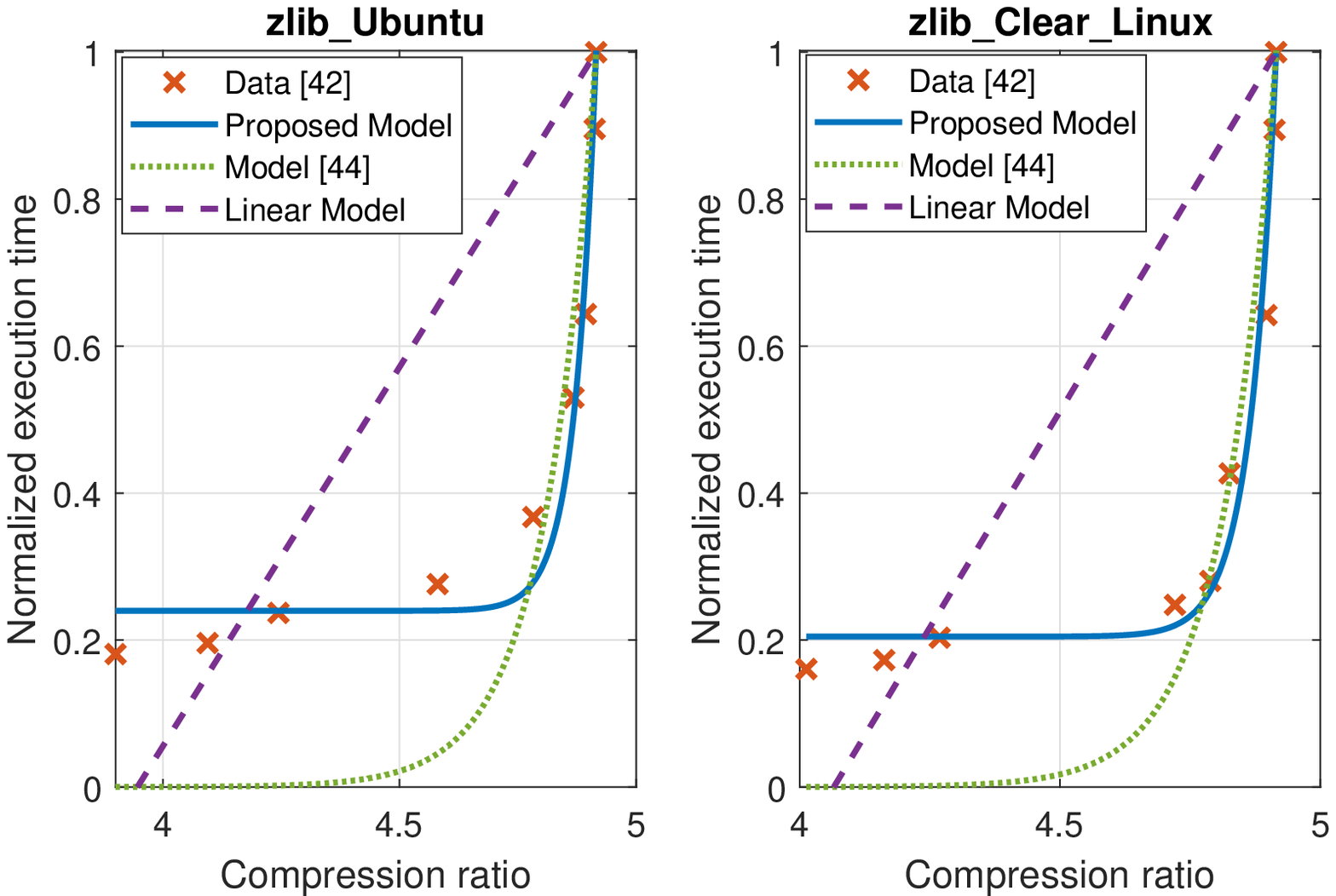}
	\caption{Comparison of different fitting models for the Zlib algorithm.}
	\label{fig4res}
\end{figure*}

\vspace{3 cm}
\begin{figure*}[h!]
	\centering
	\includegraphics[height=3.4in]{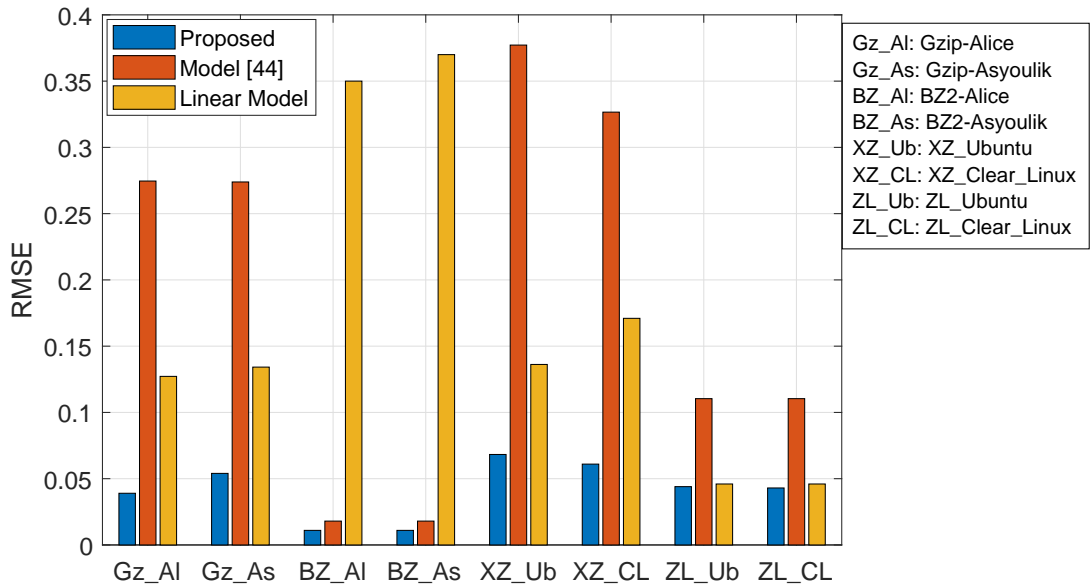}
	\caption{RMSE of different fitting models and different compression algorithms.}
	\label{fig5res}
\end{figure*}

\footnotesize{
\bibliographystyle{IEEEtran}  
\bibliography{MCCpaper}
}

\setlength{\textfloatsep}{5 pt}
\begin{table}[htb]
	\caption{List of Key Notations}
	\centering
	\label{t:observed_psrs1}
	\begin{tabular}{ll}
		\noalign{\smallskip}  \hline \noalign{\smallskip}
		Notations & Description \\
		\hline
		$K/\mathcal{K}$ & Number/set of users \\
		$k$ & User index\\
		$T_k^{\sf max}$ & Application execution interval  (seconds) \\
		$c_k$ & Computation demand of user $k$ (CPU cycles)\\
		$c_{k,0}$ & Number of CPU cycles which must be executed locally at the mobile user\\		
		$c_{k,1}$ & Number of flexible CPU cycles of user $k$ which can be processed at the user or fog/cloud server\\
		$b_k^{\sf{in}}$ & Incurred data size of user $k$ (bits) \\		
		$b_k^{\sf{out,u/f}}$ & Output data size of user $k$ after compression mobile user $k$ or at the fog (bits) \\
		$c_k^{\sf{co,u}}/ c_k^{\sf{de,u}}$ & Compression/Decompression computation load at mobile user $k$ (CPU cycles) \hspace{1 cm} \\
		$c_k^{\sf{co,f}}/ c_k^{\sf{de,f}}$ & Compression/Decompression computation load at the fog server (CPU cycles) \\
		$\gamma_{k,0}^{\sf{u}}$ & Maximum required computation load to compress the input data at mobile user $k$ (CPU cycles) \\
		$\gamma_{k,0}^{\sf{f}}$ & Maximum required computation load to compress the input data at the fog server (CPU cycles) \\
		$\gamma_{k,i}^{\sf{co/de/qu,u/f}}$ & Constant parameters which characterize the data compression model\\
		$c_k^{\sf{u}}$ & Total computation load at mobile user $k$ (CPU cycles)\\
		$c_k^{\sf{f}}$ & Total computation load for user $k$ at the fog (CPU cycles)\\
		$\xi_{1,k}^{\sf{u}}$ & Local computation energy consumed by user $k$ (Joule)\\
		$\xi_{2,k}^{\sf{u}}$ & Transmission energy of user $k$ (Joule)\\
		$\xi_{k}$ & Total energy consumption of user $k$ (Joule)\\		
		$t_{1,k}^{\sf{u}}$ & Local computation time of user $k$  (seconds) \\
		$t_{2,k}^{\sf{u}}$ & Transmission time from user $k$ to the fog (seconds) \\
		$t_{1,k}^{\sf{f}}$ & Execution time in the fog for processing the task of user $k$ (seconds) \\
		$t_{2,k}^{\sf{f}}$ & Transmission time from the fog to the cloud due to user $k$ (seconds) \\
		$T^{\sf{c}}$ & Execution time in the cloud for each user (seconds) \\
		$T_k$ & Total delay for completing the computation task of user $k$ (seconds) \\
		$w_k^{\sf{T}}/w_k^{\sf{E}}$ & The weights correspond to the service latency and energy, respectively \\
		$p_{k,0}$ & Circuit power per Hz of user $k$ (Watts/Hz) \\
		$r_k$ & Transmission rate of user $k$  (bits/seconds)\\		
		$\rho_k^{\sf{max}}$ & Maximum transmission bandwidth assigned to user $k$ (Hz)\\
		$\beta_{k,0}$ & Coefficient in signal-to-noise ratio  \\	
		$\mathcal{A}, \mathcal{B}$ & Locally executing user set and the offloading user set \\
		$\alpha_k$ & Energy coefficient specified in the CPU model\\
		\noalign{\smallskip} \hline \noalign{\smallskip}
	\end{tabular}
\end{table}
\setlength{\textfloatsep}{5 pt}
\begin{table}[htb]
	\caption{List of Key Notations (Continued)} 
	\centering
	\begin{tabular}{ll}
		\noalign{\smallskip}  \hline \noalign{\smallskip}
		Notations & Description \\
		\hline
		$\Xi_k$ & Weighted sum of energy and delay (WEDC) for user $k$  \\
		$\omega_k^{\sf{u}}$ & Compression ratio for user $k$ at mobile user $k$ \\
		$\omega_k^{\sf{f}}$ & Compression ratio for user $k$ at the fog\\
		$s_{k}^{\sf{u}}$ & Binary variable which indicates whether the user task is processed at the mobile user or not\\
		$s_{k}^{\sf{f}}$ & Binary variable which indicates whether the user task is processed at the fog or not\\
		$s_{k}^{\sf{c}}$ & Binary variable which indicates whether the user task is processed at the cloud or not\\
		$s_{k}^{\sf{m}}$ & Binary variable which indicates whether the user task is  processed at the cloud with re-compression or not\\
		$f_{k}^{\sf{u}}$ & Local CPU clock speed of user $k$  (Hz or CPU cycles/second) \\
		$f_{k}^{\sf{f}}$ & CPU clock speed assigned for processing the application of user $k$ in the fog  (Hz or CPU cycles/second) \\
		$p_k$ & Transmit power per Hz of user $k$ (Watts/Hz) \\
		$\rho_k$ & Transmission bandwidth of user $k$ (Hz)\\
		$d_{k}$ & Backhaul rate allocated for user $k$   (bits/second) \\
		$\eta$ & Auxiliary variable to capture the min-max WEDC \\		
		$\eta_{k}^{\sf{lo}}$ & Required energy for local execution (Joule)\\
		$T_k^{\sf{max}}$ & Maximum delay of user $k$\\
		$D^{\sf{max}}$ & Backhaul capacity (bits/second) \\
		$F^{\sf{f,max}}$ & Maximum CPU clock speed in the fog  (Hz or CPU cycles/second) \\
		$F_k^{\sf{max}}$ & Maximum CPU clock speed of user $k$ (Hz or CPU cycles/second) \\
		$\omega_k^{\sf{u/f,min/max}}$ & Feasible range of $\omega_k^{\sf{u/f}}$\\
		$f_k^{\sf{f,rq}}$ & Minimum required fog computing resource for executing the application of user $k$\\
		$d_k^{\sf{rq}}$ & Minimum backhaul resource required by  user $k$\\
		$G_{\mathcal{B},\eta}^{\star}$ & Minimum total required computing resource in the fog $\left(\sum_k f_k^{\sf{f}} \right)$ for a given value of $\eta$\\
		$G_{\mathcal{B},\eta}^{\sf{PLA}\star}$ & Minimum $\left(\sum_k f_k^{\sf{f}} \right)$ for a given value of $\eta$ 
		when applying PLA based algorithm\\
		$G_{\mathcal{B},\eta}^{\sf{OSTS}\star}$ & Minimum $\left(\sum_k f_k^{\sf{f}} \right)$ for a given value of $\eta$ 
		when applying the OSTS algorithm\\
		$(\mathcal{P}_1)$ & Original optimization problem\\
		$(\mathcal{P}_2)$ & Equivalent problem of problem $\mathcal{P}_1$ \\
		$(\mathcal{P}_{\mathcal{A}})$ & Sub-problem for users whose applications are processed locally \\
		$(\mathcal{P}_{\mathcal{B}})$ & Sub-problem for users whose applications are offloaded and processed in the fog/cloud \\
		$(\mathcal{P}_3)_k$ & Sub-problem to find the minimum required fog computing resource allocated to user $k$ when $s_k^{\sf{f}}=1$ \\
		$(\mathcal{P}_4)_k$ & Sub-problem to find the minimum required backhaul rate allocated to user $k$ when $s_k^{\sf{c}}=1$ \\
		$(\mathcal{P}_{\sf FV,\eta})$ & Sub-problem to find the minimum required computing resource in the fog \\
		$(\mathcal{P}_{d_k})$ & Sub-problem to find the minimum  $f_k^{\sf{f}}$ for a given $d_k$ when $s_k^{\sf{c}}=1$ and $x_k=1$\\
		$(\mathcal{P}_{\sf FV,\eta}^{\sf PLA})$ & Sub-problem to find the minimum required computing resource of the fog using PLA based algorithm\\
		$(\mathcal{P}_{\sf FV,\eta}^{\sf TSA})$ & Sub-problem to find the minimum required computing resource of the fog using TSA algorithms\\
		$(\mathcal{P}_{\sf FV,\eta}^{\sf OSTS})_{\lambda}$ & Sub-problem to solve $(\mathcal{P}_{\sf FV,\eta}^{\sf TSA})$ for a given $\lambda$  \\
		$\Omega_{1,k}$ & $\{s_k^{\sf{u}}, s_k^{\sf{f}}, s_k^{\sf{c}},  \omega_k^{\sf{u}}, f_k^{\sf{u}}, f_k^{\sf{f}}, p_k, \rho_k, d_k\}$  \\
		$\Omega_1$ & $\cup_{k\in \mathcal{K}}\Omega_{1,k}$  \\	
		$\Omega_{2,k}$ & $\{ \omega_k^{\sf{u}}, f_k^{\sf{u}}, f_k^{\sf{f}},  p_k, \rho_k  \}$  \\
		$\Omega_3$ & $\cup_{k\in \mathcal{B}}s_k^{\sf{c}}$  \\
		$\Omega_4$ & $\cup_{k\in \mathcal{B}}\{s_k^{\sf{f}}, s_k^{\sf{c}}, s_k^{\sf{m}}, d_k, f_k^{\sf{f}}, \omega_k^{\sf{f}}\}$  \\
		\noalign{\smallskip} \hline \noalign{\smallskip}
	\end{tabular}
\end{table}

\setlength{\textfloatsep}{5 pt}
\begin{table}[htb]
	\caption{List of Abbreviations}
	\centering
	\begin{tabular}{ll}
		\noalign{\smallskip}  \hline \noalign{\smallskip}
		Abbr. & Description \\
		\hline
		WEDC & Weighted energy and delay cost  \\
		JCORA & Joint compression, computation offloading, and resource allocation  \\
		PLA & Piece-wise linear approximation  \\
		OSTS & One-dimension search-based two-stage  \\
		IUTS & Iterative update-based two-stage  \\
		CPU & Central processing unit  \\
		BS & Base station  \\
		QoS & Quality of service  \\
		MIMO & Multiple-input multiple-output  \\
		ILP & Integer linear programming \\
		MILP & Mixed integer linear programming \\
		MINLP & Mixed integer non-linear programming \\
		MCC/MEC & Mobile cloud/edge computing\\
		\textbf{\textit{Mode 1}} & $s_k^{\sf f} = 1$ for user $k$\\
		\textbf{\textit{Mode 2}} & $s_k^{\sf c} = 1$ for user $k$\\
		\textbf{\textit{Mode 3}} & $s_k^{\sf m} = 1$ for user $k$\\
		\noalign{\smallskip} \hline \noalign{\smallskip}
	\end{tabular}
\end{table}
\end{document}